\documentclass[12pt]{amsart}
\usepackage{graphicx}
\usepackage{fourier}
\usepackage{graphicx}
\usepackage{epsfig}
\usepackage{color}

\newtheorem{thm}{Theorem}[section]
\newtheorem{prop}[thm]{Proposition}
\newtheorem{defn}[thm]{Definition}
\newtheorem{lemma}[thm]{Lemma}

\newtheorem{remark}[thm]{Remark}

\usepackage{amsmath}
\usepackage{amsxtra}
\usepackage{amscd}
\usepackage{amsthm}
\usepackage{amsfonts}
\usepackage{amssymb}
\usepackage{eucal}
\newcommand{\bmb}{\left( \begin{array}{rr}}
\newcommand{\enm}{\end{array}\right)}

\newcommand{\Z}{{\mathbb Z}}

\newcommand{\al}{{\alpha}}

\newcommand{\sgn}{{\rm sgn}}
\newcommand{\qq}{{\mathfrak{q}}}
\newcommand{\xx}{{\mathbf{x}}}

\numberwithin{equation}{section}

\begin{document}

\title[The Arctic curve for Aztec rectangles with defects]{The Arctic curve for Aztec rectangles with defects \\via the Tangent Method}
\author{Philippe Di Francesco} 
\address{
Department of Mathematics, University of Illinois, Urbana, IL 61821, U.S.A. 
and 
Institut de physique th\'eorique, Universit\'e Paris Saclay, 
CEA, CNRS, F-91191 Gif-sur-Yvette, FRANCE\hfill
\break  e-mail: philippe@illinois.edu
}
\author{Emmanuel Guitter}
\address{
Institut de physique th\'eorique, Universit\'e Paris Saclay, 
CEA, CNRS, F-91191 Gif-sur-Yvette, FRANCE.
\break  e-mail: emmanuel.guitter@ipht.fr
}

\begin{abstract}
The Tangent Method of Colomo and Sportiello is applied to the study of the asymptotics of domino tilings of
large Aztec rectangles, with some fixed distribution of defects along a boundary. The associated Non-Intersecting Lattice Path configurations are made of Schr\"oder paths whose weights involve two parameters $\gamma$ and $q$
keeping track respectively of one particular type of step and of the area below the paths. We derive the 
arctic curve for an arbitrary distribution of defects, and illustrate our result with a number of examples
involving different classes of boundary defects.
\end{abstract}

\maketitle
\date{\today}
\tableofcontents

\section{Introduction}
\label{sec:introduction}
Two-dimensional tiling problems of large scaled domains of the plane are known to exhibit an arctic phenomenon,
namely the existence of a sharply defined separation between frozen phases with regular lattice-like tiling configurations
and liquid phases with disorder.  There is an abundant literature on derivation of such arctic curves, starting with
the celebrated arctic circle of \cite{CLP,JPS} in the case of the domino tiling of a large Aztec diamond (see also \cite{KOS,KO1,KO2} for systematic studies of dimer models using the technology of the Kasteleyn matrix).

Recently a novel approach to determining the arctic curve was devised by Colomo and Sportiello \cite{COSPO}.
It is based on the reformulation of the tiling problems in terms of Non-Intersecting Lattice Paths (NILP). Indeed,  tiling
configurations can be bijectively described as configurations of paths with fixed ends, taking their steps on a regular 
lattice, and such that {\it no two paths share a common vertex}. The method, called the Tangent Method, consists
in slightly modifying the path configuration to use one path as a probe of the domain occupied by the others. By solving
a simple extremization problem, the method constructs a family of curves (the ``escape trajectories" of the probe)
which is tangent to the arctic curve. The latter is then recovered as the envelope of the former.
Despite the fact that it is non rigorous, this method has already provided new insights, first by reproducing known
results (see \cite{COSPO,DFLAP,DFGUI,DFG2} for many examples), and moreover by allowing for new conjectures
such as the applicability to interacting NILP like those involved in the 6 Vertex model, for which an arctic curve was
proposed \cite{COSPO,DFLAP,CPS}.

Our aim in this paper is to extend the former results of \cite{DFLAP,DFGUI,DFG2} to the case of
Aztec rectangles with \emph{an arbitrary distribution of defects} along one boundary. This involves NILP configurations
in which paths can take three types of steps. We consider here a weighting of the configurations with two parameters $q$ and $\gamma$, 
that, in the NILP language, keep track respectively of the area below each path and of the number of steps of one particular type.
In the tiling language, the weight $\gamma$ singularizes one type of tile, while the weight $q$ measures the $3D-$\emph{volume} below a landscape
of which the paths are the contour lines, thus generalizing a well-known interpretation of simpler NILP associated with rhombus tilings (with only two types of path steps) 
in terms of plane partitions.

An important effect of the weight $q$ is to bend the most likely trajectory of a free path with fixed ends, and we shall study the precise structure of these bent \emph{geodesics} in our setting. The escape trajectory of the probe used in the Tangent Method will be made of such a bent geodesic, thus
leading to a family of geodesic curves tangent to the arctic curve. Note that these geodesics become straight lines 
only in the limit $q\to 1$.

The tiled domain that we consider is an Aztec rectangle of size $n\times m$, with a number of
boundary defects characterized, in the NILP language, by a fixed set of path starting points, with positions
$0=a_0<a_1<\cdots <a_n=m$ along the lower boundary the Aztec rectangle. In the scaling limit, we assume that 
$a_i\sim n \al(i/n)$ for some continuous, piecewise differentiable function 
$\al(\sigma)$, $\sigma \in [0,1]$. Given such a distribution, the Tangent Method produces parametric equations for the arctic
curve in the limit of large $n$, which are functions of $\al(\sigma)$, of the rescaled variable $\qq=q^{1/n}$ and of $\gamma$ (see Theorem \ref{thmac}). 

The paper is organized as follows: in Section~\ref{sec:aztec}, we define the domino tiling problem of Aztec rectangles with defects, and its weighted NILP version. The paths involved are Schr\"oder paths, which are further studied in Section~\ref{sec:wls}, where we first derive a summation formula for the partition function of a single weighted path with fixed ends, and then use its asymptotics to derive geodesics in large size (see Theorem \ref{thm:geod}).
Using the Lindstr\"om-Gessel-Viennot formula, we compute in Section \ref{sec:arcticone} the partition function 
of the weighted NILP
problem, and in preparation for the Tangent Method we also compute the ``one-point function" corresponding to
one escaping path (the probe). Both computations use the explicit LU decomposition of the Gessel-Viennot matrix
and result in Theorems \ref{partitionfunction} (partition function) and \ref{onepointfunctionthm} (one-point function). We also derive the large size 
scaling asymptotics of these formulas and show in particular how to relate the position of the most likely ``exit point" where the probe leaves the Aztec rectangle
to the position of its assigned target (Theorem~\ref{thm:likely}).
In Section \ref{sec:arctictwo}, we complete the application of the Tangent Method and derive the family of tangent
geodesics (Theorem \ref{thm:geofamily}) and finally the corresponding arctic curve (Theorem \ref{thmac}).
We also discuss the effect of various types of boundaries according to properties of the distribution $\al(\sigma)$.
Section \ref{sec:examples} is devoted to a number of examples, each illustrating a particular type of boundary,
which we classify between generic, freezing and fully frozen.
We gather a few concluding remarks in Section~\ref{sec:conclusion}.

\medskip

\noindent{\bf Acknowledgments.}  

\noindent  

PDF is partially supported by the Morris and Gertrude Fine endowment and the NSF grant DMS18-02044. EG acknowledges the support of the grant ANR-14-CE25-0014 (ANR GRAAL).

\section{Weighted Domino tilings of an Aztec rectangle with boundary defects}
\label{sec:aztec}

\subsection{Definition of the model}
\label{sec:defaztec}

\begin{figure}
\begin{center}
\includegraphics[width=10cm]{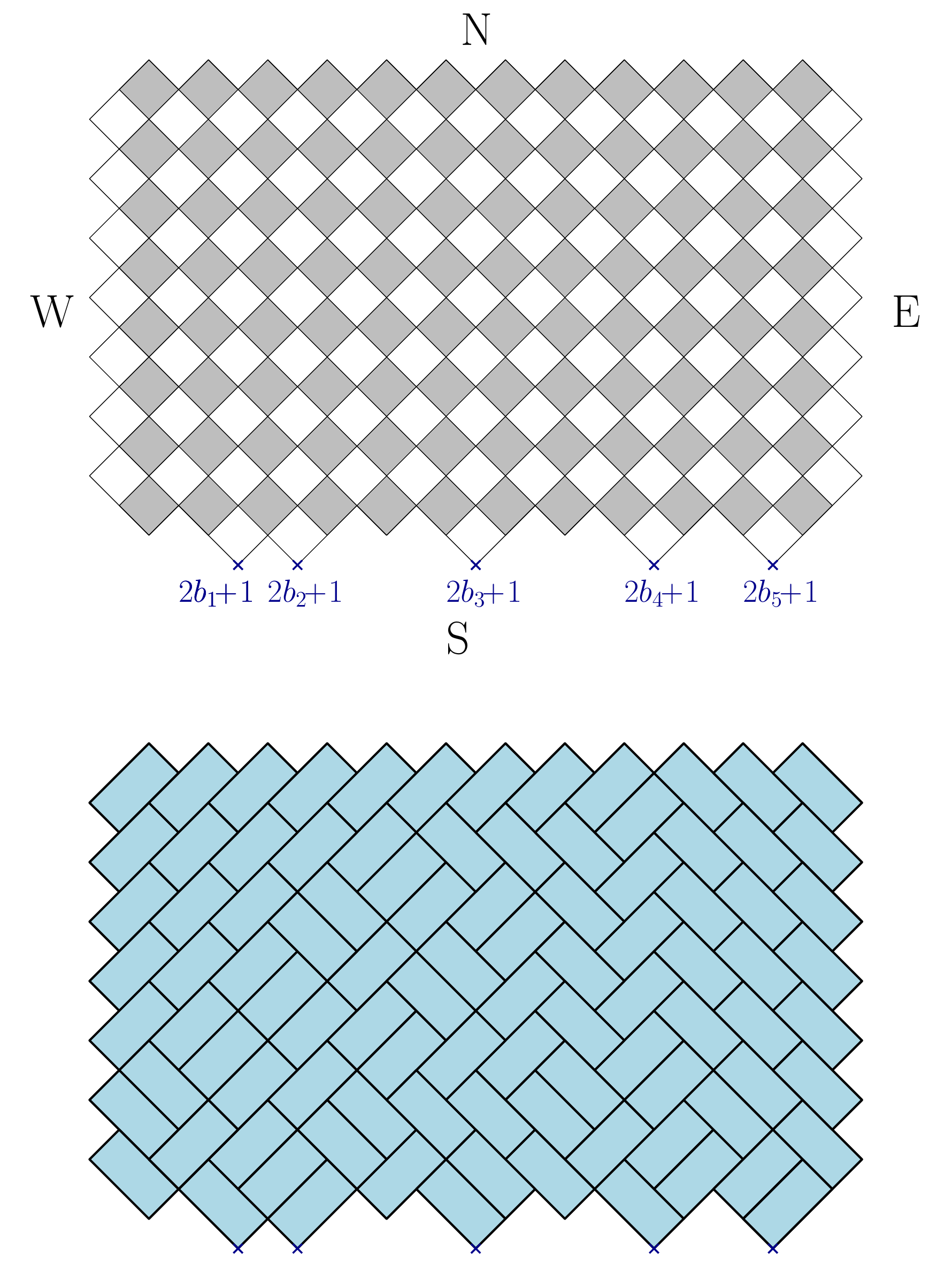}
\end{center}
\caption{\small Top: A typical Aztec rectangle of size $(n+1)\times (m+1)$ with $n=7$, $m=12$, and $m-n=5$ defects along
the $S$ boundary with positions $(2b_i+1)$, $i=1,\ldots, 5$ with $\{b_i\}_{1\leq i\leq 5}=\{ 2,3,6,9,11\}$. We have represented the natural black/white checkerboard structure of the underlying square lattice.  Bottom: A sample domino tiling of this Aztec rectangle with defects by means of $2\times 1$ dominos covering two adjacent squares.
}
\label{fig:Aztec}
\end{figure}
Our starting point is a particular domain drawn on the tilted (by $45^\circ$) square lattice with vertex set $\mathcal{L}=\big\{ (s,t)\in \Z^2\ \vert\  s+t=0 \mod 2\big\}$, and which
we may describe as an \emph{Aztec rectangle} of size $(n+1)\times (m+1)$ (with $m\geq n$) \emph{with defects along its lower boundary}. More precisely, 
we consider the domain whose zig-zag shaped boundary (passing through points of $\mathcal{L}$) is made of the following four parts, denoted $S$, $N$, $W$, $E$ respectively, where
the $S$ boundary displays a sequence of defects at positions $2b_j+1$, $j=1,2,...,m-n$ where $\{b_j\}_{1\leq j\leq m-n}$ is a strictly increasing sequence of integers in $[1,m-1]$:
\begin{itemize}
\item{S:} $\Big[\!\big\{(s,t) \in\mathcal{L}\ \vert$ $s\in [0,2m+2],\ t\in \{ 0,1\}\}\big\}\setminus\bigg(\big\{(0,0),(2m\!+\!2,0)\big\}\cup\big\{(2b_j\!+\!1,1),1\leq j\leq m\!-\!n\big\}\bigg)\!\Big]$

$\cup\big\{(2b_j\!+\!1,-1),1\leq j\leq m\!-\!n\big\}$;
\item{N:} $\big\{(s,t) \in\mathcal{L}\ \vert$  $s\in [0,2m+2],\ t\in \{ 2n+1,2n+2\}\big\}\setminus\big\{(0,2n+2),(2m+2,2n+2)\big\}$;
\item{W:}  $\big\{(s,t) \in\mathcal{L}\ \vert$ $s\in \{ 0,1\}, \ t\in [0,2n+2]\big\}\setminus\big\{(0,0),(0,2n+2)\big\}$;
\item{E:} $\big\{(s,t) \in\mathcal{L}\ \vert$ $s\in \{ 2m+1,2m+2\},  \ t\in [0,2n+2]\big\}\setminus\big\{(2m+2,0),(2m+2,2n+2)\big\}$.
\end{itemize}
Otherwise stated, the introduction of defects consists in \emph{adding} to the lower boundary of the regular $n\times m$ Aztec rectangle a number $n-m$ of \emph{elementary squares} at positions
$(2b_i+1)$, as illustrated in the top of Fig.~\ref{fig:Aztec} in the particular case $n=7$, $m=12$ and
$b_1=2$, $b_2=3$, $b_3=6$, $b_4=9$, $b_5=11$. 

Given the domain above, we now consider \emph{tiling configurations} of this domain by means of $2\times 1$ rectangular \emph{dominos} (tilted by $45^\circ$) covering two adjacent squares.
Such a domino tiling configuration is depicted in the bottom of Fig.~\ref{fig:Aztec}. 

\subsection{Lattice path formulation}
\label{sec:defaztec2}

\begin{figure}
\begin{center}
\includegraphics[width=8cm]{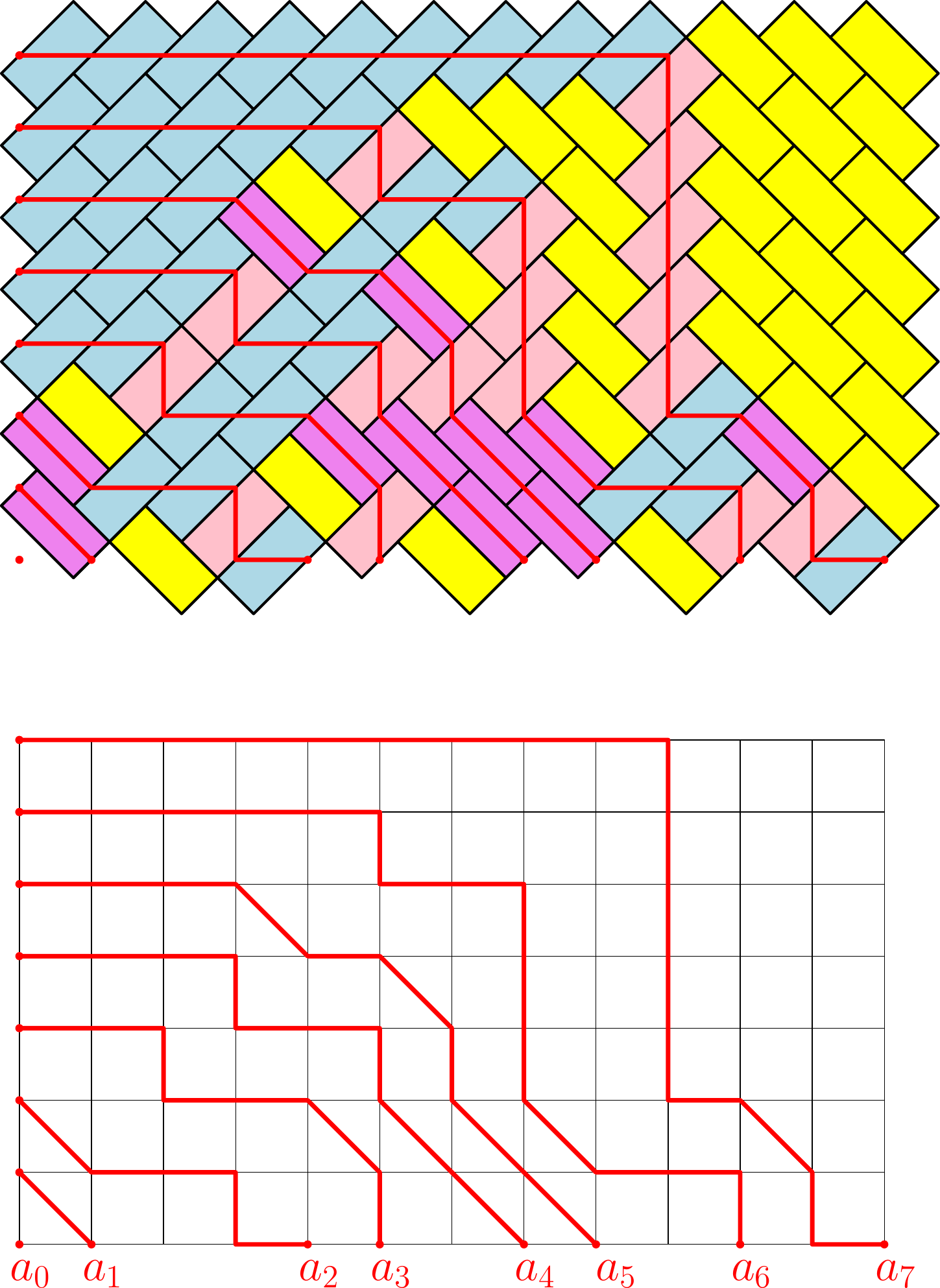}
\end{center}
\caption{\small Top: The four types of domino tiles of the Aztec rectangle tiling in the bottom of Fig.~\ref{fig:Aztec} are colored purple, yellow, pink, blue according to the order of the dictionary of Eq.~\eqref{dominos} and give rise to the system of red NILP (non-intersecting family of Schr\"oder paths).  
Bottom: We emphasize the starting points of the $n+1=8$ red NILP on a new underlying square lattice $\Z\times\Z$. These take the positions complementing the $b_i$'s on the real axis, namely $\{a_i\}_{0\leq i\leq n}=\{ 0,1,4,5,7,8,10,12\}$. 
}
\label{fig:NILP}
\end{figure}

Using the natural bi-coloration of the lattice $\mathcal{L}$, each domino 
gets bi-colored with one white and one black square. This gives rise to four kinds of domino tiles, according to their orientation and coloring. A standard bijection
then maps the domino tiling configurations to configurations of families of Non-Intersecting Lattice Paths
(NILP): this is realized 
by the following correspondence between dominos and elementary path steps:
\begin{equation}\label{dominos}
\raisebox{-.6cm}{\hbox{\epsfxsize=12.cm \epsfbox{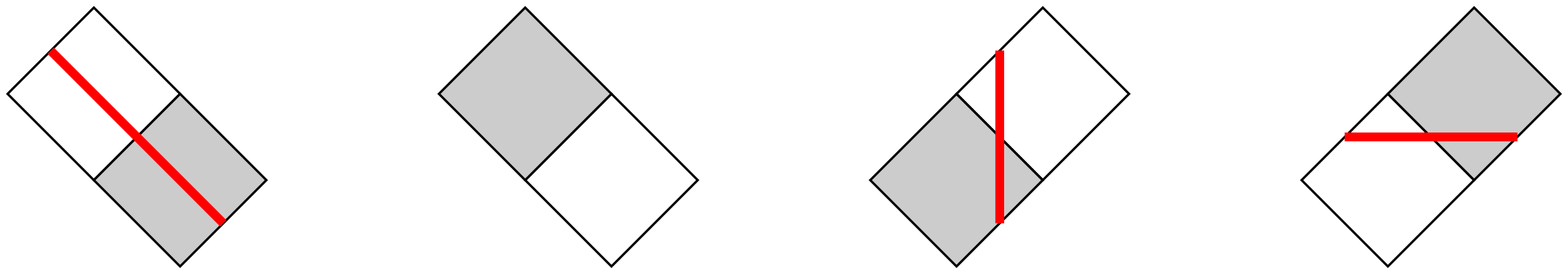}}}
\end{equation}
The paths created by this correspondence live on another (non-tilted) square lattice whose vertices are naturally labelled by integer coordinates $(i,j)$
and connect the S to the W boundary. When oriented from S to W, they use three kinds of steps: up steps $(0,1)$, left steps $(-1,0)$, and diagonal steps $(-1,1)$. 
Equivalently, the paths visit the edges of a directed lattice $\mathcal{N}$ made of west oriented horizontal lines, north-oriented vertical lines, and north-west oriented diagonal lines
through the vertices of the regular square lattice (the lattice $\mathcal{N}$ is topologically equivalent to an oriented triangular lattice).  
The particular boundary shape of the Aztec rectangle with defects on the S boundary implies that the paths of the corresponding family 
start at points $(a_i,0)$, $i=0,1,2,...,n$ and end at points $(0,j)$, $j=0,1,2,...,n$, where $a_i$ are the complements
of the positions $b_j$ on the segment $[0,m]$, namely $\{a_i\}_{0\leq i\leq n} \cup \{b_j\}_{1\leq j\leq m-n}=[0,m]$.
Note that, for technical reasons, we incorporated in the path family a trivial path from $(a_0,0)$ to $(0,0)$ starting and ending at the origin. Note also that $a_n=m$ by construction.
For illustration, the tiling of the bottom of Fig.~\ref{fig:Aztec} is mapped onto the red NILP configuration of Fig.~\ref{fig:NILP}, with $a_0=0$, $a_1=1$, $a_2=4$, $a_3=5$, $a_4=7$, $a_5=8$, $a_6=10$, $a_7=12$.

Two special restrictions of the paths considered here give rise to the so-called ``large" or ``small" ``Schr\"oder paths".
By lack of a better name, we shall refer to our paths as Schr\"oder paths.

\subsection{Weights}
\label{sec:weights}
\begin{figure}
\begin{center}
\includegraphics[width=8cm]{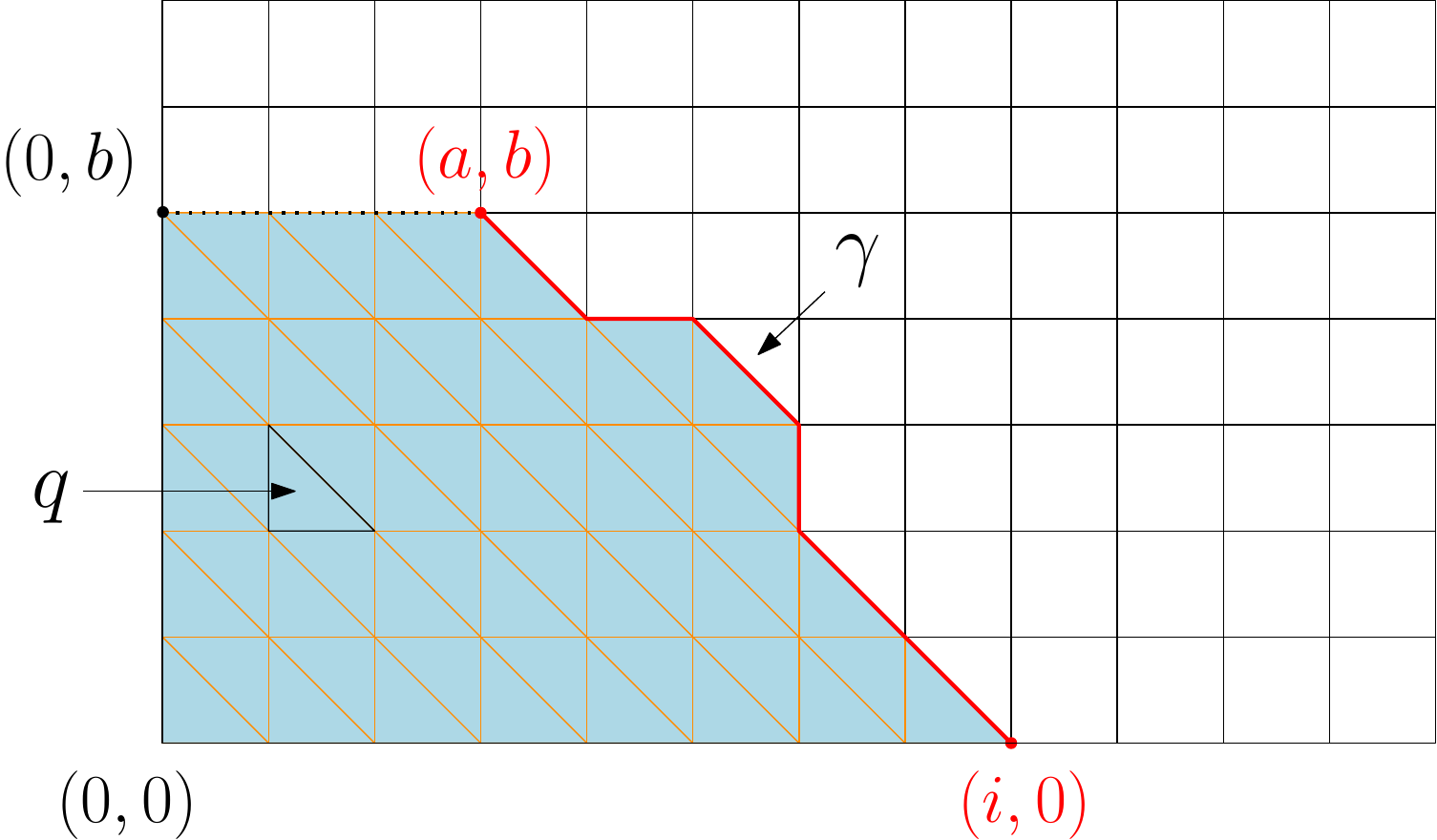}
\end{center}
\caption{{\small The weight of a Schr\"oder path $p$ (bold red) starting at $(i,0)$ and ending at $(a,b)$ is $\gamma^{k(p)} \, q^{\mathcal{A}(p)}$ where 
$k(p)$ denotes the number of diagonal steps of $p$ and $\mathcal{A}(p)$ counts the number of triangles to the left of $p$ (domain in light blue).
Otherwise stated, a weight $\gamma$ is attached to each diagonal step of $p$ and a weight $q$ to each triangle to the left of $p$.}
}
\label{fig:pathweight}
\end{figure}
In the following we will use the NILP formulation to enumerate the tiling configurations.  
We consider in fact the more general case of \emph{weighted} NILP with weights as follows: first we attach a multiplicative real weight $\gamma$ (with $\gamma \geq 0$) 
to each diagonal step, corresponding in turn to a domino of the first kind in Eq.\eqref{dominos} (see Fig.~\ref{fig:pathweight}). For any Schr\"oder path $p$, if $k(p)$
denotes the number of diagonal steps of $p$, the corresponding path weight  therefore reads $\gamma^{k(p)}$.

Next we introduce another weight keeping track
of the \emph{area "to the left" of the path}. More precisely, let $p=(p_0,p_1,...,p_r)$ be a Schr\"oder path starting at $p_0=(i,0)$ and ending at $p_r=(a,b)$ (with $p_{s+1}-p_s\in \{(0,1),(-1,0),(-1,1)\}$
by definition). Further denoting by $p_s=(x_s,y_s)$ the labelling of the successive points,
we define the area $\mathcal{A}(p)$ to the left of the path $p$ as:
$$ \mathcal{A}(p):= \sum_{s=0}^{r-1} (x_s+x_{s+1})(y_{s+1}-y_s)\ . $$
Splitting each square of the underlying square lattice into two triangles along the diagonal $(-1,1)$ (or equivalently visualizing the paths on the oriented triangular lattice $\mathcal{N}$), 
$\mathcal{A}(p)$ counts the number of \emph{triangles} to the left of the path $p$, namely in the domain of the first quadrant $\{(i,j), i,j\geq 0\}$
delimited by $p$ itself and by the horizontal segment joining $(a,b)$ to $(0,b)$ (see Fig.~\ref{fig:pathweight}). In many of the applications below, the endpoint will be at $a=0$. 
In this case, the area to the left of the path may also be interpreted as the area below the path. Given a path $p$, we decide to attach a real weight $q$ to each
triangle (with $q>0$) to the left of $p$, so $p$ receives eventually a total weight $\gamma^{k(p)} \, q^{\mathcal{A}(p)}$. The weight 
of a NILP configuration is then defined as the product of its path weights.

The particular case $\gamma=0$, corresponding to paths made of vertical (up) and horizontal (left) elementary steps only was already studied in detail in Ref.~\cite{DFG2}, (with the slight modification $q^2\to q$). In terms of tilings, setting $\gamma=0$ suppresses one of the four possible domino tiles. 
As explained in \cite{DFGUI}, the associated NILP configurations may then be put in correspondence with some particular rhombus tiling problem
with three types of elementary rhomboidal tiles.

\subsection{Partition function}

The partition function of the weighted tiling model for our Aztec rectangle with defects is defined in terms of NILP as
the sum $\mathcal{Z} := \mathcal{Z}(\gamma,q,\{a_i\}_{0\leq i\leq n})$ over all the corresponding NILP 
configurations of the product of the associated path weights. 
It may be computed by use of the celebrated Lindstr\"om Gessel-Viennot formula \cite{LGV1,GV} as follows:

\begin{defn}\label{parfundef}
Let $Z_{(i,0)\to(0,j)}:= Z_{(i,0)\to(0,j)}(\gamma,q)$ denote the partition function of {\emph single} weighted Schr\"oder paths starting at $(i,0)$
and ending at $(0,j)$, namely
\begin{equation}\label{defZij}
Z_{(i,0)\to (0,j)}:=\sum_{{\rm Schr\"oder\ paths}\, p\atop {\rm from}\ (i,0)\ {\rm to}\ (0,j)} \gamma^{k(p)}\ q^{\mathcal{A}(p)} \ .\end{equation}
\end{defn}

Let us form the $n+1\times n+1$ matrix $A$ with entries:
\begin{equation}
A_{i,j}:=Z_{(a_i,0)\to (0,j)} \qquad i,j=0,1,...,n\ .
\label{Adef}
\end{equation}
Then the partition function for families of $n+1$ NILP from the points $\{(a_i,0)\}_{0\leq i\leq n}$ to the points
$\{(0,j)\}_{0\leq j\leq n}$ reads:
\begin{equation}\label{defZ}
\mathcal{Z} =\det(A)= \det \left((Z_{(a_i,0)\to (0,j)})_{0\leq i,j \leq n}\right) \ .
\end{equation}
as a direct application of the Lindstr\"om Gessel-Viennot formula, which we may use since our NILP fulfill the following two required conditions: (i) the 
NILP weights may be described \emph{locally} by transforming 
the area and diagonal path weights into local weights attached to the visited edges of the oriented lattice $\mathcal{N}$, namely a weight 
$q^{2x}$ for a vertical edge $(x,y)\to (x,y+1)$,  $q^{2x-1}\gamma$ for a diagonal edge $(x,y)\to (x-1,y+1)$ and $1$ for a horizontal edge; (ii) the 
oriented lattice $\mathcal{N}$ and the choice of starting and endpoints satisfy the required ``crossing condition" that any pair of paths with starting and endpoints in opposite order (i.e. $(a_i,0)\to (0,j)$
and $(a_k,0)\to (0,\ell)$ with $i<k$ and $j>\ell$) must share a vertex.    

\section{Properties of a single weighted Schr\"oder path}
\label{sec:wls}

This section deals with properties of a \emph{single} Schr\"oder path and its statistics dictated by the weights $\gamma$ and $q$. 

\subsection{A summation formula}
\label{sec:wlsum} 
\begin{figure}
\begin{center}
\includegraphics[width=8cm]{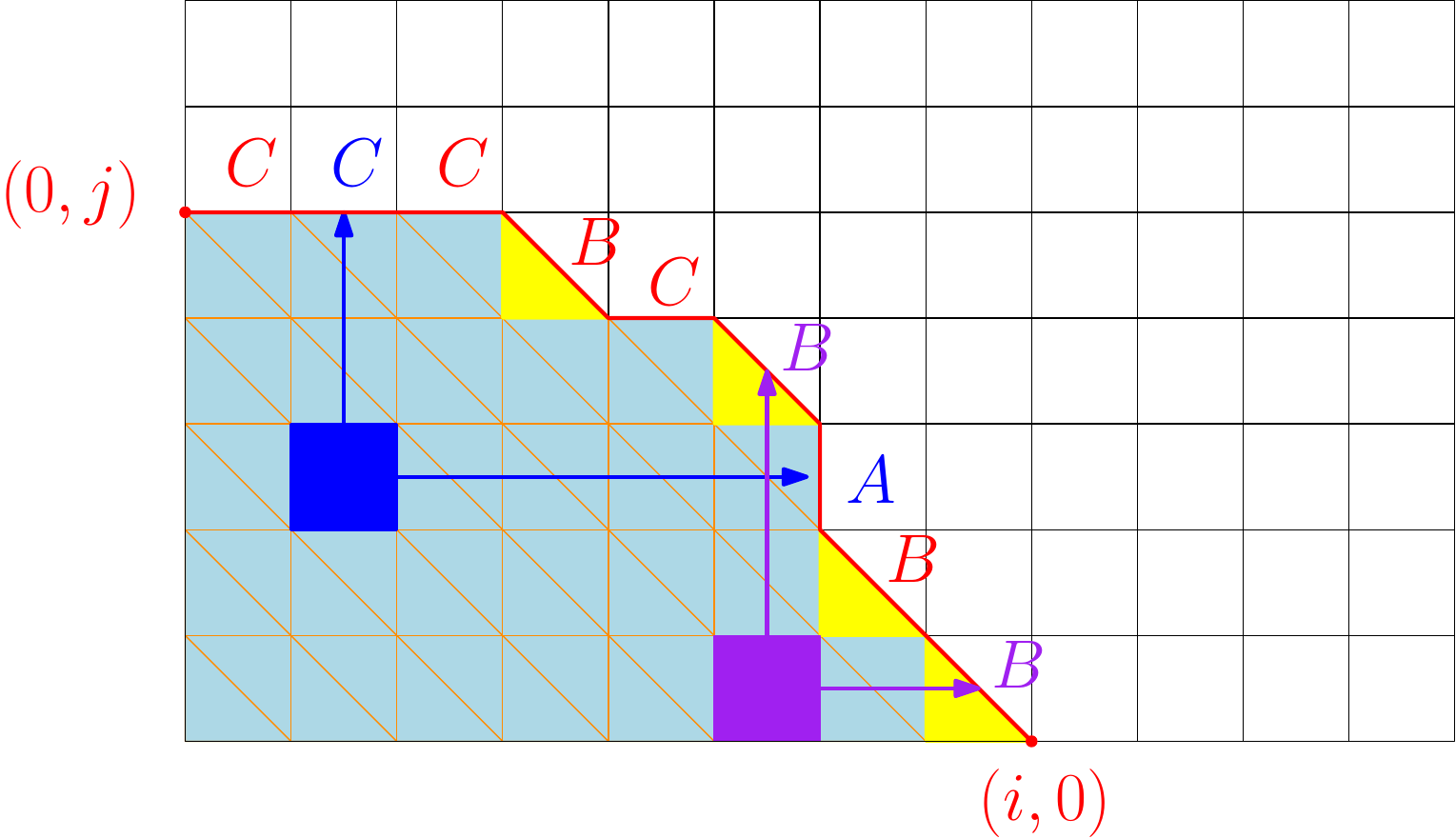}
\end{center}
\caption{{\small A Schr\"oder path connecting the point $(i,0)$ to the point $(0,j)$ and 
its associated word $CCCBCBABB$ obtained by associating the letter $A$ (respectively $B$ and $C$) to each vertical (respectively diagonal and horizontal) step.  Each unit square below the path is characterized by  
a pair $(C,A)$, $(C,B)$, $(B,A)$ or $(B,B)$ appearing in this order the word, hence to an inversion within the word, or to a $(B,B)$ pair. Here we emphasized for illustration 
a unit square associated with a $(C,A)$ pair (in blue)
and one associated with a $(B,B)$ pair (in magenta). The sample word has $21$ inversions and $6$ $(B,B)$ pairs, corresponding to a total of $27$ unit squares below the path.
The total area below the path also involves the triangles (in yellow) immediately below the diagonal steps (hence the letters $B$). 
}}
\label{fig:inversion}
\end{figure}
We start our analysis by deriving a simple formula for the partition function $Z_{(i,0)\to(0,j)}$ of \eqref{defZij}. 
As a first remark, let us mention the following 
recursion relation (over $i+j$) for $Z_{(i,0)\to(0,j)}$: decomposing along the first (left, diagonal or up) step from $(i,0)$, we
immediately deduce the recursive identity:
\begin{equation*}
Z_{(i,0)\to(0,j)}=Z_{(i-1,0)\to(0,j)}+\gamma\, q^{2i-1}\, Z_{(i-1,0)\to(0,j-1)}+q^{2i}\, Z_{(i,0)\to(0,j-1)}
\end{equation*}
where the last two terms implicitly involve a downwards vertical shift by $1$. Together with  $Z_{(i,0)\to(0,0)}=Z_{(0,0)\to(0,j)}=1$, this recursion 
determines  $Z_{(i,0)\to(0,j)}$ completely. Rather than solving this recursion, let us show directly by combinatorial arguments the following:

\begin{thm}\label{trinothm}
The partition function $Z_{(i,0)\to(0,j)}$ of \eqref{defZij} for weighted Schr\"oder paths admits the following simple summation formula:
\begin{equation}
 Z_{(i,0)\to(0,j)}=\sum_{k=0}^{{\rm Min}(i,j)} \gamma^k\, q^{k^2}\ {i+j-k \brack j-k,k,i-k}_{q^2}\ ,
\label{eq:Zij}
\end{equation}
where we introduced the $q^2$-trinomial\footnote{The appearance of $q^2$-trinomials rather than $q$-trinomials is due to our definition of the area $\mathcal{A}(p)$ enumerating triangles,
so that the ``unit" square $[0,1]\times[0,1]$ has area $2$.}:
\begin{equation}
{a+b+c\brack a,b,c}_{q^2}:= 
\frac{\prod\limits_{s=1}^{a+b+c}(q^{2s}-1)}{\prod\limits_{s=1}^{a}(q^{2s}-1)\prod\limits_{s=1}^{b}(q^{2s}-1)\prod\limits_{s=1}^{c}(q^{2s}-1)} \quad \hbox{\rm for}\ a,b,c \geq 0\ .
\label{eq:qtrinomial}
\end{equation}
\end{thm}
\begin{proof}
Recall first the classical combinatorial interpretation of the $q^2$-trinomial appearing in \eqref{eq:qtrinomial} as the sum over all words 
$w$ in the alphabet $\{A,B,C\}$, of total length $(a+b+c)$ with $a$ occurrences of the letter $A$,
$b$ occurrences of the letter $B$ and $c$ occurrences of the letter $C$, 
weighted by $(q^{2})^{I(w)}$ where ${I}(w)$ denotes the number of \emph{inversions} in 
the word $w$. By an inversion in $w$ we mean a pair of letters $\{B,A\}$ in $w$ (respectively $\{C,A\}$ or $\{C,B\}$) where the letter $B$ appears before (i.e.\ to the left of) the letter $A$ in $w$ (respectively
$C$ before $A$ or $C$ before $B$). Otherwise stated, ${I}(w)$ is the minimal number of permutations of pairs of adjacent letters needed to bring the word to the ordered 
form $AAA\ldots BBB\ldots CCC$\footnote{For instance the word $ABBACCBCA$ with $a=b=c=3$ has $10$ inversions.}.

Consider now a Schr\"oder path $p$ connecting the point $(i,0)$ to the point $(0,j)$ and with a total of $k(p)$ diagonal steps. 
This path may be coded bijectively by a word $w(p)$ of length $i+j-k(p)$ obtained by listing the sequence of paths steps read downwards from left to right
with a letter $A$ for each vertical step, a letter $B$ for each diagonal step and a letter $C$ for each horizontal step
(see Fig.~\ref{fig:inversion}). The numbers $a,b,c$ of letters $A$, $B$, $C$, namely
the numbers of vertical, diagonal and horizontal path steps respectively, satisfy $a+b=j$ and $c+b=i$ and $b=k(p)$, hence $(a,b,c)=(j-k(p),k(p),i-k(p))$. 
We have in particular the necessary condition $0\leq k(p)\leq {\rm Min}(i,j)$. As appears clearly in Fig.~\ref{fig:inversion},
each unit square below the path $p$ is determined by the data of the horizontal or diagonal path step above it and of the diagonal or vertical path step to its right, 
hence by a pair made of a letter $C$ or $B$, and a letter $B$ or $A$ appearing later in $w(p)$. Each unit square (with contribution $2$ to $\mathcal{A}(p)$) below the path is thus associated to 
an inversion $\{C,A\}$, $\{B,A\}$ or $\{C,B\}$ or to 
a pair $\{B,B\}$ with two \emph{distinct} occurrences of the letter $B$. Since there are $b$ occurrences of the letter $B$, the total number of unit squares below
the path is thus ${I}(w(p))+b(b-1)/2$. To get the total area, we also need to add a contribution $b$ from the remaining single triangles (each contributing $1$ to $\mathcal{A}(p)$) 
adjacent (and below) each of the $b$ diagonal steps. To summarize, the total area below the path $p$ is 
$$\mathcal{A}(p)=2\left({I}(w(p))+b(b-1)/2\right)+b=
2 {I}(w(p))+b^2$$
with $b=k(p)$ and we therefore get a total path weight $\gamma^{k(p)}\, q^{k(p)^2}\, (q^{2})^{{I}(w(p))}$.
Summing over all Schr\"oder paths $p$ from $(i,0)$ to $(0,j)$ with a fixed $k(p)=k$ boils down to enumerating all inversion-weighted words with 
$(a,b,c)=(j-k,k,i-k)$, and the formula \eqref{eq:Zij} follows by summing over the allowed values of $k$. This completes the proof of the theorem.
\end{proof}

We conclude this section with a useful polynomiality property of the path partition function $Z_{(i,0)\to(0,j)}$.

\begin{thm}\label{polthm}
The partition function for weighted Schr\"oder paths $Z_{(i,0)\to(0,j)}$ may be written as a \emph{polynomial 
$z_j(t)$  of degree $j$ in the variable $t=q^{2i}$}, where:
\begin{equation}\label{eq:polz}
z_j(t):=\sum_{k=0}^{j} \gamma^k\, q^{k^2}\ \frac{\prod\limits_{s=1}^{j}(t\, q^{2(s-k)}-1)}{\prod\limits_{s=1}^{j-k}(q^{2s}-1)\prod\limits_{s=1}^{k}(q^{2s}-1)}=\sum_{k=0}^{j} \gamma^k\, q^{k^2}\, {j\brack k}_{q^2}
\,\prod\limits_{s=1}^{j}\frac{t\, q^{2(s-k)}-1}{q^{2s}-1}
\end{equation}
in terms of the $q^2$-binomial coefficient ${j\brack k}_{q^2}:={j\brack k,j-k,0}_{q^2}$.
\end{thm}
\begin{proof}
From the $q^2$-trinomial expression, we may write
\begin{equation*}
{i+j-k\brack j-k,k,i-k}_{q^2}= \frac{\prod\limits_{s=1}^{i+j-k}(q^{2s}-1)}{\prod\limits_{s=1}^{j-k}(q^{2s}-1)\prod\limits_{s=1}^{k}(q^{2s}-1)\prod\limits_{s=1}^{i-k}(q^{2s}-1)}=
\frac{\prod\limits_{s=1}^{j}(q^{2(s+i-k)}-1)}{\prod\limits_{s=1}^{j-k}(q^{2s}-1)\prod\limits_{s=1}^{k}(q^{2s}-1)} \ ,
\end{equation*}
so that
\begin{equation*}
 Z_{(i,0)\to(0,j)}=\sum_{k=0}^{{\rm Min}(i,j)} \gamma^k\, q^{k^2}\ \frac{\prod\limits_{s=1}^{j}(q^{2(s-k+i)}-1)}{\prod\limits_{s=1}^{j-k}(q^{2s}-1)\prod\limits_{s=1}^{k}(q^{2s}-1)}\ .
\end{equation*}
Let us now show that we may replace the upper bound ${\rm Min}(i,j)$ in the summation by $j$. This is straightforward if $j\leq i$ since ${\rm Min}(i,j)=j$  in this case. As for the situation
where $j>i$, we may extend the sum over $k$ from ${\rm Min}(i,j)+1=i+1$ to $j$ since the product in the numerator vanishes identically for all these extra added terms 
due to the contribution $(q^0-1)=0$ arising from the
term $s=k-i$ (which lies in the integer interval $[1,j]$ since $k\in [ i+1,j] \subset [ i+1,i+j]$).  
This allows to rewrite $Z_{(i,0)\to(0,j)}=z_j(q^{2i})$, with $z_j(t)$ as in \eqref{eq:polz},
which displays $z_j(t)$ as an explicit polynomial of degree $j$ in $t$.
\end{proof}

\subsection{Scaling limit}
\label{sec:scalingone}
Our NILP problem involves a collection of $(n+1)$ non-intersecting Schr\"oder paths with endpoints $(0,j)$ for $j=0,\ldots,n$. Since we will eventually consider the limit of large $n$, 
the positions of the starting and endpoints of our paths may eventually become large, of the order of $n$.
A natural question is then to estimate the partition function $Z_{(i,0)\to(0,j)}$ of a single Schr\"oder path from $(i,0)$ to the point $(0,j)$ in the limit 
where $i$ and $j$ become large and scale as $n$. A sensible large $n$ limit is reached by keeping $\gamma$ fixed but letting $q$ scale exponentially with $1/n$.
In other words, we consider the scaling
\begin{equation*}
i=u\, n\ , \qquad j= v\, n\ , \qquad q=\qq^{\frac{1}{n}}
\end{equation*}
with large $n$ and finite $u,v\geq 0$ and $\qq>0$.
 
With this scaling, $Z_{(i,0)\to(0,j)}$ grows exponentially with $n$ as follows:
\begin{thm}\label{schroas}
The scaled partition function for weighted Schr\"oder paths behaves for large $n$ as 
\begin{equation}
Z_{(un,0)\to(0,vn)}\underset{n \to \infty}{\sim} {\rm e}^{n\, S_0(u,v)}
\label{eq:Zijscal}
\end{equation}
where 
$S_0(u,v)=S_0(u,v,\phi(u,v))$ and
\begin{equation}
\begin{split}
S_0(u,v,\phi)& =\phi {\rm Log}\, \gamma +\phi^2\, {\rm Log}\, \qq +\int_{0}^{u+v-\phi}d\sigma\,  {\rm Log}(\qq^{2\sigma}-1)
- \int_{0}^{v-\phi}d\sigma\,  {\rm Log}(\qq^{2\sigma}-1)\\
  & \qquad \qquad \qquad \qquad \qquad-\int_{0}^{\phi}d\sigma\,  {\rm Log}(\qq^{2\sigma}-1)
-\int_{0}^{u-\phi}d\sigma\,  {\rm Log}(\qq^{2\sigma}-1)\\
\end{split}
\label{eq:Sform}
\end{equation}
while $\phi(u,v)$ is the 
unique solution to $\frac{\partial{S_0(u,v,\phi)}}{\partial\phi}=0$ satisfying $0\leq \phi\leq {\rm Min}(u,v)$, namely:
\begin{equation}\label{sapok}
\gamma\, \qq^{2\phi} \frac{(\qq^{2(u-\phi)}-1)(\qq^{2(v-\phi)}-1)}{(\qq^{2 \phi}-1)(\qq^{2(u+v-\phi)}-1)}=1 \qquad (0\leq \phi\leq {\rm Min}(u,v) ).
\end{equation}
\end{thm}
\begin{proof}
Using the $q^2$-trinomial expression of Theorem \ref{trinothm}, and substituting $i=un,j=vn,k=\phi n$,
we may express the leading exponential behavior of 
$Z_{(un,0)\to(0,vn)}$ as the integral:
$$Z_{(un,0)\to(0,vn)}\underset{n \to \infty}{\sim} \int_0^{{\rm Min}(u,v)} d\phi {\rm e}^{n S_0(u,v,\phi)} $$
with $S_0$ as in \eqref{eq:Sform}. This integral is dominated by the saddle-point $\phi=\phi(u,v)$ 
that maximizes $S_0$, hence satisfies \eqref{sapok}.
The existence and uniqueness of this point
is best seen upon setting
\begin{equation*}
F=\qq^{2\phi}\ , \qquad U=\qq^{2u}\ ,\qquad V=\qq^{2v}\ ,
\end{equation*}
so that the equation \eqref{sapok} for $\phi$ reads
\begin{equation*}
F^2-\left(U+V+\frac{(U-1)(V-1)}{1+\gamma}\right)F+U\, V=0. 
\end{equation*}
Assuming $\qq>1$ and, say $u\leq v$, we have $1\leq U\leq V$. In particular, the roots $F_1$ and $F_2$ of this equation satisfy
$F_1F_2=UV$ and $F_1+F_2=\left(U+V+\frac{(U-1)(V-1)}{1+\gamma}\right)$ so that $F_1+F_2\geq 2$ and $(F_1-1)(F_2-1)=\gamma (U-1)(V-1)/(1+\gamma)\geq 0$,
which implies that both $F_1$ and $F_2$ are larger than or equal to $1$. 
Finally, from $\left(U+V+\frac{(U-1)(V-1)}{1+\gamma}\right)\geq U+V$, $F_1$ and $F_2$ also satisfy the inequality 
\begin{equation*}
0\leq F^2-\left(U+V\right)F+U\, V =(F-U)(F-V)\ .
\end{equation*}
This implies that $F_1$ and $F_2$ lie outside of the segment $]U,V[$ and, 
since their product is $U\, V$, one of the roots lies below $U$ and the other above $V$. There is therefore 
a unique solution $F$ with $1\leq F\leq U$, that is $0\leq \phi\leq u={\rm Min}(u,v)$.
A similar argument can be worked out for $u>v$ or $\qq<1$.
\end{proof}

\subsection{Equation for geodesics}
\label{sec:freetrajectory}
\begin{figure}
\begin{center}
\includegraphics[width=12cm]{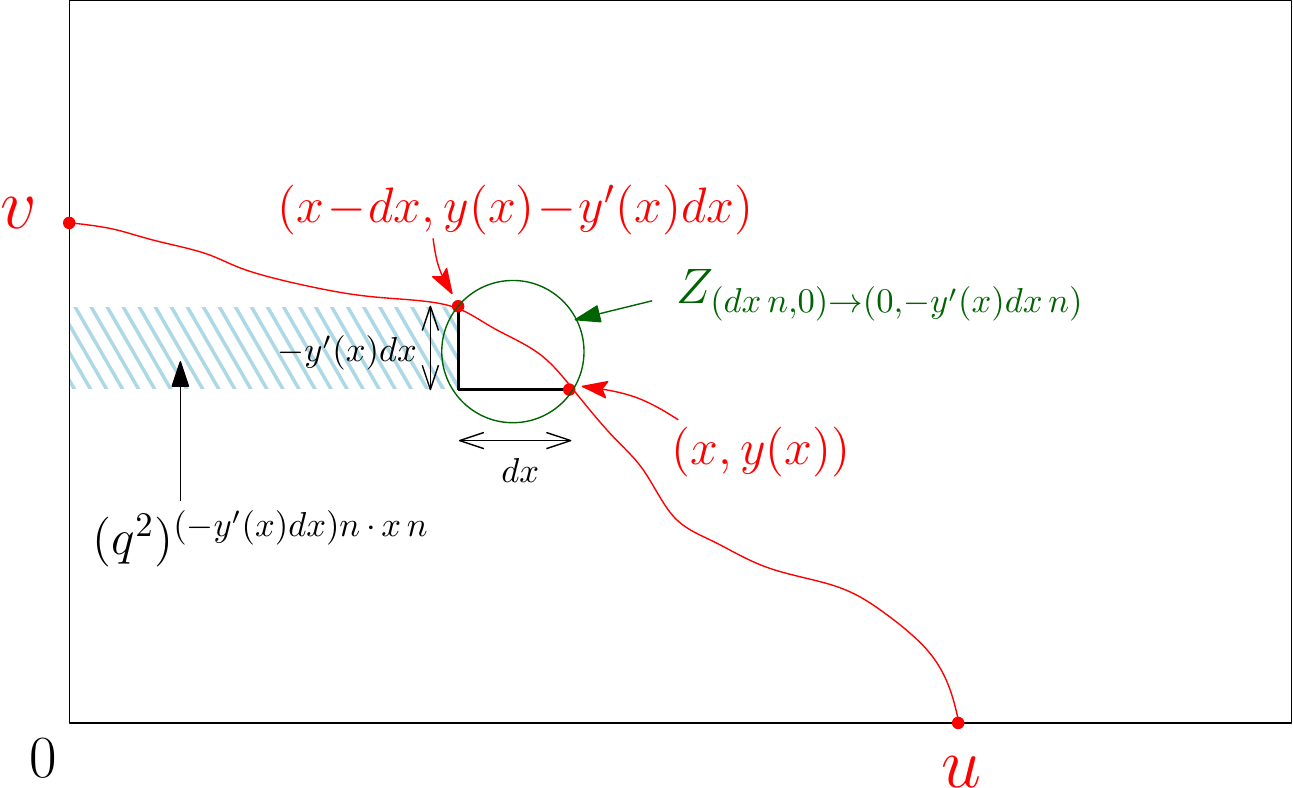}
\end{center}
\caption{\small{A schematic picture, here in rescaled coordinates, of a path $(xn,y(x)n)$ from $(un,0)$ to $(0,vn)$,  
characterized by its shape $y(x)$, $0\leq x\leq u$ with $y(0)=v$, $y(u)=0$ and $y'(x)\leq 0$. The contribution to the partition function of an infinitesimal 
portion of path from $(x,y(x))$ to $((x-dx),(y(x)-y'(x)dx))$ is the product of 
$Z_{(dx\, n,0)\to(0,-y'(x)dx\, n)}$ (accounting for the contribution in the circle) times $(q^{2})^{(-y'(x) dx)n\, \cdot\, x\, n}=(\qq^{2})^{- n\,x\, y'(x)dx}$ (accounting for the area weight
of the shaded region).
}}
\label{fig:geodesic}
\end{figure}
Besides the actual exponential growth of $Z_{(un,0)\to(0,vn)}$, an interesting question is that of finding the geodesic path from $(un,0)$ to $(0,vn)$,
i.e.\ the path whose contribution is maximal in the partition function $Z_{(un,0)\to(0,vn)}$. Having an explicit expression for geodesics will be instrumental for applying 
the Tangent Method in Section~\ref{sec:arctictwo}. A rescaled path from $(un,0)$ to $(0,vn)$ may be 
characterized by its shape $y(x)$, $0\leq x\leq u$ (see Fig.~\ref{fig:geodesic}) which gives the sequence of its positions $(xn,y(x)n)$. This function satisfies 
$y(0)=v$, $y(u)=0$ and $y'(x)\leq 0$ along the whole path. Reading the path from right to left, its contribution to $Z_{(u\, n,0)\to(0,v\, n)}$ is obtained by multiplying 
the contribution of all the infinitesimal portions of path from $(xn,y(x)n)$ to $((x-dx)n,(y(x)-y'(x)dx)n)$, which involves computing the infinitesimal partition function
$Z_{(xn,y(x)n)\to((x-dx)n,(y(x)-y'(x)dx)n)}$. By a simple shift of the origin to position $((x-dx)n,y(x)n)$, we may write
\begin{equation*}
Z_{(xn,y(x)n)\to((x-dx)n,(y(x)-y'(x)dx)n)}= (\qq^2)^{-n\, x\, y'(x)dx} Z_{(dx\, n,0)\to(0,-y'(x)dx\,n)}
\end{equation*}
where the factor $(\qq^2)^{-n\,x\, y'(x)dx}$ incorporates the area weight to the left\footnote{Here we use the area to the left of the path, in accordance with 
our general convention, but we could as well use the area below the path as the two are equivalent for a path that ends on the vertical axis. Using the area below the path would produce 
instead a term $(\qq^2)^{n\, y(x)dx}$, leading to an overall identical contribution since $\int_0^udx\, y(x)=\int_0^udx(-x\, y'(x))$.} of the curve between heights $y(x)\, n$ and $(y(x)-y'(x)dx)n$. The quantity
$Z_{(dx\, n,0)\to(0,-y'(x)dx\ ,n)}$ may then be evaluated from the general expression \eqref{eq:Zijscal}-\eqref{eq:Sform}
of Theorem \ref{schroas} for $Z_{(un,0)\to(0,vn)}$ by performing 
the substitutions $u\to dx$, $v\to -y'(x) dx$, $\phi \to \mu dx$ 
(with $0\leq \mu \leq {\rm Min}(1, -y'(x))$).  We immediately obtain:
\begin{equation*}
Z_{(dx\, n,0)\to(0,-y'(x)dx\, n)}= {\rm e}^{n\, \mathcal{L}(\{y'(x)\} dx}\\
\end{equation*}
with $\mathcal{L}(\{y'(x)\}):=\mathcal{L}(\{y'(x)\},\mu(\{y'(x)\}))$, where 
\begin{equation*}
\begin{split}
\mathcal{L}(\{y'(x)\},\mu)&=\mu \ {\rm Log}\, \gamma +(1-y'(x)-\mu)\,  {\rm Log}(1-y'(x)-\mu)
-(1-\mu)\,  {\rm Log}(1-\mu)\\
&\ \ \ \ -\mu\  {\rm Log}(\mu) -(-y'(x)-\mu)\,  {\rm Log}(-y'(x)-\mu) \\
\end{split}
\end{equation*}
(recall that $y'(x)\leq 0$). As before, the above expression must be taken at 
the value $\mu=\mu(\{y'(x)\})$ in $[0,{\rm Min}(1,-y'(x))]$ such that $\frac{\partial \mathcal{L}(\{y'(x)\},\mu)}{\partial \mu}=0$,
namely the following infinitesimal version of \eqref{sapok}:
\begin{equation}
\gamma \frac{(1-\mu)(-y'(x)-\mu)}{\mu \,(1-y'(x)-\mu)}=1\ .
\label{eq:mu}
\end{equation}
This allows to write formally $Z_{(un,0)\to(0,vn)}$ as a functional integral
\begin{equation*}
Z_{(un,0)\to(0,vn)}=\int_{y(0)=v\atop y(u)=0} \mathcal{D}y(x) {\rm e}^{\textstyle{n \int_0^{u}\left( -{\rm Log}(\qq^2) x\, y'(x)+\mathcal{L}(\{y'(x)\} \right) dx}}
\end{equation*}
and to deduce by a variational principle the equation of geodesics:
\begin{equation*}
-\left(\frac{\delta}{\delta y'(x)}\left(- {\rm Log}(\qq^2) x\, y'(x)+\mathcal{L}(\{y'(x)\}) \right)\right)'=0
 \qquad   \Leftrightarrow \qquad {\rm Log}(\qq^2)=\left(\frac{\delta \mathcal{L}(\{y'(x)\})}{\delta y'(x)}\right)'
\end{equation*}
with boundary conditions $y(0)=v$ and $y(u)=0$. Using 
\begin{equation*}
\frac{\delta \mathcal{L}(\{y'(x)\})}{\delta y'(x)}=\frac{\partial \mathcal{L}(\{y'(x)\},\mu)}{\partial y'(x)}+\frac{\partial \mathcal{L}(\{y'(x)\},\mu)}{\partial \mu}\frac{\delta \mu}{\delta y'(x)}=
\frac{\partial \mathcal{L}(\{y'(x)\},\mu)}{\partial y'(x)}
={\rm Log}\left(\frac{-\mu-y'(x)}{1-\mu-y'(x)}\right)
\end{equation*}
and, from \eqref{eq:mu},
\begin{equation*}
\frac{\delta \mu}{\delta y'(x)}=\frac{\mu\, (1-\mu)}{y'(x)\, (1-2\mu-y'(x))}\ ,
\end{equation*}
we deduce
\begin{equation*}
\begin{split}
\left(\frac{\delta \mathcal{L}(\{y'(x)\}}{\delta y'(x)}\right)'&= y''(x)\left(\frac{\partial}{\partial y'(x)}{\rm Log}\left(\frac{-\mu-y'(x)}{1-\mu-y'(x)}\right)
+\frac{\partial}{\partial \mu}{\rm Log}\left(\frac{-\mu-y'(x)}{1-\mu-y'(x)}\right)\times \frac{\delta \mu}{\delta y'(x)}\right)\\
&=\frac{y''(x)}{y'(x)\, (1-2\mu-y'(x))}
\end{split}
\end{equation*}
so that the equation for geodesics finally reads
\begin{equation*}
{\rm Log}(\qq^2)=\frac{y''(x)}{y'(x)\, (1-2\mu-y'(x))}
\end{equation*}
with $\mu=\mu(\{y'(x)\})$ as in \eqref{eq:mu}.

\medskip
To solve this equation, we set 
\begin{equation*}
X=\qq^{2x}\ , \qquad Y(X)=\qq^{2y(x)}\ ,  \qquad M(X)=\mu(x)\ , 
\end{equation*}
and introduce the function
\begin{equation*}
W(X):=y'(x)=\frac{X\, Y'(X)}{Y(X)}
\end{equation*}
so that $y''(x)={\rm Log}(\qq^2)\, X\, W'(X)$.
The equation for geodesics then simplifies into
\begin{equation*}
X\, W'(X)=W(X)(1-2 M(X)-W(X))\ , \qquad \gamma \frac{(1-M(X))(-W(X)-M(X))}{M(X) \,(1-W(X)-M(X))}=1\ ,
\end{equation*} 
which, upon eliminating $M(X)$, yields
\begin{equation*}
X\ W'(X)= W(X) \sqrt{\frac{\gamma (1+W(X))^2+(1-W(X))^2}{1+\gamma}}\ .
\end{equation*}
Here the choice of the correct branch when solving the quadratic equation for $M(X)$ is dictated by the fact that $0\leq M(X)\leq 1$ 
(choosing the other branch would introduce a global minus sign in the right hand side of the equation above). 
Equivalently, this choice ensures that $X W'(X)<0$ (recall that $W(X)=y'(x)<0$) and, since $y''(x)={\rm Log}(\qq^2)\, X\, W'(X)$, that $y''(x)<0$ for $\qq>1$ and $y''(x)>0$
for $\qq<1$, as expected on physical grounds: for $\qq>1$ (respectively $\qq<1$), the geodesic path tends
to be concave (respectively convex) to increase (respectively decrease) the area below. 
The above equation is easily integrated into
\begin{equation*}
W(X)=- \frac{(1+\gamma)^2}{2\gamma} \frac{X\, X_0}{\left(X_0+X\frac{1+\gamma}{2}\right)\left(X_0-X \frac{1+\gamma}{2\gamma}\right)}
\end{equation*}
where the integration constant $X_0$, to be determined later, must satisfy\footnote{The given expression for $W(X)$ is indeed solution of the differential equation only if $\left(X_0+X\frac{1+\gamma}{2}\right)\left(X_0-X \frac{1+\gamma}{2\gamma}\right)>0$, i.e.\, for the allowed range of $X$ (i.e.\ $[1,U]$ for $\qq>1$ or $[U,1]$ for $\qq<1$) only if $X_0<-\frac{1+\gamma}{2}\max(1,U)$ or $X_0> \frac{1+\gamma}{2\gamma} \max(1,U)$.
The first range of $X_0$ is ruled out by the fact that $W(X)$ must be negative.} $X_0> \frac{1+\gamma}{2\gamma} \max(1,U)$. From $W(X)=\frac{X\, Y'(X)}{Y(X)}$, we then deduce by integration
\begin{equation*}
Y(X)=Y_0\, \frac{X\, (1+\gamma)-2 X_0\, \gamma}{X\, (1+\gamma)+2 X_0}
\end{equation*}
with $Y_0$ yet to be determined.
Introducing as before the quantities $U=\qq^{2u}$ and $V=\qq^{2v}$, the constants $X_0$ and $Y_0$ are determined by imposing 
$Y(1)=V$ (i.e.\ $y(0)=v$) and $Y(U)=1$ (i.e.\ $y(u)=0$). This latter condition fixes $Y_0$ so that
\begin{equation}
Y(X)= \frac{X\, (1+\gamma)-2 X_0\, \gamma}{X\, (1+\gamma)+2 X_0}\cdot \frac{U\, (1+\gamma)+2 X_0}{U\, (1+\gamma)-2 X_0\, \gamma}\ .
\label{eq:geodone}
\end{equation}
As for the first condition, it yields a quadratic equation for $X_0$ with solution\footnote{Note that, for $U,V>0$, $(\sqrt{UV}-1)^2\geq 0$ implies $(UV+1)\geq 2\sqrt{UV}$
and $(U+V)\geq 2\sqrt{UV}\geq 4 UV/(UV+1)$. The quantity appearing in the square root in $\delta(U,V)$ is thus clearly non negative for $\gamma\geq 0$.}
\begin{equation}
\begin{split}
X_0&=\frac{1+\gamma}{4\gamma(V-1)}\left\{UV-1+\gamma(U-V)+\epsilon\, \delta(U,V)\right\}\\
& \hbox{with}\quad  \delta(U,V)=\sqrt{(UV-1)^2+2\gamma((U+V)(UV+1)-4UV)+\gamma^2(U-V)^2}\\
\end{split}
\label{eq:geodtwo}
\end{equation}
where the choice of sign $\epsilon=\pm 1$ fixing the correct branch of solution may be fixed by the condition $X_0> \frac{1+\gamma}{2\gamma} \max(1,U)$.
For $\qq>1$, we have $U>1$ and $V>1$ and we have to impose $X_0> \frac{1+\gamma}{2\gamma}U$, with 
\begin{equation*}
\begin{split}
X_0-\frac{1+\gamma}{2\gamma}  U &=\frac{1+\gamma}{4\gamma(V-1)}\left\{\alpha(U,V)+\epsilon\, \delta(U,V)\right\}\\
\alpha(U,V)&=2U-UV-1+\gamma(U-V)\ .\\
\end{split}
\end{equation*}
Noting that $\delta^2(U,V)-\alpha^2(U,V)=4(1+\gamma)U(U-1)(V-1)>0$, we deduce that only $\epsilon=+1$ fulfills the desired requirement.
Assume now $\qq<1$, in which case $U<1$ and $V<1$ and we have to impose $X_0> \frac{1+\gamma}{2\gamma}$, with 
\begin{equation*}
\begin{split}
X_0-\frac{1+\gamma}{2\gamma} &=-\frac{1+\gamma}{4\gamma(1-V)}\left\{\beta(U,V)+\epsilon\, \delta(U,V)\right\}\\
\beta(U,V)&=UV+1-2V+\gamma(U-V)\ .\\
\end{split}
\end{equation*}
Noting that $\delta^2(U,V)-\beta^2(U,V)=4(1+\gamma)V(1-U)(1-V)>0$, we deduce that only $\epsilon=-1$ fulfills the desired requirement.
We finally plug \eqref{eq:geodtwo} into \eqref{eq:geodone} to get the following:
\begin{thm}
\label{thm:geod}
The geodesic $(x,y(x))$ joining
$(u,0)$ to $(0,v)$ for the scaling limit of weighted Schr\"oder paths is given by:
\begin{equation}
\begin{split}
Y(X)&=\frac{\big(U V\!-\!1\!-\!2 (V\!-\!1) X\!+\!\gamma(U\!-\!V)\!+\!\epsilon(\qq)\delta(U,V)\big) \big(UV\!-\!1\!+\!\gamma(2 UV\!-\!U\!-\!V)\!+\!\epsilon(\qq)\delta(U,V) \big)}
{\big(2U\!-\!UV\!-\!1\!+\!\gamma(U\!-\!V)\!+\!\epsilon(\qq)\delta(U,V)\big) \big(UV\!-\!1\!+\!\gamma(U\!-\!V\!+\!2(V\!-\!1)X)\!+\!\epsilon(\qq)\delta(U,V)\big)}
\\
&= 1+\frac{(V\!-\!1)(U\!-\!X)\big((X\!-\!1)(V\!-\!1)(U\!+\!\gamma)\!+\!(U\!-\!1)(X\!+\!1)(1\!+\!\gamma)\!+\!(X\!-\!1)\epsilon(\qq)\delta(U,V)\big)}{2(U\!-\!1)\big((U\!-\!X)(1\!+\!\gamma X)\!+\!V(U\!+\!\gamma X)(X\!-\!1)\big)}\\
&\hbox{where}\ \ \ \ Y(X)=\qq^{2y(x)}\ , \quad X=\qq^{2x}\ ,
\quad U=\qq^{2u}\ , \quad V=\qq^{2v}\ ,\\
\end{split}
\label{eq:eqgeod}
\end{equation}
with $\delta(U,V)$ as in \eqref{eq:geodtwo}, $\epsilon(\qq)=\sgn({\rm Log}(\qq))$ and $x\in[0,u]$.
\end{thm}

Let us conclude with a few remarks. First, solving the last equation above for $\epsilon(\qq)\delta(U,V)$ and then writing $(\epsilon(\qq)\delta(U,V))^2=(UV-1)^2+2\gamma((U+V)(1+UV)-4UV)+\gamma^2(U-V)^2$ yields the following algebraic equation for the geodesic $Y=Y(X)$:
\begin{eqnarray}
&&\!\!\!\!\!\!G_\gamma(X,Y;U,V):=(U\!-\!1) (V\!-\!1)\Big(\gamma\left(UV\!+\!X^2 Y^2\right)\!-\!\left(V X^2\!+\!U Y^2\right)\Big)
\nonumber\\
&&-X Y \bigg((U\!+\!1) (V\!+\!1) \big(1\!+\!UV\!+\!\gamma (U\!+\!V) \big)\!-\!8(1\!+\!\gamma )UV\bigg)
+\!(V\!-\!1)\left(X^2\!+\!U\right) Y \big(UV\!-\!1\!+\!\gamma (V\!-\!U)\big)\nonumber\\
&&\qquad\qquad\qquad \qquad\qquad\qquad \qquad  +\!(U\!-\!1) \left(Y^2\!+\!V\right) X \big (UV\!-\!1\!+\!\gamma(U\!-\!V) \big)=0\ .\label{geo}
\end{eqnarray}
This curve contains the two branches corresponding to $\qq>1$ and $\qq<1$ respectively. We note the symmetry
of this curve under the flip w.r.t. the first diagonal under which $u\leftrightarrow v$ and $x\leftrightarrow y$,
namely under the simultaneous interchange $U\leftrightarrow V$ and $X\leftrightarrow Y$:
$$ G_\gamma(Y,X;V,U)=G_\gamma(X,Y;U,V) .$$
This is due to the fact that
the two definitions of the area (i) to the left of the curve or (ii) under the curve are equivalent. Another manifest symmetry
of the curve \eqref{geo} is obtained by reinterpreting geodesic paths 
under complement in the rectangle $n\, u\times n\, v$. Indeed, any path 
$p=\{(n\, x,n\, y)\}$ from 
$(n\, u,0)$ to $(0,n\, v)$ can be 
viewed, by performing a half-turn rotation by $180^\circ$ and reversing the travel direction, as a path 
$\tilde{p}=\{(n\, (u-x),n\, (v-y))\}$ from $(n\, u,0)$ to $(0,n\, v)$. The total area of the rectangle $n\, u\times n\, v$
is equal to $2 n^2 \, u\, v= {\mathcal A}(p)+{\mathcal A}(\tilde{p})$, as the area to the {\it left} of $\tilde {p}$ is nothing but the area to the {\it right} of $p$ within the rectangle. This immediately implies the symmetry:
$$ G_\gamma\left(\frac{X}{U},\frac{Y}{V};\frac{1}{U},\frac{1}{V}\right)=\frac{1}{U^3V^3}\, G_\gamma(X,Y;U,V) .$$
From this interpretation, the map $(X,Y,U,V)\mapsto (U/X,V/Y,1/U,1/V)$ clearly interchanges the two geodesic branches $\qq>1$ and $\qq<1$.

Next, for $\gamma= 0$, the formula for $Y(X)$ simplifies drastically since in this case $\epsilon(\qq)\delta(U,V)=(U V-1)$ so that \eqref{eq:eqgeod} becomes
\begin{equation*}
Y(X)=1+\frac{(V-1)(U-X)}{U-1}\quad (\gamma=0)
\end{equation*}
and the algebraic equation reduces to $(V-1)X+(U-1)Y=UV-1$. We recover here the geodesic equation found in \cite{DFG2}.

Finally, for $\qq\to 1$, using $X\sim 1+2 x (\qq-1)$, $Y(X)\sim 1+2y(x)(\qq-1)$, $U\sim 1+2 u (\qq-1)$, $V\sim 1+2v(\qq-1)$ (which yields in particular 
$\epsilon(\qq)\delta(U,V)\sim  2 (\qq-1)\sqrt{(1+\gamma)\big((u+v)^2+\gamma(u-v)^2\big)}$\ ), Eq.~\eqref{eq:eqgeod} becomes independent of $\gamma$ at leading order in $\qq-1$ 
and yields 
\begin{equation*}
y(x)=\frac{v(u-x)}{u}\quad (\qq\to 1)
\end{equation*}
which, as expected, is nothing but the equation of the straight line passing trough $(u,0)$ and $(0,v)$, irrespectively of $\gamma$. 

\section{Tangent Method and Arctic curve I}
\label{sec:arcticone}

\subsection{Model partition function, one-point function and single free path partition function}
\label{sec:gen}
\begin{figure}
\begin{center}
\includegraphics[width=10cm]{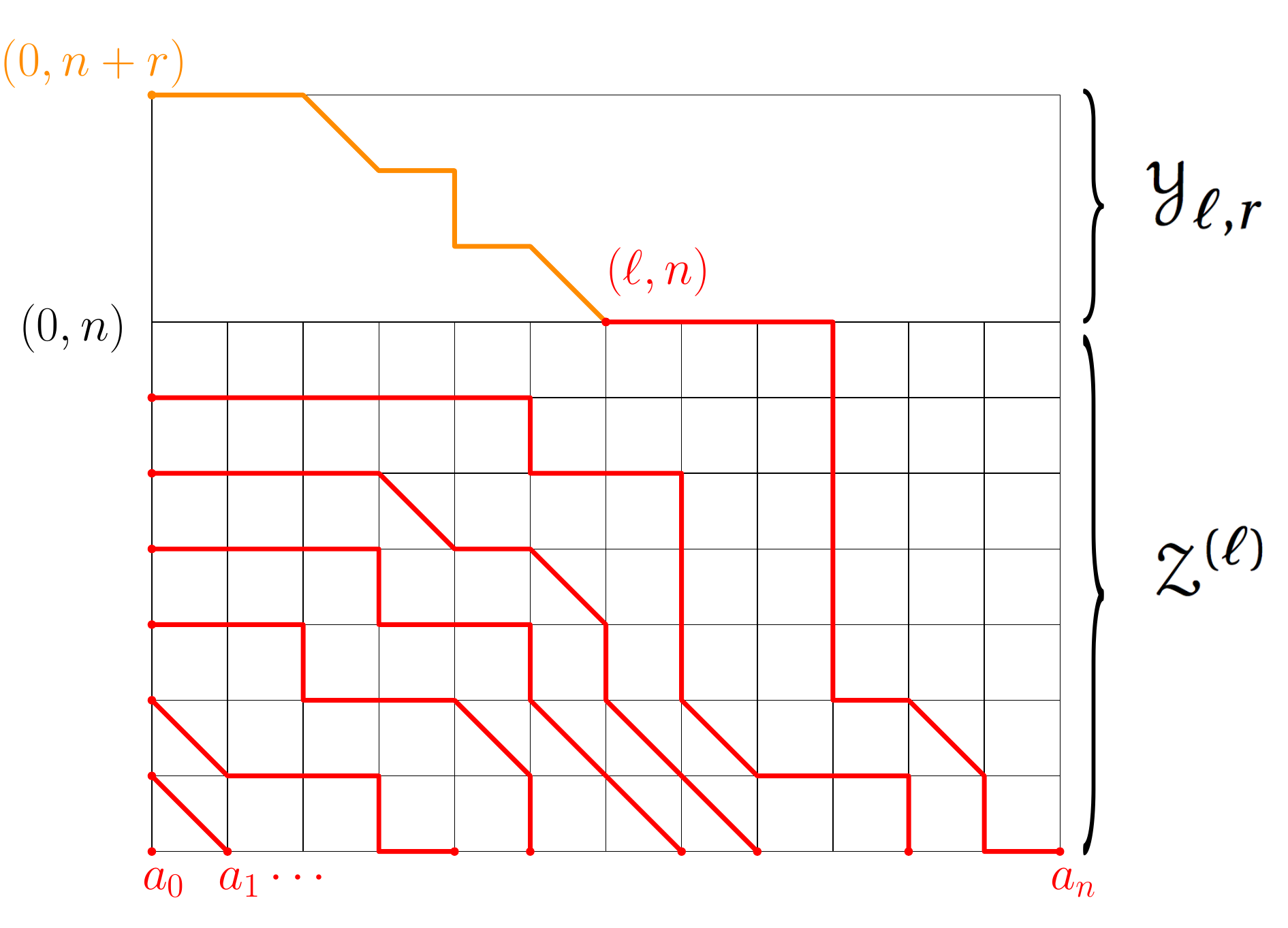}
\end{center}
\caption{\small{ A NILP configuration with a modified endpoint $(0,n+r)$ for the outermost path. This path must exit the originally allowed domain (materialized by the square grid) 
at some exit point $(\ell,n)$ by a diagonal or vertical step. The partition function of such a modified NILP with a fixed $\ell$ is the product of the
partition function ${\mathcal Z}^{(\ell)}$ of the NILP where the outmost path ends at the exit point $(\ell,n)$, times the partition function  ${\mathcal Y}_{\ell,r}$ of a single (non interacting)
Schr\"oder path from the exit point $(\ell,n)$ to the endpoint $(0,n+r)$.}
}
\label{fig:exitpoint}
\end{figure}

The aim of this section is to set the stage for applying the so-called Tangent Method \cite{COSPO}
to determining the arctic
curve of our tiling model of the Aztec rectangle with defects.  This involves evaluating a modified partition function
${\tilde {\mathcal Z}}$ obtained from $\mathcal Z$ by slightly modifying its path setting, namely by moving the endpoint
of the outermost path from its original position $(0,n)$ to a further position along the $y$-axis, say to a point $(0,n+r)$, $r>0$ (see Fig.~\ref{fig:exitpoint}). The general
idea is that this outermost path will follow asymptotically the arctic curve (induced by the interaction with the other paths 
in the non-intersecting family) before escaping tangentially along a geodesic until it reaches the new endpoint. Indeed, once it escapes, the path is no longer sensitive to the presence of the other paths of the NILP, and its trajectory follows 
a free geodesic. By a variational principle at large $n$, we shall determine the most likely exit point $(\ell,n)$ from the original
Aztec rectangle as a function of the vertical shift $r$, and the tangent geodesic will be determined as the unique geodesic through this most likely exit point and the modified endpoint. 
By moving the modified endpoint along the $y$-axis, we generate a family of tangent geodesics, from which the arctic curve is finally recovered as their envelope.

To carry out this program, we need to split the new partition function ${\tilde {\mathcal Z}}$ into a sum:
$${\tilde {\mathcal Z}} =\sum_{\ell=0}^{m} {\mathcal Z}^{(\ell)} \, {\mathcal Y}_{\ell,r} $$
according to the position of the exit point $(\ell,n)$ of the modified outermost path from the original Aztec rectangle.
Here we introduce two new partition functions: the first one, ${\mathcal Z}^{(\ell)}$, is the partition function of weighted
NILP configurations of $(n+1)$ paths with starting points $(a_i,0)$, $i=0,1,...,n$ and with endpoints $(0,j)$, $j=0,1,...,n-1$ for the first $n$ paths, while the outermost path \emph{stops at 
the exit point} $(\ell,n)$. Note that, in ${\mathcal Z}^{(\ell)}$, the contribution of the outermost path is computed using the area \emph{to its left} and not that below it (which is in general smaller).  
For simplicity, it turns out to be easier to consider the so called \emph{``boundary one-point function"} 
${\mathcal H}^{(\ell)}= {\mathcal Z}^{(\ell)}/{\mathcal Z}$. This function will be computed in the next section.

The second partition 
${\mathcal Y}_{\ell,r}$ entering the above decomposition of ${\tilde {\mathcal Z}}$ is the partition function for the last part of the outermost path from the exit point $(\ell,n)$ to the endpoint
$(0,n+r)$. As this portion involves a (free) single weighted Schr\"oder path, it is easily computed as follows. First notice that as the path must exit from the original rectangle, it must start with 
either a vertical or a diagonal step, which makes it start in practice at either position 
$(\ell,n+1)$ or $(\ell-1,n+1)$ after this first step (assuming $\ell\geq 1$ for convenience). In terms of single path partition functions such as in \eqref{defZij}, this gives (with obvious notations):
\begin{equation*}
{\mathcal Y}_{\ell,r}=q^{2\ell}\, Z_{(\ell,n+1)\to (0,n+r)} +\gamma \, q^{2\ell-1}\, Z_{(\ell-1,n+1)\to (0,n+r)}
\end{equation*}
where the $q$-dependent prefactors restore the correct area weights by taking into account the area of the strip of height 1 on the left of the first step.
Finally, by moving the origin to the point $(0,n+1)$, we may write equivalently:
\begin{equation}\label{Yvalue}
{\mathcal Y}_{\ell,r}=q^{2\ell}\, Z_{(\ell,0)\to (0,r-1)} +\gamma \, q^{2\ell-1}\, Z_{(\ell-1,0)\to (0,r-1)}=
q^{2\ell}\, z_{r-1}(q^{2\ell}) +\gamma \, q^{2\ell-1}\, z_{r-1}(q^{2\ell-2})
\end{equation}
solely in terms of the partition functions \eqref{defZij} and \eqref{eq:polz}.

\subsection{LU decomposition and integral formulas}
\label{sec:lu}

Similarly to ${\mathcal Z}$, the partition function ${\mathcal Z}^{(\ell)}$ is expressed through the
Lindstr\"om Gessel-Viennot formula:
\begin{equation*}
{\mathcal Z}^{(\ell)}=\det(A^{(\ell)}),\qquad A_{i,j}^{(\ell)}=\left\{ \begin{matrix} 
A_{i,j} & {\rm for} \ j<n \\
Z_{(a_i,0)\to (n,\ell)} & {\rm for}\ j=n 
\end{matrix}\right.
\end{equation*}
where $A$ is as in \eqref{Adef}. Note that $A^{(\ell)}$ differs from $A$ only in its last column.
We shall now use the LU decomposition method \cite{DFLAP,DFGUI,DFG2} to compute ${\mathcal Z}^{(\ell)}$:
assume we have written $A=LU$ as the product of a lower uni-triangular matrix $L$ and an upper triangular matrix $U$,
then of course $\det(A)=\prod_{i=0}^n U_{i,i}$. Since $U=L^{-1} A$ is upper triangular and $A$ and $A^{(\ell)}$ differ only in their last column,
the matrix $U^{(\ell)}:=L^{-1} A^{(\ell)}$ is upper triangular as well. Finally, in terms of $U$ and $U^{(\ell)}$, the one-point function 
${\mathcal H}^{(\ell)}$ reads simply 
\begin{equation}\label{Hvalue}
{\mathcal H}^{(\ell)} =\frac{\det(A^{(\ell)})}{\det(A)}=\frac{U^{(\ell)}_{n,n}}{U_{n,n}}\ ,
\end{equation}
so that only the elements of $U$ and $U^{(\ell)}$ with highest indices are in practice required.

\medskip
We have the following:

\begin{lemma}\label{Llem}
The unitriangular $(n+1)\times (n+1)$ matrix $L^{-1}$ with entries
\begin{equation*}
(L^{-1})_{i,j}:= \left\{ 
\begin{matrix}
\displaystyle{\frac{\prod\limits_{s=0}^{i-1} q^{2a_i}-q^{2a_s}}{\prod\limits_{s=0\atop s\neq j}^i q^{2a_j}-q^{2a_s} }} & \quad i\geq j\\
0 & {\rm otherwise}
\end{matrix}\right.
\end{equation*}
is such that $U:=L^{-1}A$ is upper triangular.
\end{lemma}
\begin{proof}
We compute:
\begin{eqnarray}
U_{i,j}&=&\sum_{k=0}^n (L^{-1})_{i,k} A_{k,j} = \sum_{k=0}^i 
\frac{\prod\limits_{s=0}^{i-1} q^{2a_i}-q^{2a_s}}{\prod\limits_{s=0\atop s\neq k}^i q^{2a_k}-q^{2a_s} }\,
Z_{(a_k,0)\to (0,j)}\nonumber \\
&=&\oint_{\mathcal C(a_0,...,a_i)} \frac{dt}{2{\rm i}\pi} \frac{\prod\limits_{s=0}^{i-1} q^{2a_i}-q^{2a_s}}{\prod\limits_{s=0}^i t-q^{2a_s} }\,
z_j(t)\ ,\label{Uelem}
\end{eqnarray}
where we realized the sum over $k$ as a sum of residues at $t=q^{2a_k}$ for the corresponding contour integral 
along a contour $\mathcal C(a_0,...,a_i)$ 
of the complex plane encircling all the points $q^{2a_s}$, $s=0,1,...,i$. Finally we have identified the last term
in the integrand in terms of the polynomial $z_j$ defined in \eqref{eq:polz}. Let us now show that $U$ is upper triangular.

Assume $i>j$. Note that since $\mathcal C(a_0,...,a_i)$ encircles all the finite poles of the integrand, the residue integral may be expressed as minus the contribution of the pole at infinity.
Using Theorem~\ref{polthm} which states that $z_j(t)$ is a polynomial of degree $j$ in $t$, we get the large $t$ asymptotics:
\begin{equation} \frac{\prod\limits_{s=0}^{i-1} q^{2a_i}-q^{2a_s}}{\prod\limits_{s=0}^i t-q^{2a_s} }\,
z_j(t) \propto t^{j-i-1}\label{asymptoz}
\end{equation}
and as $i>j$, there is no residue at $t=\infty$. We conclude that $U_{i,j}=0$ when $i>j$ and the Lemma follows.
\end{proof}

\begin{thm}\label{partitionfunction}
The partition function ${\mathcal Z}$ \eqref{defZ} for the domino tiling of the Aztec rectangle with defects reads:
\begin{equation*}
{\mathcal Z}=q^{\frac{n(n+1)(2n+1)}{6}} \, \prod_{s=0}^{n-1} (\gamma+q^{2s+1})^{n-s}\, 
\frac{\Delta(q^{2a_0},q^{2 a_1},...,q^{2a_n})}{\Delta(q^0,q^2,...,q^{2n})} 
\end{equation*}
where $\Delta(x_0,x_1,...,x_n)$ stands for the Vandermonde determinant $\Delta(x_0,x_1,...,x_n)=\prod_{i>j} (x_i-x_j)$.
\end{thm}
\begin{proof}
Recall that ${\mathcal Z}=\det(A)=\prod_{i=0}^n U_{i,i}$ with $U$ as in Lemma \ref{Llem}.
The result of the Theorem follows from the fact that:
\begin{equation}\label{Uii}
U_{i,i}=\prod_{s=0}^{i-1} q^{2s+1}(\gamma+q^{2s+1})\, \frac{q^{2a_i}-q^{2a_s}}{q^{2i}-q^{2s}} \ .
\end{equation}
To show this, let us use the contour integral formula \eqref{Uelem} for $j=i$, and express the result in terms 
of the residue at $t=\infty$. The latter has a non-vanishing contribution, as readily seen from \eqref{asymptoz}
for $j=i$. More precisely, using the explicit formula \eqref{eq:polz} for $z_j(t)$ we find:
\begin{eqnarray*}
U_{i,i}&=&-{\rm Res}_{t\to\infty} \, \frac{z_i(t)}{t^{i+1}}\prod_{s=0}^{i-1} (q^{2a_i}-q^{2a_s})=\sum_{k=0}^{i} \gamma^k\, q^{k^2}\ \frac{\prod\limits_{s=0}^{i-1} q^{2a_i}-q^{2a_s}}{\prod\limits_{s=1}^{i-k}(q^{2s}-1)\prod\limits_{s=1}^{k}(q^{2s}-1)}\prod_{s=1}^i q^{2(s-k)}
 \\
&=&\prod_{s=0}^{i-1}q^{2s}\frac{q^{2a_i}-q^{2a_s}}{q^{2i}-q^{2s}} \sum_{k=0}^{i} \gamma^k\, q^{k^2}\  {i\brack k}_{q^2}
\prod_{s=1}^i q^{2(s-k)}=
\prod_{s=0}^{i-1}q^{2s}\frac{q^{2a_i}-q^{2a_s}}{q^{2i}-q^{2s}} \sum_{k=0}^{i} \gamma^{i-k}\, q^{i+k^2}\  {i\brack k}_{q^2}\\
&=&\prod_{s=0}^{i-1} q^{2s+1}(\gamma+q^{2s+1})\, \frac{q^{2a_i}-q^{2a_s}}{q^{2i}-q^{2s}}\ , 
\end{eqnarray*}
where we have first reexpressed the summand in terms of the $q^2$-binomial ${i\brack k}_{q^2}={i\brack k,i-k,0}_{q^2}$,
and then performed a change of summation $k\to i-k$. Finally, we have used the product formula\footnote{This relation is easily proved by recursion from the identity
${i\brack k}_{q^2}={i-1\brack k}_{q^2}+{i-1\brack k-1}_{q^2}q^{2(i-k)}$.}:
$$\sum_{k=0}^{i} \gamma^{i-k}\, q^{k^2}\  {i\brack k}_{q^2} =\prod_{s=0}^{i-1} (\gamma+q^{2s+1})\ .
$$
\end{proof}

Using the relation \eqref{Hvalue}, we now get an explicit formula for ${\mathcal H}^{(\ell)}$.
\begin{thm}\label{onepointfunctionthm}
The boundary one-point function ${\mathcal H}^{(\ell)}$ for the NILP with an outermost path exiting at point $(\ell,n)$ reads:
\begin{equation}
{\mathcal H}^{(\ell)}=q^{n(2\ell-n)} \prod_{s=0}^{n-1} \frac{q^{2n}-q^{2s}}{\gamma+q^{2s+1}} \,\oint_{\mathcal C_\ell(a_0,a_1,...,a_n)} \frac{dt}{2{\rm i}\pi} 
\frac{z_n(t q^{-2\ell})}{\prod\limits_{s=0}^n t-q^{2a_s} }
\label{eq:hlexp}
\end{equation}
where the contour of integration $\mathcal{C}_\ell(a_0,a_1,...,a_n)$ encircles only the points $q^{2a_s}$ such that $a_s\geq \ell$.
\end{thm}
\begin{proof}
Similarly to the computation of $U_{i,i}$, we obtain:
\begin{eqnarray*}
U_{n,n}^{(\ell)}&=& \sum_{k=0}^n (L^{-1})_{n,k}\, A^{(\ell)}_{k,n} =\sum_{k=0}^n \frac{\prod\limits_{s=0}^{n-1} q^{2a_n}-q^{2a_s}}{\prod\limits_{s=0\atop s\neq k}^n q^{2a_k}-q^{2a_s} }\, q^{2n\ell}\, \times \left\{ 
\begin{matrix} z_n(q^{2a_k-2\ell}) & {\rm if} \ a_k \geq \ell \\
0 & {\rm otherwise} 
\end{matrix} \right.
\\
&=&q^{2n\ell}\, \prod_{s=0}^{n-1} (q^{2a_n}-q^{2a_s}) \,\oint_{\mathcal C_\ell(a_1,a_2,...,a_n)} \frac{dt}{2{\rm i}\pi} 
\frac{z_n(t q^{-2\ell})}{\prod\limits_{s=0}^n t-q^{2a_s}} \ .
\end{eqnarray*}
Here we have first rewritten $A^{(\ell)}_{k,n}=Z_{(a_k,0)\to (\ell,n)}=q^{2n\ell} Z_{(a_k-\ell,0)\to (0,n)}$
by moving the origin to the position $(\ell,0)$ and correcting the area
factor by $q^{2n\ell}$, and then used our general expression $Z_{(a_k-\ell,0)\to (0,n)}= z_n(q^{2a_k-2\ell})$ in terms of the polynomial $z_n(t)$ of Eq.~\eqref{eq:polz}.
This identity is valid only for $a_k\geq \ell$ while $Z_{(a_k-\ell,0)\to (0,n)}=0$ for $a_k<\ell$ (since a Schr\"oder path cannot move toward east).
This condition $a_k\geq \ell$ is automatically fulfilled in the contour integral by the choice of integration contour $\mathcal{C}_\ell(a_0,a_1,...,a_n)$, encircling 
only those points $q^{2a_s}$ such that $a_s\geq \ell$. The Theorem follows by dividing the above expression by that for $U_{n,n}$ from \eqref{Uii} for $i=n$.
\end{proof}

\subsection{Scaling limit}
\label{sec:saddle}

We may now proceed with the second step of the Tangent Method, by deriving asymptotic estimates
for the quantities ${\mathcal Y}_{\ell,r}$ and ${\mathcal H}^{(\ell)}$. Following the same principle as in Refs.~\cite{DFGUI,DFG2},
we consider a scaling limit of large $n$, in which the positions $\{a_k\}_{0\leq k\leq n}$ of starting points of the paths in the NILP are asymptotically
distributed according to a piecewise differentiable function $\al(\sigma)$ via
\begin{equation*}
a_k=\left\lfloor n\, \al\left(\frac{k}{n}\right)\right\rfloor\ .
\end{equation*}
The function $\sigma \mapsto \al(\sigma)$, $0\le \sigma \leq 1$ is strictly increasing and such that $\al'(\sigma)\geq 1$ whenever defined (to ensure $a_{k+1}-a_k\geq 1$). Note also that 
$\al(0)=0$ and $\al(1)=\lim_{n\to\infty} a_n/n =\mu$ if we let $m=n\, \mu$ for some finite $\mu$ in the scaling limit.
Moreover, we introduce the following scaling variables for the various integers entering the formulas for 
${\mathcal Y}_{\ell,r}$ and ${\mathcal H}^{(\ell)}$:
\begin{equation*}
\ell=n \, \xi,\ \ r=n\, \rho, \ \  q=\qq^{\frac {1}{n}}\ .
\end{equation*}

Using the expression \eqref{Yvalue} for ${\mathcal Y}_{\ell,r}$ and the explicit asymptotic formula \eqref{eq:Zijscal} for $Z_{(i,0)\to (0,j)}$ with $j=r-1\sim n\, \rho$ 
and $i=\ell,\ell-1\sim n\, \xi$, we get immediately:

\begin{lemma}\label{lemY}
In the scaling limit $n\to \infty$, the quantity\ \,${\mathcal Y}_{\ell,r}$ has the leading exponential behavior
\begin{equation*}
{\mathcal Y}_{n \xi,n\rho}\underset{n \to \infty}{\sim} {\rm e}^{n S_0(\xi,\rho)}
\end{equation*}
where 
$S_0(\xi,\rho)=S_0(\xi,\rho,\phi(\xi,\rho))$ and
\begin{equation*}
\begin{split}
S_0(\xi,\rho,\phi)& =\phi {\rm Log}\, \gamma +\phi^2\, {\rm Log}\, \qq +\int_{0}^{\xi+\rho-\phi}d\sigma\,  {\rm Log}(\qq^{2\sigma}-1)
- \int_{0}^{\rho-\phi}d\sigma\,  {\rm Log}(\qq^{2\sigma}-1)\\
  & \qquad \qquad \qquad \qquad \qquad-\int_{0}^{\phi}d\sigma\,  {\rm Log}(\qq^{2\sigma}-1)
-\int_{0}^{\xi-\phi}d\sigma\,  {\rm Log}(\qq^{2\sigma}-1)\\
\end{split}
\end{equation*}
while $\phi(\xi,\rho)$ is the 
unique solution to $\frac{\partial{S_0(\xi,\rho,\phi)}}{\partial\phi}=0$ satisfying $0\leq \phi\leq {\rm Min}(\xi,\rho)$, namely:
\begin{equation*}
\gamma\, \qq^{2\phi} \frac{(\qq^{2(\xi-\phi)}-1)(\qq^{2(\rho-\phi)}-1)}{(\qq^{2 \phi}-1)(\qq^{2(\xi+\rho-\phi)}-1)}=1 \qquad (0\leq \phi\leq {\rm Min}(\xi,\rho) ).
\end{equation*}
\end{lemma}

Similarly, using the expressions \eqref{eq:hlexp} for ${\mathcal H}^{(\ell)}$ and \eqref{eq:polz} for 
$z_j(t)$, with $k\sim n\, \kappa$, we obtain:
\begin{lemma}\label{lemH}
In the scaling limit $n\to \infty$, the quantity ${\mathcal H}^{(\ell)}$ has the leading exponential behavior
\begin{equation*}
{\mathcal H}^{(n\xi)}\underset{n \to \infty}{\sim} \int_{{\mathcal C}_\xi} \frac{dt}{2{\rm i}\pi} {\rm e}^{n S_1(t,\xi)}
\end{equation*}
where $S_1(t,\xi)=S_1(t,\xi,\kappa(t,\xi))$ and
\begin{eqnarray}
S_1(t,\xi,\kappa)&:=&2\xi\,{\rm Log}\, \qq+\kappa\,{\rm Log}\,\gamma +\kappa^2\,{\rm Log}\,\qq+\int_0^{1}
d\sigma {\rm Log}(t\qq^{2(\sigma-\kappa-\xi)}-1) \nonumber \\
&&\quad -\int_0^{1-\kappa}
d\sigma {\rm Log}(\qq^{2\sigma}-1)-\int_0^{\kappa}
d\sigma {\rm Log}(\qq^{2\sigma}-1) \nonumber \\
&&\quad -\int_0^1d\sigma {\rm Log}(t-\qq^{2\al(\sigma)}) +C \nonumber
\end{eqnarray}
with $C$ some unimportant $\qq$-dependent constant, while $\kappa(t,\xi)$ is the 
unique solution to $\frac{\partial{S_1(t,\xi,\kappa)}}{\partial\kappa}=0$ satisfying $0\leq \kappa \leq 1$, namely:
\begin{equation*}
\gamma \, \qq^{2\kappa}\frac{(\qq^{2(1-\kappa)}-1)(\qq^{2(\kappa+\xi)}-t)}{(\qq^{2\kappa}-1)(\qq^{2(\kappa+\xi)}-\qq^2\, t)}=1 \ .
\end{equation*}
The contour ${\mathcal C}_\xi$ is the scaling limit of the contour ${\mathcal C}_\ell(a_0,a_1,...,a_n)$ and therefore
encircles only the values $\qq^{2\al}$ with $\al\in [\xi,\al(1)]$. In particular, since $\al(1)=\mu$, the contour may be chosen so as to cross the real axis at $\qq^{2\xi}$ and anywhere 
above $\qq^{2\mu}$ if $\qq >1$ (respectively anywhere below $\qq^{2\mu}$ if $\qq<1$). 
\end{lemma}

\subsection{Most likely exit point}
We are now ready for the third stage of the Tangent Method: find the scaling limit of the most likely exit point 
of the outermost path in our NILP. This is determined by finding the value of $\ell$ that maximizes the contribution
to the sum:
\begin{equation*}
\frac{\tilde{\mathcal Z}}{\mathcal Z}=\sum_{\ell=0}^m {\mathcal H}^{(\ell)}\, {\mathcal Y}_{\ell,r}\ .
\end{equation*}
In the scaling limit, this gives to leading exponential order in $n$:
\begin{equation}\label{oneptsapo}
\frac{\tilde{\mathcal Z}}{\mathcal Z}\underset{n \to \infty}{\sim} \int_0^\mu d\xi \oint_{C_\xi} \frac{dt}{2{\rm i}\pi} {\rm e}^{n (S_0(\xi,\rho)+S_1(t,\xi))}\ .
\end{equation}

\begin{thm}
\label{thm:likely}
The most likely exit point of the outermost path in our NILP is given in the scaling limit by $(n\xi,n)$ where 
$(\xi,\kappa,\phi,t)$ is the unique solution satisfying $\xi\in [0,\mu]$, $\phi\in [0,{\rm Min}(\xi,\rho)]$,
 $\kappa\in [0,1]$ and $t\geq \qq^{2\mu}$ if $\qq>1$ (respectively $t\leq\qq^{2\mu}$ if $\qq<1$) to the system:
\begin{eqnarray}
\gamma\, \qq^{2\phi} \frac{(\qq^{2(\xi-\phi)}-1)(\qq^{2(\rho-\phi)}-1)}{(\qq^{2 \phi}-1)(\qq^{2(\xi+\rho-\phi)}-1)}&=&1 \label{oneH}\\
\gamma \, \qq^{2\kappa}\frac{(\qq^{2(1-\kappa)}-1)(\qq^{2(\kappa+\xi)}-t)}{(\qq^{2\kappa}-1)(\qq^{2(\kappa+\xi)}-\qq^2\, t)}&=&1\label{twoH}\\
\frac{\qq^{2(\kappa+\xi)}-\qq^2\, t}{\qq^{2(\kappa+\xi)}- t} \, \xx(t)&=&1\label{threeH}\\
\qq^2\, \frac{(\qq^{2(\xi+\rho-\phi)}-1)(\qq^{2(\kappa+\xi)}- t)}{(\qq^{2(\xi-\phi)}-1)(\qq^{2(\kappa+\xi)}-\qq^2\, t)}&=&1 \label{fourH}
\end{eqnarray}
where\ \, $\xx(t)$ is the $\qq$-exponential moment-generating function for the distribution $\al$:
\begin{equation}
\xx(t):=\qq^{-2t\textstyle{ \int_0^1 \frac{d\sigma}{t-\qq^{2\al(\sigma)}}}}\ .
\label{defxt}
\end{equation}
\end{thm}
\begin{proof}
The leading contribution to the integral formula \eqref{oneptsapo} maximizes 
$S(t,\xi,\rho):=S_0(\xi,\rho)+S_1(t,\xi)$
and must satisfy $\frac{\partial S}{\partial t}=\frac{\partial S}{\partial \xi}=0$. Using  $\frac{\partial S}{\partial t}=\frac{\partial S_1(t,\xi)}{\partial t}=\frac{\partial S_1(t,\xi,\kappa)}{\partial t}$ at
$\kappa=\kappa(t,\xi)$ (since at this point $\frac{\partial S_1(t,\xi,\kappa)}{\partial \kappa}=0$), and similarly $\frac{\partial S}{\partial \xi}=\frac{\partial S_0(\xi,\rho,\phi)}{\partial \xi}+\frac{\partial S_1(t,\xi,\kappa)}{\partial \xi}$ at $\phi=\phi(\xi,\rho)$ and 
$\kappa=\kappa(t,\xi)$, the above two extremization conditions lead directly to \eqref{threeH} and
\eqref{fourH} respectively, while \eqref{oneH} and \eqref{twoH} determine $\phi=\phi(\xi,\rho)$ and 
$\kappa=\kappa(t,\xi)$ as in Lemmas \ref{lemY} and \ref{lemH}.
To investigate the existence and uniqueness of the solution, let us use variables:
$$F=\qq^{2\phi},\ \ K=\qq^{2\kappa},\ \ R=\qq^{2\rho},\ \ L=\qq^{2\xi}$$
and rewrite the above system as:
\begin{eqnarray*}
&&\gamma\, (L-F)(R-F)=(F-1)(LR-F), \quad \gamma\, (\qq^2-K)(KL-t)=(K-1)(KL-\qq^2\, t)\\
&&(KL-\qq^2\, t)\, \xx(t)=(KL-t),\quad \qq^2\, (LR-F)(KL-t)=(L-F)(KL-\qq^2\, t)
\end{eqnarray*}
whose solution may be written as
\begin{equation}
\begin{split}
K&=K(t)= \frac{1+\gamma\, \qq^2 \xx(t)}{1+\gamma\, \xx(t)}\ ,\\
F&=F(t)=\frac{(1+\gamma) t\,\big(1+\gamma\, \xx(t)\big)}{\big(1+\gamma\, \qq^2\xx(t)\big)\big(t\, \big(1+\gamma\, \xx(t)\big)+\gamma\, \big(1- \xx(t)\big)\big)}\ ,\\
L&=L(t)=\frac{t\, \big(1-\qq^2 \xx(t)\big)\big(1+\gamma\, \xx(t)\big)}{\big(1-\xx(t)\big)\big(1+\gamma\, \qq^2 \xx(t)\big)}\ ,\\
R&=R(t)=\frac{t\, \big(1+\gamma\, \xx(t)\big)-\big(1- \xx(t)\big)}{\qq^2 \xx(t)\, \big(t\, \big(1+\gamma\, \xx(t)\big)+\gamma\, \big(1- \xx(t)\big)\big)}\ ,\\
\end{split}
\label{KFLR}
\end{equation}
which yields the most likely value of $\xi$ as a function of $\rho$, or equivalently the most likely value of $L$ as a function of $R$ in
the \emph{parametric form} $L=L(t)$, $R=R(t)$ with a varying $t$. The value of $t$ must be real (for $K,F,L,R$ to be real) and must lie on the contour $\mathcal{C}_\xi$,
which implies\footnote{The equation for $L=L(t)$ forbids $t=\qq^{2\xi}=L$ unless $\xx(t)=0$, which corresponds to the degenerate situation where $t=\qq^{2\mu}$ and $\xi=\mu$.}  $t\geq\qq^{2\mu}$ if $\qq>1$ (respectively $t\leq\qq^{2\mu}$ if $\qq<1$). All the values of $t$ in this range are valid and lead to values of
$L$, $R$, $K$, and $F$ in their respective allowed range. Assuming for instance $\qq>1$, we must have $1\leq K\leq \qq^2$, $1\leq F\leq \min(L,R)$, $1\leq L\leq \qq^{2\mu}$
and $R\geq 1$. The inequalities for $K$ follow from the fact that $\xx(t)$ increases from $0$ to $\qq^{-2}$ when $t$ increases from $\qq^{2\mu}$ to $+\infty$.
As for the other inequalities, it is easily checked from the above expressions that $R-1$, $L-F$, $R-F$ and $F-1$ all have the same sign as $L-1$. 
Indeed we have the ratios 
\begin{equation}
\frac{F\!-\!1}{R\!-\!F}=\qq^2\gamma\, \xx(t), \quad \frac{F\!-\!1}{L\!-\!F}=\frac{\gamma\, \big(1\!-\!\xx(t)\big)}{t\, \big(1\!+\!\gamma\ \xx(t)\big)}, \quad 
\frac{R\!-\!1}{F\!-\!1}=1+\frac{1}{\gamma\, \qq^2\, \xx(t)},
\quad \frac{L\!-\!1}{F\!-\!1}=1+\frac{t\, \big(1\!+\!\gamma\, \xx(t)\big)}{\gamma\,  \big(1\!-\!\xx(t)\big)}\ ,
\label{ratios}
\end{equation}
all manifestly positive since $t\geq \qq^{2\mu}\geq 1$ and $0\leq \xx(t)\leq \qq^{-2}< 1$.
It is thus sufficient to prove $1\leq L\leq \qq^{2\mu}$. From the definition of $\xx(t)$ and the fact that $\al(\sigma)\geq \sigma$, we deduce that the value of $\xx(t)$ for any 
acceptable distribution $\al(\sigma)$ is bounded from above by its value for 
$\al(\sigma)= \sigma$, namely $\xx(t)\leq (t\, \qq^{-2}-1)/(t-1)$, henceforth
\begin{equation*}
\begin{split}
& 1-\qq^2 \xx(t)\geq \frac{\qq^2-1}{t-1}\quad \Rightarrow \quad t\, \big(1-\qq^2 \xx(t)\big)\geq \qq^2\, \big(1-\xx(t)\big)\\
& \Rightarrow t\, \big(1-\qq^2 \xx(t)\big)\big(1+\gamma\, \xx(t)\big)\geq \big(1-\xx(t)\big)\big(\qq^2+\gamma\, \qq^2\, \xx(t)\big)\geq 
\big(1-\xx(t)\big)\big(1+\gamma\, \qq^2\, \xx(t)\big)\\
\end{split}
\end{equation*}
which implies $L\geq 1$, as wanted. Similarly, from the fact that $\al(1)-\al(\sigma)\geq 1-\sigma$, hence $\al(\sigma)\leq \mu+1-\sigma$, 
we deduce that the value of $\xx(t)$ for any acceptable $\al(\sigma)$ is bounded from below by its value for 
$\al(\sigma)=\mu+1-\sigma$, namely $\xx(t)\geq (t-\qq^{2\mu})/(t\, \qq^2-\qq^{2\mu})$, henceforth
\begin{equation*}
\begin{split}
& 1-\qq^2 \xx(t)\leq\qq^{2\mu} \frac{1-\qq^{-2}}{t-\qq^{2\mu-2}}\quad \Rightarrow \quad t\, \big(1-\qq^2 \xx(t)\big)\leq\qq^{2\mu}\, \big(1-\xx(t)\big)\\
& \Rightarrow t\, \big(1-\qq^2 \xx(t)\big)\big(1+\gamma\, \xx(t)\big)\leq\qq^{2\mu} \big(1-\xx(t)\big)\big(1+\gamma\, \xx(t)\big)\leq
\qq^{2\mu} \big(1-\xx(t)\big)\big(1+\gamma\, \qq^2\, \xx(t)\big)\\
\end{split}
\end{equation*}
which implies $L\leq\qq^{2\mu}$. A similar argument holds for $\qq<1$ and $t\leq \qq^{2\mu}$, implying now $ \qq^2 \leq K\leq 1$, $\max(L,R) \leq F\leq 1$, $ \qq^{2\mu}\leq L\leq 1$
and $R\leq 1$. For $\qq<1$, we have $t\leq \qq^{2\mu}$ and we must distinguish the case where $t\geq 0$ (in which case $0\leq \xx(t)\leq 1$ as $\xx(t)$ increases from $0$ to $1$
when $t$ decreases from $\qq^{2\mu}$ to $0$) and the case $t<0$ (in which case $\xx(t)>1$ as $\xx(t)$ increases from $1$ to $\qq^{-2}$ with increasing $|t|$). In particular, the quantity $t/(1-\xx(t))$
appearing in some of the ratios \eqref{ratios} above, and the combination $(1-\qq^2\, \xx(t))$ appearing in $L$ are always positive. This again proves that $\kappa$, $\phi$, $\xi$ and $\rho$ are in the announced range, provided now $t\leq \qq^{2\mu}$.

\end{proof}

\section{Tangent Method and Arctic curve II}
\label{sec:arctictwo}

\subsection{The family of tangent curves}
\label{sec:family}

The arctic curve for our NILP is tangent to the family of geodesics passing through the modified endpoints $(0,n (1+\rho))$ and the associated most likely exit points $(n\xi,n)$ for a varying $\rho$. Theorem~\ref{thm:likely} gives the relation 
between $\rho$ and $\xi$ in a parametric form with parameter $t$, leading to a family of tangent geodesics 
parametrized by $t$. 
We have the following:
\begin{thm}\label{thm:geofamily}
The family of geodesics through the points $(0,n (1+\rho(t)))$ (modified endpoint) and $(n\xi(t),n)$ (most likely exit point from the Aztec rectangle) reads, in cartesian coordinates $(x,y_t(x))$ in the scaling limit:
\begin{equation}\label{family}
\begin{split}
Y_t(X)&=\frac{1}{\xx(t)}\ \frac{t\, \big(1+\gamma\, \xx(t)\big)-\big(1-\xx(t)\big)\, X}{t\, \big(1+\gamma\, \xx(t)\big)+\gamma\, \big(1-\xx(t)\big)\,  X} \\
&\!\!\!\!\!\!\!\!\!\!\!\!\!\!\!\!\!\!\hbox{where}\ \ Y_t(X)=\qq^{2y_t(x)}\ , \ \ {and} \ \  X=\qq^{2x}
\end{split}
\end{equation}
for \, $\xx(t)$ as in \eqref{defxt}. We get the same expression for the tangent geodesics, 
whether $\qq>1$ or $\qq<1$: only the range of 
$t$ differs, with $t\in [\qq^{2\mu},\infty[$ for $\qq>1$ and $t\in ]-\infty, \qq^{2\mu}]$ for $\qq<1$. 
\end{thm}
\begin{proof}
The theorems \ref{thm:geod} and \ref{thm:likely} are all we need to construct our family of tangent curves. 
Let us move the origin to position $(0,n)$. The geodesic
passing through $(0,n (1+\rho))$ and $(n\xi,n)$ is obtained by taking the expression \eqref{eq:eqgeod} of Theorem~\ref{thm:geod} with $X\to \qq^{2x}=X$, $Y\to \qq^{2(y-1)}=Y \qq^{-2}$, 
$U\to \qq^{2\xi}=L$ and $V\to \qq^{2(1+\rho-1)}=\qq^{2\rho}=R$,
where $L$ and $R$ must be expressed as in \eqref{KFLR}, as derived in the proof of Theorem~\ref{thm:likely}.
Plugging these expressions into \eqref{eq:eqgeod}, we must use the value of $\delta(U=L,V=R)$ of \eqref{eq:geodtwo},
which satisfies:
\begin{equation*}
\delta^2(L,R)=(L-1)^2\left(\frac{t\, \big(1+\gamma \qq^2 \xx(t)\big)\big(1+\gamma\, \xx(t)\big)+\gamma\, \big(1-\xx(t)\big)\big(1- \qq^2 \xx(t)\big)}{
\qq^2\, \xx(t)\big(t\, \big(1+\gamma\, \xx(t)\big)+\gamma\, \big(1- \xx(t)\big)\big)}\right)^2
\end{equation*}
where the term in the second square is always positive (since in particular, as we have already seen, $t/(1-\xx(t))$ and $(1-\qq^2\, \xx(t))$ are positive).
Since $(L-1)$ has the sign $\epsilon(\qq)$ (i.e is positive iff $\qq>1$), we deduce
\begin{equation*}
\epsilon(\qq)\, \delta(L,R)=(L-1)\, \frac{t\, \big(1+\gamma \qq^2 \xx(t)\big)\big(1+\gamma\, \xx(t)\big)+\gamma\, \big(1-\xx(t)\big)\big(1- \qq^2 \xx(t)\big)}{
\qq^2\, \xx(t)\big(t\, \big(1+\gamma\, \xx(t)\big)+\gamma\, \big(1- \xx(t)\big)\big)}
\end{equation*}
and \eqref{family} follows from substituting this into \eqref{eq:eqgeod}, whereas the ranges of $t$ follow from the discussion in the proof of Theorem \ref{thm:likely}.
\end{proof}

\begin{remark}
\label{rmkfamily}
Let us consider the limit $\qq\to1$ of the result of Theorem \ref{thm:geofamily}.
Setting $t=\qq^{2\tau}$ and expanding the above equation at first order in $\qq-1$ yields the parametric family:
\begin{equation*}
\begin{split}
y_\tau(x)&=\frac{\big(1-\xx(\tau)\big)\big(1+\gamma\, \xx(\tau)\big)}{\xx(\tau)(1+\gamma)}\ (\tau-x)\\
\xx(\tau)&:= {\rm e}^{\textstyle{-\int_0^1d\sigma \frac{1}{\tau-\al(\sigma)}}} \\
\end{split}
\end{equation*}
with $\tau\in [\mu,+\infty[$, and with a slight abuse of notation $\xx(\tau)$ and $y_\tau(x)$ for the $\qq\to 1$ limits 
as well. We see that the tangent curves are straight lines in this limit.
\end{remark}

\begin{remark}
In the limit $\gamma\to 0$, the result of Theorem \ref{thm:geofamily} reduces to:
\begin{equation*}
Y_t(X)= \frac{t-\big(1-\xx(t)\big)\, X}{t\,\xx(t)} \\
\end{equation*}
with $\xx(t)$ as in \eqref{defxt}.
This matches the result of \cite{DFG2} (Section 5.1).
\end{remark}

\subsection{The Arctic curve and its properties}
\label{sec:arcurve}
The family of tangent curves \eqref{family} may alternatively be characterized by their equation in the $(X,Y)$ plane:
\begin{equation}
\mathcal{F}_t(X,Y):= \xx(t)\big(t\, \big(1+\gamma\, \xx(t)\big)+\gamma\, \big(1-\xx(t)\big)\,  X\big)Y-t\, \big(1+\gamma\, \xx(t)\big)+\big(1-\xx(t)\big)\, X =0\ .
\label{eq:Ft}
\end{equation}
The core of the Tangent Method is that the arctic curve is precisely the envelope of this family of tangent curves, hence it may be obtained as the solution of 
\begin{equation}\label{envelope}
\mathcal{F}_t(X,Y)=\frac{\partial}{\partial t}\mathcal{F}_t(X,Y)=0.
\end{equation} 
Solving these equations in $X$ and $Y$ yields the arctic curve in parametric form $(X(t),Y(t))$:
\begin{thm}
\label{thmac}
The arctic curve for the domino tiling of the Aztec rectangle with defects reads in parametric form:
\begin{equation*}
\begin{split}
X(t)&=\qq^{2x(t)}=t-\frac{(1+\gamma) \bigg(\xx(t)\, \big(1-\xx(t)\big)\, \big(1+\gamma\, \xx(t)\big)+t\ \xx'(t)\, \big(1+\gamma\  \xx(t)^2\big)
 \bigg)-\epsilon(\qq)\, \omega(t)}{2 \gamma \big(1-\xx(t)\big)^2 \xx'(t)}\\
 Y(t)&=\qq^{2y(t)}=\frac{\epsilon(\qq)\, \omega(t)-(1-\gamma)\, \xx(t)\, \big(1-\xx(t)\big) \big(1+\gamma\, \xx(t)\big)-(1+\gamma)\, t\ \xx'(t)\, \big(1-\gamma\, \xx(t)^2\big)}
 {2 \gamma\,  \xx(t)^2\, \bigg(\big(1-\xx(t)\big) \big(1+\gamma\, \xx(t)\big)+(1+\gamma)\,t\ \xx'(t)\bigg)}
 \\
\omega(t)&= \Bigg(
\bigg\{ (1+\gamma)\, \xx(t)\, \big(1-\xx(t)\big)\big(1+\gamma\, \xx(t)\big)+t\ \xx'(t)\, \big(\big(1+\gamma\xx(t)\big)^2
-\gamma\, \big(1-\xx(t)\big)^2\big)\bigg\}^2 \\
&
\quad\quad \quad\quad\qquad\qquad\qquad\qquad\qquad\qquad\qquad\qquad+4 
\gamma \bigg\{ t\  \xx'(t)\, \big(1+\gamma\, \xx(t)\big)\big(1-\xx(t)\big)\bigg\}^2
\Bigg)^{1/2}\\
   \end{split}
\end{equation*}
for \, $\xx(t)$ as in \eqref{defxt} and, as before, with \, $\epsilon(\qq)=1$ if \, $\qq>1$ and \, $\epsilon(\qq)=-1$ if \, $\qq<1$. 
\end{thm}
\begin{proof}
Eliminating $Y$ between the two equations of \eqref{envelope}, we obtain the quadratic equation:
\begin{eqnarray*}
&t^2\,\big(1+\gamma\, \xx(t)\big)^2\xx'(t) -X\, \bigg( (1+\gamma)\, \xx(t)\, \big(1-\xx(t)\big)\big(1+\gamma\, \xx(t)\big)
+t\ \xx'(t)\big(\big(1+\gamma\, \xx(t)\big)^2-\gamma\, \big(1-\xx(t)\big)^2\big)\bigg)\\
&\qquad\qquad\qquad\qquad\qquad\quad\qquad\qquad\qquad\qquad\qquad\qquad\qquad\qquad-X^2 \gamma\, \big(1-\xx(t)\big)^2\, \xx'(t)=0\ .
\end{eqnarray*}
Note that the constant and $X^2$ coefficients have opposite signs, hence we must pick the unique solution
with $X>0$. The corresponding manifestly positive discriminant reads:
\begin{eqnarray*}&\omega(t)^2=
\bigg\{ (1+\gamma)\, \xx(t)\, \big(1-\xx(t)\big)\big(1+\gamma\, \xx(t)\big)+t\ \xx'(t)\, \big(\big(1+\gamma\xx(t)\big)^2
-\gamma\, \big(1-\xx(t)\big)^2\big)\bigg\}^2\\
&\quad\quad \quad\quad\qquad\qquad\qquad\qquad\qquad\qquad\qquad\qquad+4 
\gamma \bigg\{ t\, \big(1+\gamma\, \xx(t)\big)\big(1-\xx(t)\big)\, \xx'(t)\bigg\}^2\ .
\end{eqnarray*}
To get a positive $X$, we must take in front of $\omega(t)$ the sign $\epsilon(\qq)=+1$ if  $\qq>1$ and  $\epsilon(\qq)=-1$ if  $\qq<1$
to ensure that the global prefactor $\epsilon(\qq)/(2\gamma(1-\xx(t))^2\xx'(t))$ of  $\omega(t)$ in $X(t)$ is positive, 
where $\xx'(t)={\rm Log}(\qq^2)\, \xx(t) \int_0^1 d\sigma \frac{\qq^{\al(\sigma)}}{(t-\qq^{\al(\sigma)})^2}$ has the sign of ${\rm Log}(\qq)$.

 \end{proof}
 
 \begin{remark}\label{qonerem}
Let us consider the limit $\qq\to1$ of the result of Theorem \ref{thmac}.
Setting as before $t=\qq^{2\tau}$ and expanding the above equation at first order in $\qq-1$ yields the parametric equation
of the arctic curve:
\begin{equation}\label{aczero}
\begin{split}
x(\tau)& =\tau-\frac{\xx(\tau)\, \big(1-\xx(\tau)\big)\big(1+\gamma\, \xx(\tau)\big)}{\big(1+\gamma\, \xx(\tau)^2\big)\, \xx'(\tau)}\\
y(\tau)&=\frac{\big(1-\xx(\tau)\big)^2\big(1+\gamma\, \xx(\tau)\big)^2}{(1+\gamma)\big(1+\gamma\, \xx(\tau)^2\big)\, \xx'(\tau)}\\
\end{split}
\end{equation}
with $\xx(\tau)$ as in Remark~\ref{rmkfamily}. This result is in agreement with the expression found in \cite{BuKni} (see Eq.~(5.3)) 
in some particular case (see Section~\ref{sec:summary} for a detailed discussion).
\end{remark}

\begin{remark}
In the limit $\gamma\to 0$, the result of Theorem \ref{thmac} reduces to:
\begin{equation*}
\begin{split}
X(t)&=\frac{t^2\ \xx'(t)}{\xx(t)\, \big(1-\xx(t)\big)+t\ \xx'(t)}\\
 Y(t)&=\frac{ \big(1-\xx(t)\big)+t\ \xx'(t)}{\xx(t)\, \big(1-\xx(t)\big)+t\ \xx'(t)}\\
\end{split}
\end{equation*}
with $\xx(t)$ as in \eqref{defxt}.
This matches the result of \cite{DFG2} (Theorem 1.1).
\end{remark}

\begin{figure}
\begin{center}
\includegraphics[width=8cm]{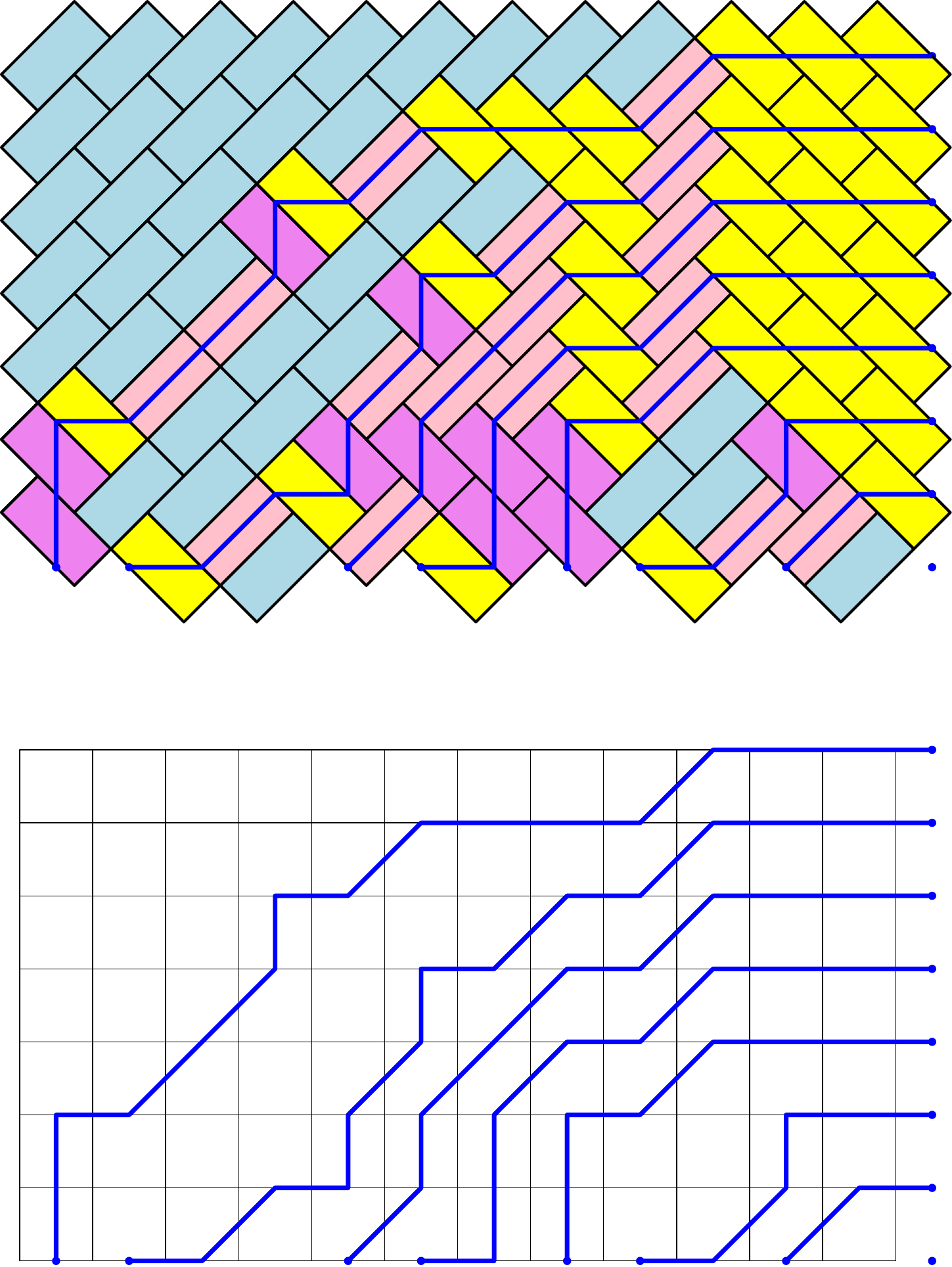}
\end{center}
\caption{\small{A alternative description of the tiling configuration of Fig.~\ref{fig:NILP} by a Non-intersecting Lattice Path configuration. The correspondence between 
elementary path steps and dominos is obtained from that of \eqref{dominos} by a left-right symmetry.}}
\label{fig:otherNILP}
\end{figure}
 
In Theorem \ref{thmac}, it was, so far, implicitly assumed that
$t$ runs over the range $[\qq^{2\mu},+\infty[$ for $\qq>1$ and over the range $]-\infty,\qq^{2\mu}]$ for $\qq<1$, with $\mu=\al(1)$.
This is indeed the domain in which the Tangent Method approach that we described is valid. 
As explained in \cite{DFGUI} and \cite{DFG2}, this limited range of $t$ gives only one portion of the arctic curve. To get all the portions,
we may consider other families of NILP which are (bijectively) equivalent to the original tiling problem and let the Tangent Method machinery act on these new 
NILP. For instance, the set of NILP described in Fig.~\ref{fig:otherNILP} provides an equivalent description of the tiling problem amenable to the Tangent Method approach.
As discussed in details in \cite{DFGUI} and \cite{DFG2}, using this new set of NILP allows to extend the range of validity of Theorem \ref{thmac} by allowing $t$ to now span the larger domain
$]-\infty,1]\cup[\qq^{2\mu},+\infty[$ if $\qq>1$ or $]-\infty,\qq^{2\mu}]\cup[1,+\infty[$ if $\qq<1$. More generally, the result of the previous studies in \cite{DFGUI}, \cite{DFG2} and \cite{DR}
is that \emph{the range of $t$ for which Theorem \ref{thmac} is valid is that for which for which $\xx(t)$ is real}. Letting $t$ vary in this range produces the entire arctic curve which
is in general made of several, possibly disconnected, portions (corresponding to disconnected allowed intervals for $t$). 
Recall that the asymptotic distribution $\al(\sigma)$ of path starting points is a strictly increasing function from $\al(0)=0$ to $\al(1)=\mu$ (with moreover $\al'(\sigma)\geq 1$ when defined). For a 
generic such distribution, the ranges of $t$ for which $\xx(t)$ is real is simply $]-\infty,\qq^{2\al(0)}]\cup[\qq^{2\al(1)},+\infty[$ for $\qq>1$ and $]-\infty,\qq^{2\al(1)}]\cup[\qq^{2\al(0)},+\infty[$ for $\qq<1$,
corresponding to the extended range just discussed. The corresponding arctic curve is thus made of two portions which meet at some limiting point, tangentially to 
the N-boundary and corresponding to $t\to \pm \infty$. Indeed, for $t\to \pm \infty$, $\xx(t)\to \qq^{-2}$ and $\xx'(t)\sim B\, \qq^{-2}/t^2$ with 
$B=2{\rm Log}(\qq)\, \int_0^1\qq^{2\al(\sigma)}d\sigma$,
and the expressions in Theorem~\ref{thmac} lead to $Y\to \qq^2$ (i.e. $y\to 1$) while $X\to B \frac{\qq^2(\gamma+\qq^2)}{(\qq^2-1)(1+\gamma)}$, and it is easily checked that 
the tangent curve has slope $0$ at this point. 

\subsection{Freezing boundaries}
\label{sec:freezing}
As discussed in detail in \cite{DFGUI}, \cite{DFG2} and \cite{DR}, extra domains where $\xx(t)$ is real may appear for particular distributions 
$\al(\sigma)$ corresponding to the following two situations, referred to as {\emph``freezing boundaries"} since the new created portions of arctic curve delimitate frozen domains
which are adjacent to the S boundary:
\begin{itemize}
\item{The function $\al(\sigma)$ presents some \emph{``gaps"} i.e.\ has discontinuities at some values $\sigma=\sigma_p\in [0,1]$ (i.e.\ $\al(\sigma_p^+)-\al(\sigma_p^-)=\delta_p>0$,
corresponding to a region of the S boundary of linear size $n\, \delta_p$ free of starting points, i.e.\ filled with defects).
Then the function $\xx(t)$, as given by \eqref{defxt}, remains well defined and real positive
for $t\in [\qq^{2\al(\sigma_p^-)},\qq^{2\al(\sigma_p^+)}]$ if $\qq>1$ 
(respectively $t\in [\qq^{2\al(\sigma_p^+)},\qq^{2\al(\sigma_p^-)}]$ if $\qq<1$). 
This may be seen by writing
\begin{equation}
\xx(t)=\qq^{-2t\textstyle{ \int_0^{\sigma_p^-} \frac{d\sigma}{t-q^{2\al(\sigma)}}}}\ \qq^{-2t\textstyle{ \int_{\sigma_p^+}^1 \frac{d\sigma}{t-q^{2\al(\sigma)}}}}
\label{eq:xtgap}
\end{equation}
}with, assuming say $\qq>1$,  a first and second factor well defined and real positive\footnote{Here we allow for convenience the limiting possibility $\xx(t)\to +\infty$.} 
respectively for $t\in]-\infty,1]\cup[\qq^{2\al(\sigma_p^-)},+\infty[$ and
for $t\in]-\infty,\qq^{\al(\sigma_p^+)}]\cup[\qq^{2\al(1)},+\infty[$ so that $\xx(t)$ is well defined and real positive for $t\in]-\infty,1]\cup[\qq^{2\al(\sigma_p^-)},\qq^{2\al(\sigma_p^+)}]\cup[\qq^{2\al(1)},+\infty[$.
The new (middle) interval of $t$ then creates a new portion of arctic curve via the parametric expression of Theorem~\ref{thmac}.

\item{The function $\al(\sigma)$ presents some \emph{``minimal slope"} intervals, i.e.\ satisfies $\al'(\sigma)=1$ for $\sigma\in [\tau_p,\tau_{p+1}]$ for some $\tau_p$'s in $[0,1]$
(corresponding to a linear portion of the S boundary of length $n(\tau_{p+1}-\tau_p)$ without defect). Then the function $\xx(t)$ 
remains well defined (by analytic continuation) and is now real \emph{negative} for $t\in [\qq^{2\al(\tau_p)},\qq^{2\al(\tau_{p+1})}]$ if $\qq>1$ (respectively 
$t\in [\qq^{2\al(\tau_{p+1})},\qq^{2\al(\tau_p)}]$ if $\qq<1$).
This may be seen by writing
\begin{equation}
\begin{split}
\xx(t)&=\qq^{-2t\textstyle{ \int_0^{\tau_p} \frac{d\sigma}{t-q^{2\al(\sigma)}}}}\ \qq^{-2t\textstyle{ \int_{\tau_p}^{\tau_{p+1}} \frac{d\sigma}{t-q^{2(\al(\tau_p)+\sigma-\tau_p)}}}}\ \qq^{-2t\textstyle{ \int_{\tau_{p+1}}^1 \frac{d\sigma}{t-q^{2\al(\sigma)}}}}\\
&=\qq^{-2t\textstyle{ \int_0^{\tau_p} \frac{d\sigma}{t-q^{2\al(\sigma)}}}}\ \frac{\qq^{2\al(\tau_p)}\, \left(t-\qq^{2\al(\tau_{p+1})}\right)}{\qq^{2\al(\tau_{p+1})}\, \left(t-\qq^{2\al(\tau_{p})}\right)}\ \qq^{-2t\textstyle{ \int_{\tau_{p+1}}^1 \frac{d\sigma}{t-q^{2\al(\sigma)}}}}\\
\end{split}
\label{eq:xtfilled}
\end{equation}
(where we used $\al(\tau_{p+1})-\al(\tau_p)=\tau_{p+1}-\tau_p$). Assuming say $\qq>1$,  the left and right factors are well defined and real positive respectively for $t\in]-\infty,1]\cup[\qq^{2\al(\tau_p)},+\infty[$ and
for $t\in]-\infty,\qq^{\al(\tau_{p+1})}]\cup[\qq^{2\al(1)},+\infty[$ and the middle factor is well defined for all values\footnote{Again we allow for convenience the limiting possibility $\xx(t)\to \pm\infty$.}
of $t$ so that $\xx(t)$ is well defined and real for $t\in]-\infty,1]\cup[\qq^{2\al(\tau_p)},\qq^{2\al(\tau_{p+1})}]\cup[\qq^{2\al(1)},+\infty[$.
The new (middle) interval of $t$ then creates a new portion of arctic curve via the expression of Theorem~\ref{thmac}. Note that $\xx(t)$ is positive for 
$t\in]-\infty,1]\cup[\qq^{2\al(1)},+\infty[$ but negative for $t\in[\qq^{2\al(\tau_p)},\qq^{2\al(\tau_{p+1})}]$.
}
\end{itemize}
At this stage, let us make the following remark: from the expression \eqref{eq:Ft} for the family of tangent curves, we deduce at $Y=1$ the identity
\begin{equation*}
\mathcal{F}_t(X,1)=\big(X-t\big)\big(1-\xx(t)\big)\big(1+\gamma\, \xx(t)\big) \ .
\end{equation*}
Therefore, if we may find a \emph{finite} (and strictly positive) value of $t$ such that either $\xx(t)=1$ or $\xx(t)=-1/\gamma$, then, upon differentiating $\mathcal{F}_t(X,1)$ with respect
to $t$, we deduce that both $\mathcal{F}_t(X,1)$ and $\frac{\partial \mathcal{F}_t}{\partial t}(X,1)$ vanish at $X=t$, hence the point $\big({\rm Log}(t)/{\rm Log}(\qq^2),0\big)$ lies on the
arctic curve. Moreover, since $\frac{\partial \mathcal{F}_t}{\partial X}(X,1)$ also vanishes, the tangent curve, hence the arctic curve itself, has a slope $0$ at this point. In other words, 
we have the following property:
\begin{prop}
\label{proptang}
The arctic curve is tangent to the horizontal axis at the point $\big({\rm Log}(t)/{\rm Log(\qq^2)},0\big)$ for any \emph{finite} (and strictly positive) value of $t$ such that either\, $\xx(t)=1$ or\, $\xx(t)=-1/\gamma$.
\end{prop}\
\begin{figure}
\begin{center}
\includegraphics[width=8cm]{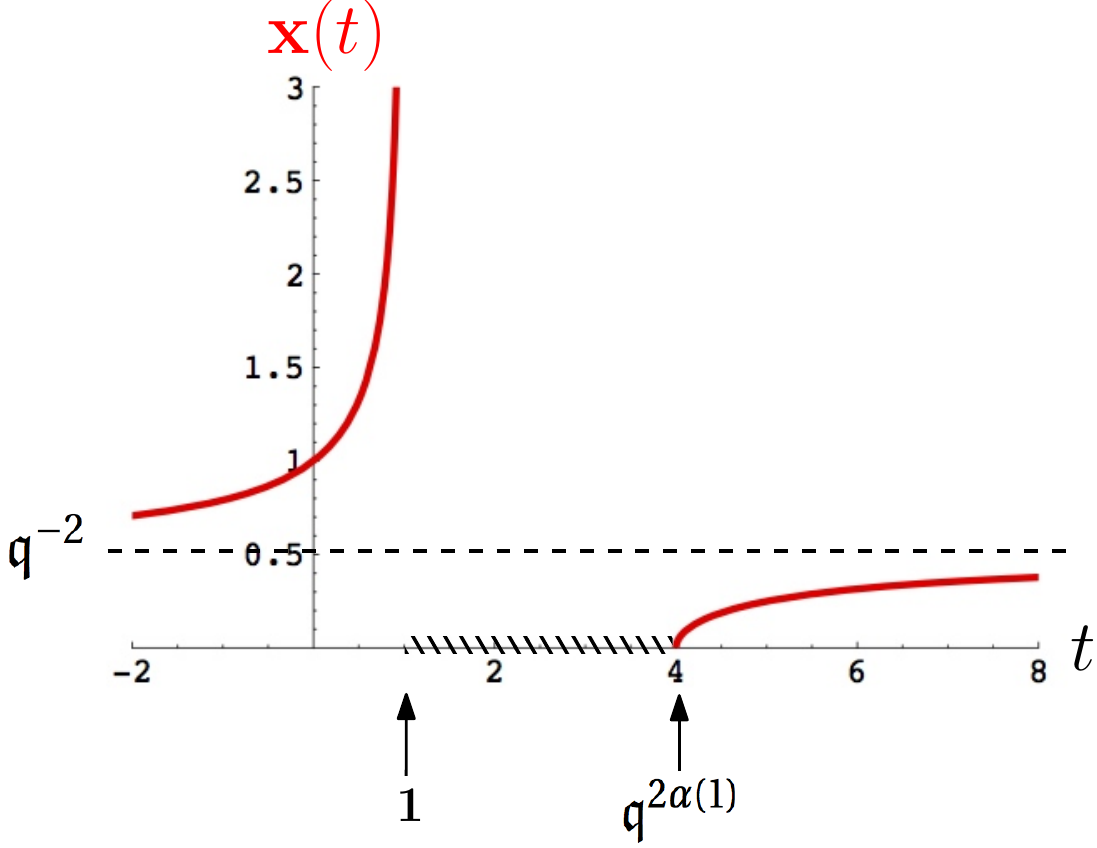}
\end{center}
\caption{\small{Variations of the function $\xx(t)$ for a generic distribution $\al(\sigma)$, here for $\qq>1$.}}
\label{fig:xt}
\end{figure}Note that such a situation never occurs for a generic distribution $\al(\sigma)$ for which, as illustrated in Fig.~\ref{fig:xt}, $\xx(t)$ remains positive 
(therefore cannot be equal to $-1/\gamma$) and
is such that $\xx(t)=1$ only for $t=0$. Indeed, for $\qq>1$, $\xx(t)$ increases from $\qq^{-2}<1$ to $1$ when $t$ increases from $-\infty$ to $0$, then increases 
from $1$ to $+\infty$ when $t$ increases from $0$ to $1$, and finally increases from $0$ to $\qq^{-2}<1$ when $t$ increases from $\qq^{2\al(1)}$ to $+\infty$ (recall that
$\xx(t)$ is generically not defined for $t\in]1,\qq^{2\al(1)}[$). As for $\qq<1$, $\xx(t)$ decreases from $\qq^{-2}>1$ to $1$ when $t$ increases from $-\infty$ to $0$, then decreases 
from $1$ to $0$ when $t$ increases from $0$ to $\qq^{2\al(1)}$, and finally decreases from $+\infty$ to $\qq^{-2}>1$ when $t$ increases from $1$ to $+\infty$ (recall that
$\xx(t)$ is generically not defined for $t\in]\qq^{2\al(1)},1[$). 

\begin{figure}
\begin{center}
\includegraphics[width=8cm]{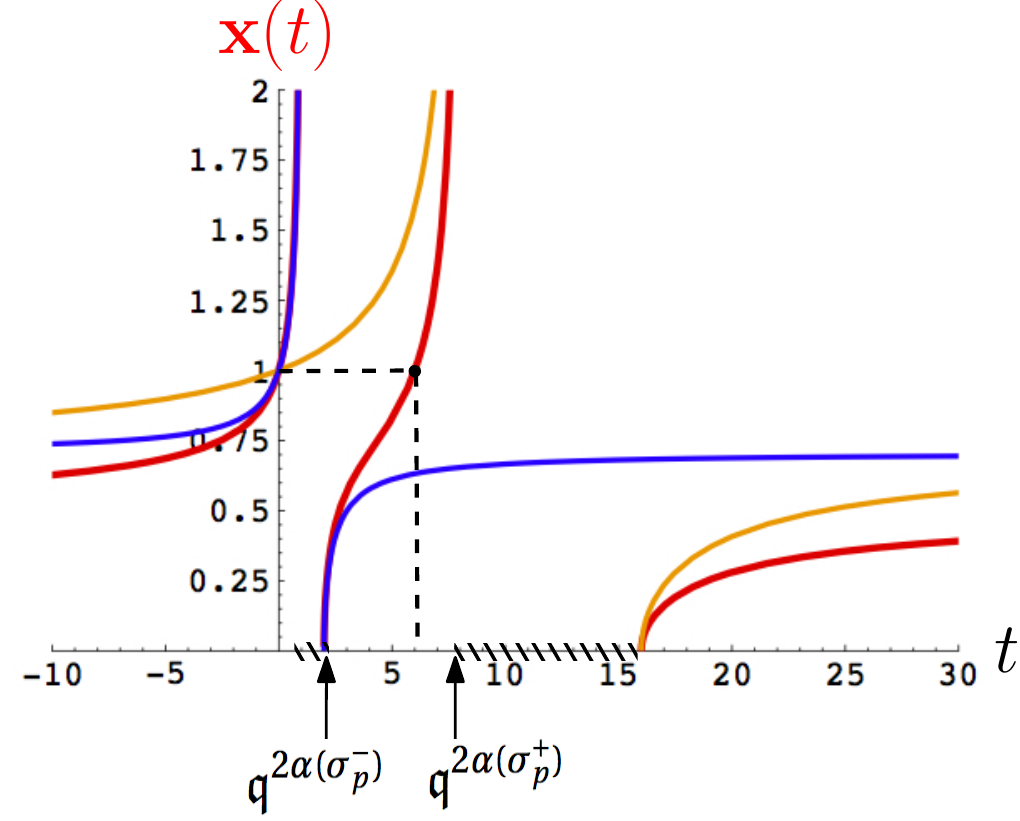}
\end{center}
\caption{\small{Variations of the function $\xx(t)$ of \eqref{eq:xtgap} for a distribution $\al(\sigma)$ presenting
a single gap, and for $\qq>1$. The first contribution to $\xx(t)$ (in blue) increases from $0$ to $\qq^{-2}$
for $t$ increasing from $\qq^{2\al(\sigma_p^-)}$ to $+\infty$. The second contribution (in orange) increases from $\qq^{-2}$ to $+\infty$ 
for $t$ increasing from $-\infty$ to $\qq^{2\al(\sigma_p^+)}$. The function $\xx(t)$ ) (in bold red) therefore increases continuously from $0$ to $+\infty$ when $t$ increases 
from $\qq^{2\al(\sigma_p^-)}$ to $\qq^{2\al(\sigma_p^+)}$, hence passes through the value $\xx(t)=1$ for some finite $t\in [\qq^{2\al(\sigma_p^-)},\qq^{2\al(\sigma_p^+)}]$.}}
\label{fig:xtone}
\end{figure}
The possibility of finding a finite $t$ for which $\xx(t)=1$ is however encountered in the above-described case of a distribution $\al(\sigma)$ presenting
a gap: indeed, as illustrated in Fig.~\ref{fig:xtone} for $\qq>1$, since the first contribution to $\xx(t)$ in \eqref{eq:xtgap} increases from $0$ to $\qq^{-2}$
for $t$ increasing from $\qq^{2\al(\sigma_p^-)}$ to $+\infty$ (assuming for simplicity a single gap) while the second contribution to $\xx(t)$  increases from $\qq^{-2}$ to $+\infty$ 
for $t$ increasing from $-\infty$ to $\qq^{2\al(\sigma_p^+)}$, then $\xx(t)$ increases continuously from $0$ to $+\infty$ when $t$ increases from $\qq^{2\al(\sigma_p^-)}$ to $\qq^{2\al(\sigma_p^+)}$,
hence passes through the value $\xx(t)=1$ for some finite $t\in [\qq^{2\al(\sigma_p^-)},\qq^{2\al(\sigma_p^+)}]$, leading to a tangency point
with $x\in[\al(\sigma_p^-),\al(\sigma_p^+)]$. A similar scenario holds for $\qq<1$.

\begin{figure}
\begin{center}
\includegraphics[width=8cm]{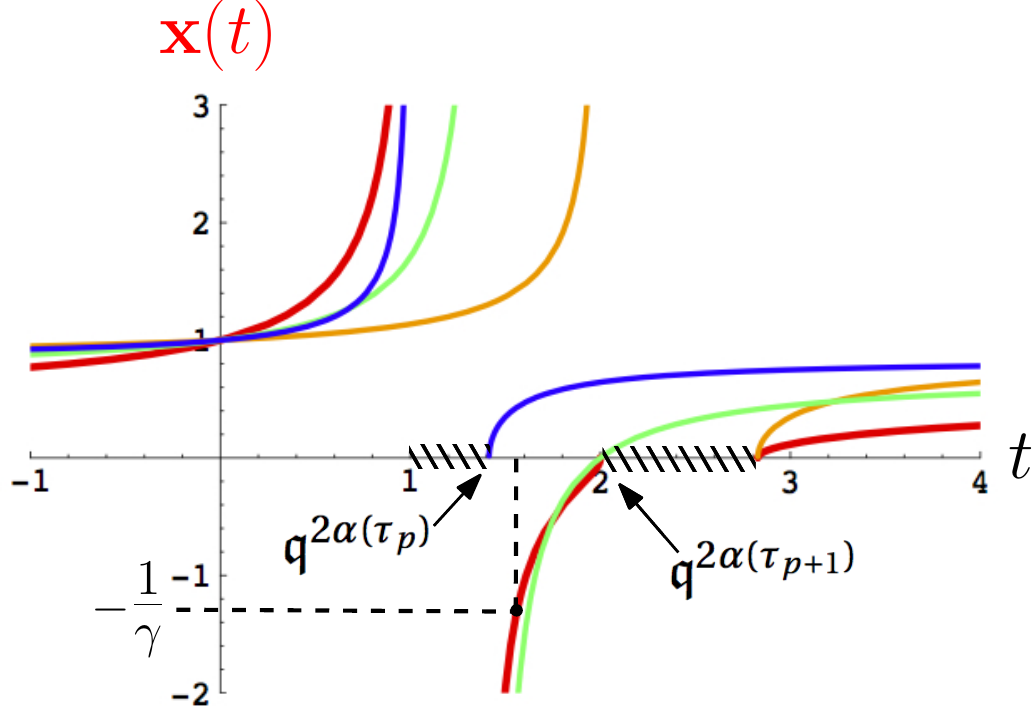}
\end{center}
\caption{\small{Variations of the function $\xx(t)$ of \eqref{eq:xtfilled} for a distribution $\al(\sigma)$ presenting
a single minimal slope interval, and for $\qq>1$. The function $\xx(t)$ (in bold red) is the product of a first (in blue), a second (in green) and third (in orange)
contributions and runs from $-\infty$ to $0$ for $t\in [\qq^{2\al(\tau_p)},\qq^{2\al(\tau_{p+1})}]$, hence 
 passes through the value $\xx(t)=-1/\gamma$ for some finite $t$ in this interval.}}
\label{fig:xtgamma}
\end{figure}
The case $\xx(t)=-1/\gamma$ occurs for a distribution $\al(\sigma)$ presenting now a minimal slope interval, as illustrated in Fig.~\ref{fig:xtgamma} for $\qq>1$.
Looking now at the expression \eqref{eq:xtfilled} for $\xx(t)$, the left contribution increases from $0$ to $\qq^{-2}$
for $t$ increasing from $\qq^{2\al(\tau_p)}$ to $+\infty$ (assuming for simplicity a single minimal slope interval) while the right contribution to $\xx(t)$  increases from $\qq^{-2}$ to $+\infty$ 
for $t$ increasing from $-\infty$ to $\qq^{2\al(\tau_{p+1})}$, so that their product increases from $0$ to $+\infty$ when $t$ increases from $\qq^{2\al(\tau_p)}$ to $\qq^{2\al(\tau_{p+1})}$.
A closer look at the integrals shows that the approach to $0$ for $t\to \qq^{2\al(\tau_p)}$ is of the form $(t-\qq^{2\al(\tau_p)})^{1/\al'(\tau_p^+)}$ and the approach to $+\infty$ for $t\to \qq^{2\al(\tau_{p+1})}$ 
is of the form $(\qq^{2\al(\tau_{p+1})}-t)^{-1/\al'(\tau_{p+1}^-)}$. This product must be multiplied by the middle contribution in \eqref{eq:xtfilled} which is \emph{negative} and runs from $-\infty$ to $0$ in this interval,
with a divergence of the form $-(t-\qq^{2\al(\tau_p)})^{-1}$ and a vanishing of the form $-(\qq^{2\al(\tau_{p+1})}-t)$. Since $\al(\tau_p^+)>1$ and $\al(\tau_{p+1}^-)>1$ (we assume that the
minimal slope interval does not extend outside the interval $[\qq^{2\al(\tau_p)},\qq^{2\al(\tau_{p+1})}]$), the function $\xx(t)$ runs from $-\infty$ to $0$ in this interval, 
hence passes through the value $\xx(t)=-1/\gamma$ for some finite $t\in [\qq^{2\al(\tau_p)},\qq^{2\al(\tau_{p+1})}]$, leading to a tangency point
with $x\in[\al(\tau_p),\al(\tau_{p+1})]$.

\begin{remark}
\label{qonetgrmk} We have the following $\qq=1$ equivalent of Proposition \ref{proptang} (obtained by setting $t=\qq^{2\tau}$ and letting $\qq\to 1$):

\begin{itemize}
\item[]{the arctic curve for $\qq=1$ is tangent to the horizontal axis at the point $(\tau,0)$ for any \emph{finite} (and strictly positive) value of\, $\tau$ such that either\, $\xx(\tau)=1$ or\, $\xx(\tau)=-1/\gamma$,
with $\xx(\tau)$ as in Remark~\ref{rmkfamily}.}
\end{itemize}

\noindent The above discussion extends straightforwardly and allows to follow the variation of the function $\xx(\tau)$ in a generic case or in the presence of a gap or a minimal slope interval
in the distribution $\al(\sigma)$. Similarly, the situation with $\xx(\tau)=1$ $(\tau>0)$ is encountered in the presence of a gap while that with $\xx(\tau)=-1/\gamma$ is encountered in the presence
of a minimal slope interval.
\end{remark}

\section{Examples}
\label{sec:examples}
This section is devoted to the exposition of a number of examples of arctic curves for various boundary conditions encoded in the distribution 
$\al(\sigma)$ of path starting points. We display in particular the deformation of the arctic curve for varying values of $\qq$ and $\gamma$
and its limiting shape for $\qq$ and $\gamma$ getting small or large. As already mentioned, the arctic curve decomposes into various portions 
corresponding to various domains of the variable $t$. We will distinguish three categories of starting point distributions:
\begin{itemize}
\item{A generic case, with $\al(\sigma)$ continuous and  $\al'(\sigma)>1$. The variable $t$ then spans two semi-infinite interval and the arctic curve 
is made of two portions meeting at some limiting point with tangent $y=1$ (corresponding to $t\to \pm \infty$).} 
\item{A case with \emph{freezing boundaries}, i.e.\  with either minimal slope intervals on which $\al'(\sigma)=1$ (intervals of the S-boundary with no defect) or 
with gaps i.e.\ discontinuities of $\al(\sigma)$ (intervals of the S-boundary filled with defects) interspersed with generic portions with $\al'(\sigma)>1$. The variable $t$ then spans some extra finite intervals, thus adding new portions to the arctic curve corresponding to the creation of frozen domains adjacent to the S-boundary.}
\item{A case of \emph{fully frozen boundaries} with an alternation of freezing boundaries of both types (minimal slope intervals and gaps)
with no generic interval in-between. This includes in particular the case of an S-boundary 
with no defect, corresponding to the original Aztec Diamond tiling problem.}    
\end{itemize}

\subsection{Generic case}
For illustration, we discuss here the simplest case of a generic distribution of starting points by choosing $\al(\sigma)=2\sigma$ corresponding for instance to
a regular pattern of defects, with a defect at every second point of the S-boundary. We have in this case
\begin{equation*}
\xx(t)=\frac{1}{\qq^2}\, \sqrt{\frac{t-\qq^4}{t-1}}\ .
\end{equation*}
\begin{figure}
\begin{center}
\includegraphics[width=8cm]{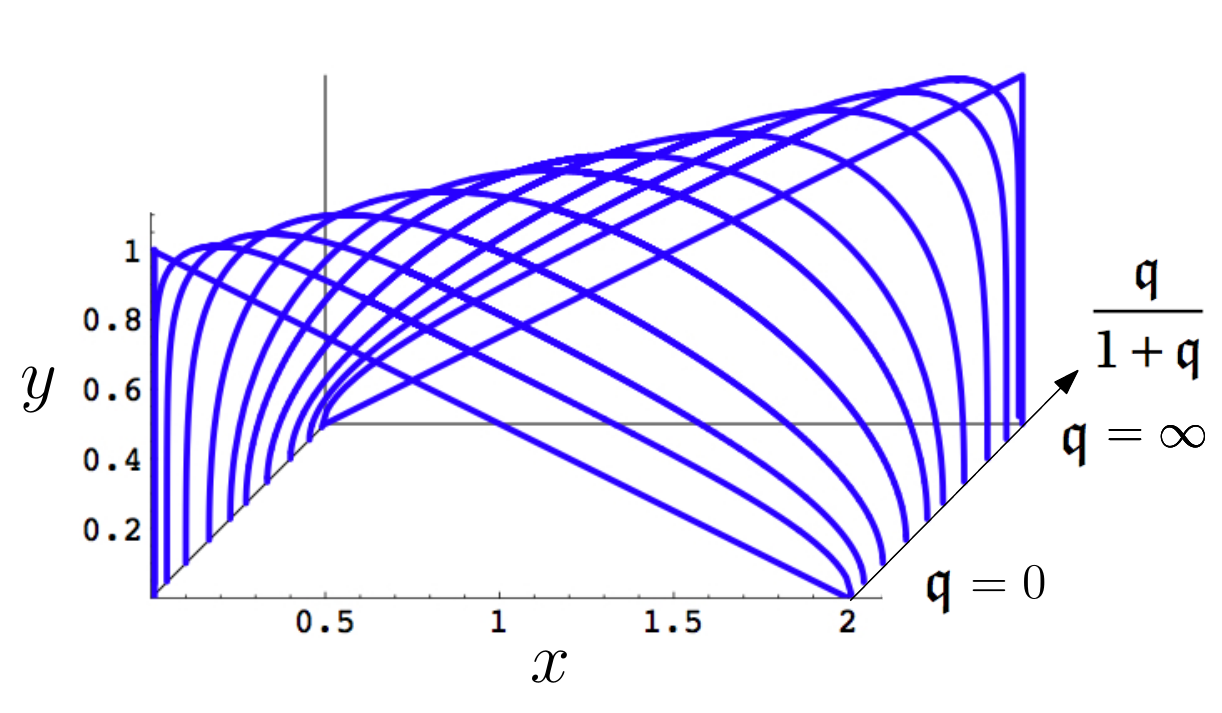}
\end{center}
\caption{\small{The variation of the arctic curve for $\al(\sigma)=2\sigma$ at fixed $\gamma=1$ as a function of $\qq\in [0,+\infty[$. The curve for a given $\qq$ is displayed in
the vertical plane of depth $\qq/(1+\qq)$.}}
\label{fig:genericq}
\end{figure}
Fig.~\ref{fig:genericq} displays the variation of the arctic curve for $\gamma=1$ and a value of $\qq$ ranging from $0$ to $+\infty$. For $\qq=0$, 
the arctic curve degenerates into a broken line made of a vertical line segment from $(0,0)$ to $(0,1)$ and a line segment with slope $-1/2$ from $(0,1)$ to $(2,0)$.
This corresponds to the unique NILP configuration with minimal area, where each path alternates between horizontal and diagonal steps. The above 
segment with slope $-1/2$ is the limit of the region occupied by these paths. Similarly, for $\qq\to+\infty$, the arctic curve is made of a vertical line segment from $(2,0)$ to $(2,1)$ 
and a segment with slope $1/2$ from $(0,0)$ to $(2,1)$. This now corresponds to the unique NILP configuration with maximal area, where the $i$-th path (from the bottom) 
is made of $i$ vertical steps followed by $2i$ horizontal steps. The segment with slope $1/2$ is the locus of the change from vertical to horizontal.
The manifest symmetry relating small and large $\qq$ is a consequence of a more general left-right symmetry discussed in Section \ref{sec:symsec}.

\begin{figure}
\begin{center}
\includegraphics[width=8cm]{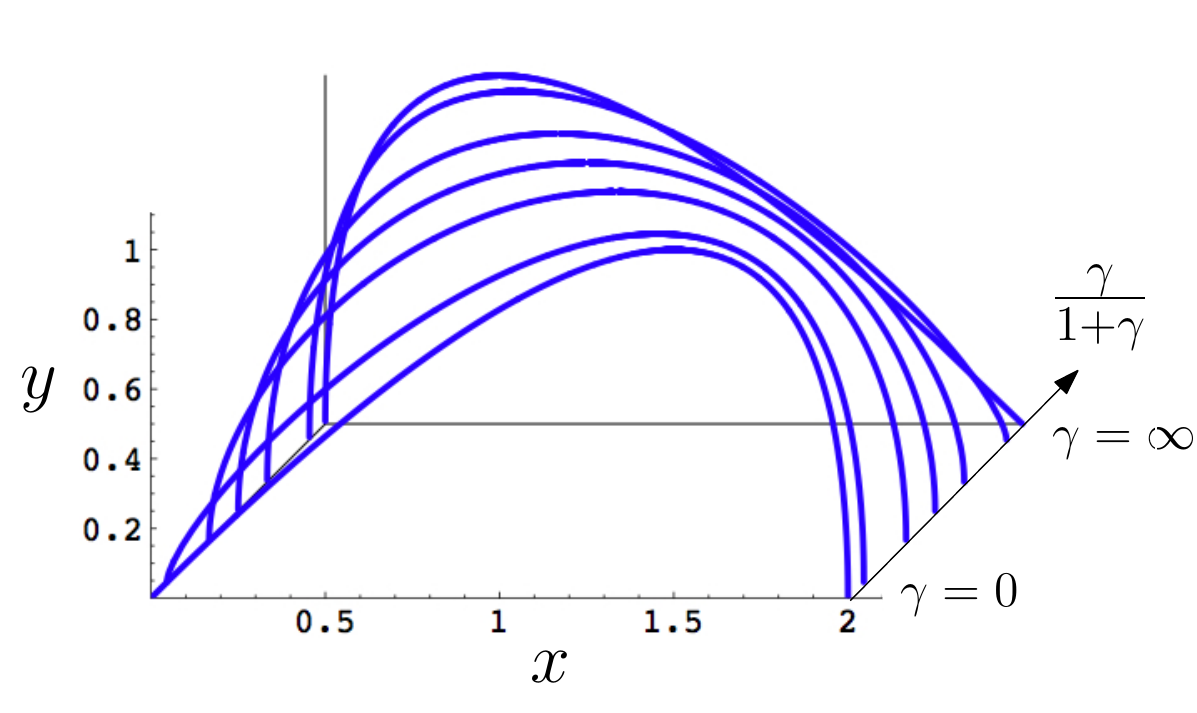}
\end{center}
\caption{\small{The variation of the arctic curve for $\al(\sigma)=2\sigma$ at fixed $\qq=1$ as a function of $\gamma\in [0,+\infty[$. The curve for a given $\gamma$ is displayed in
the vertical plane of depth $\gamma/(1+\gamma)$.}}
\label{fig:genericgamma}
\end{figure}
Fig.~\ref{fig:genericgamma} displays the variation of the arctic curve for $\qq=1$ and a varying $\gamma$ in the range $[0,+\infty[$. Again we have a manifest symmetry relating small and large $\gamma$.
The arctic curve for $\gamma=0$ or $+\infty$ is not degenerate. As already mentioned, the case $\gamma=0$ corresponds to that studied in \cite{DFGUI} in connection with rhombus tilings with defects
and the arctic curve for $\al(\sigma)=2\sigma$ is a portion of parabola, as first obtained in \cite{DFLAP}. The arctic curve for $\gamma=+\infty$ is the reflected piece of parabola under $x\to 2-x$. Note also that the arctic curve for $\gamma=1$ is the semi-circle $x^2+y^2-2x=0$
of radius $1$ centered in $(1,0)$. For generic $\gamma$, the curve has a maximum at $(x,y)=\left(\frac{3+\gamma}{2(1+\gamma)},1\right)$ with a horizontal tangent.

\subsection{Freezing boundaries}
We now address the case of a freezing boundary and start with a distribution $\al(\sigma)$ presenting a unique minimal slope interval for $\sigma\in [1/3,2/3]$. For illustration,
we choose a distribution 
$\al(\sigma)=2\sigma$ for $\sigma\in[0,1/3]$, $\al(\sigma)=\sigma+1/3$ for $\sigma\in[1/3,2/3]$  and $\al(\sigma)=2\sigma-1/3$ for $\sigma\in[2/3,1]$. This leads to
\begin{equation*}
\xx(t)=\frac{1}{\qq^2}\, \sqrt{\frac{t-\qq^{4/3}}{t-1}}\ \frac{t-\qq^{2}}{t-\qq^{4/3}}\ \sqrt{\frac{t-\qq^{10/3}}{t-\qq^2}}\ .
\end{equation*}
For $\qq\to 1$, the corresponding function $\xx(\tau)$ reads
\begin{equation*}
\xx(\tau)=\sqrt{\frac{\tau-2/3}{\tau}}\ \frac{\tau-1}{\tau-2/3}\ \sqrt{\frac{\tau-5/3}{\tau-1}}
\end{equation*}
which is real negative in the new interval $\tau\in [2/3,1]$. 
\begin{figure}
\begin{center}
\includegraphics[width=14cm]{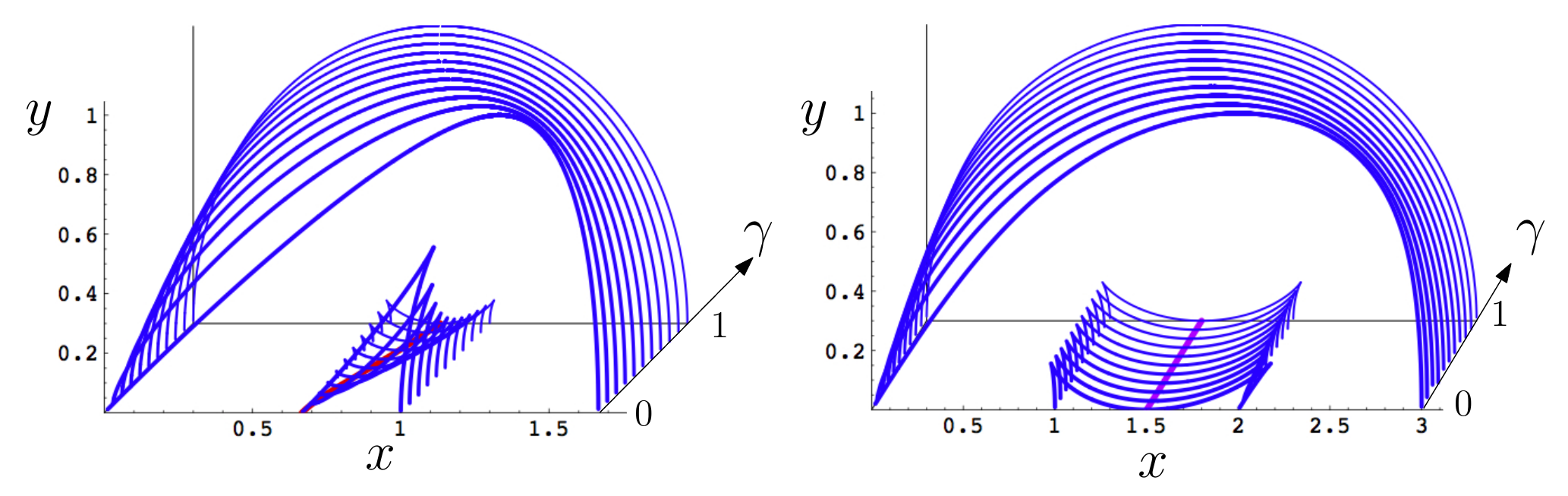}
\end{center}
\caption{\small{The variation of the arctic curve for $\qq=1$ as a function of $\gamma\in[0,1]$ in two situations with a freezing boundary. Left: the
case of a single minimal slope interval between two generic domains (see text for the value of $\al(\sigma)$). Right: the
case of a gap between two generic domains. In both cases, a new portion of arctic curve appears with two cusps and a tangency point at $x=\tau$ such that 
$\xx(\tau)=-1/\gamma$} (left) or $\xx(\tau)=1$ (right). We have indicated the loci of tangency points by thick red/purple curves.}
\label{fig:gapfilledgamma}
\end{figure}
As displayed in Fig.~\ref{fig:gapfilledgamma} (left), this creates a new portion of arctic curve with two cusps and a tangency point at $x=\tau$ such that $\xx(\tau)=-1/\gamma$,
in agreement with Remark~\ref{qonetgrmk}.

We display similarly in Fig.~\ref{fig:gapfilledgamma} (right) the arctic curves for a freezing boundary corresponding now to a gap in $\al(\sigma)$.
The corresponding function $\xx(\tau)$ reads
\begin{equation*} 
\xx(\tau)=\sqrt{\frac{(\tau-1)(\tau-3)}{\tau(\tau-2)}}
\end{equation*}
which is real positive in the new interval $\tau\in [1,2]$. 
This now creates a new portion of arctic curve with two cusps and a tangency point at $x=\tau$ such that $\xx(\tau)=1$,
in agreement with Remark~\ref{qonetgrmk}. Note that, as opposed to the previous case, the $x$-coordinate of the tangency
point (here $\tau=3/2$) is independent of $\gamma$. 

\begin{figure}
\begin{center}
\includegraphics[width=14cm]{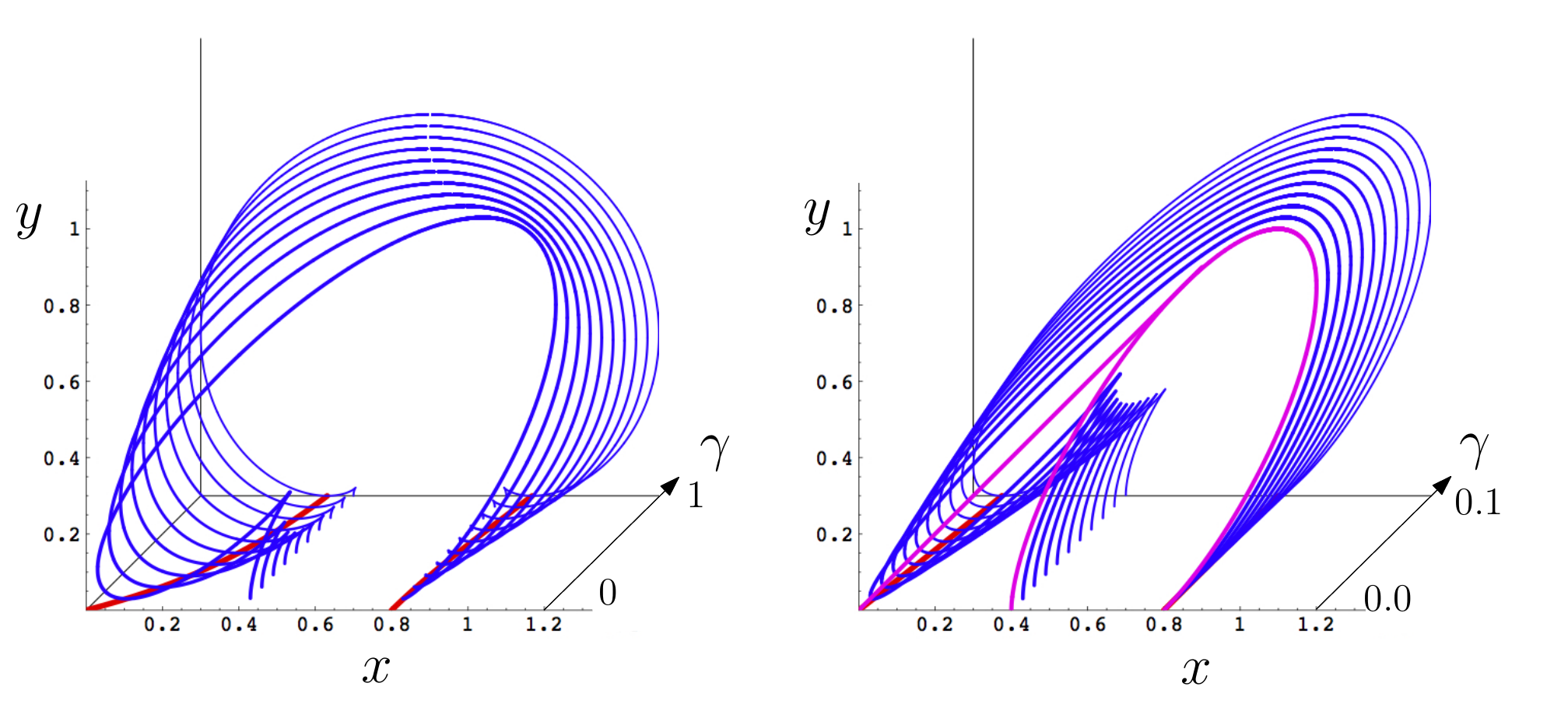}
\end{center}
\caption{\small{The variation of the arctic curve for $\qq=1$ as a function of $\gamma\in[0,1]$ in the case of a freezing boundary with 
two minimal slope intervals separated by a generic domain (see text for the value of $\al(\sigma)$). Two extra portions of arctic curve are created, 
each with a cusp and tangent to the $x$-axis at $x=\tau$ solution of \ $\xx(\tau)=-1/\gamma$. The locus of the two solutions of this equation is 
indicated by the two thick red curves. The left figure shows the variation for $\gamma\in[0.1,1]$ and the right one for $\gamma\in [0,0.1]$ in order to 
emphasize the approach to the $\gamma=0$ degenerate limit indicated in magenta.}}
\label{fig:onetwoone}
\end{figure}
Finally, we present a slightly more involved example of freezing boundaries with two minimal slope intervals separated by a generic portion, namely
$\al(\sigma)=\sigma$ for $\sigma\in[0,2/5]$, $\al(\sigma)=2\sigma-2/5$ for $\sigma\in[2/5,3/5]$ and $\al(\sigma)=\sigma+1/5$ for $\sigma\in[3/5,1]$.
The corresponding function $\xx(\tau)$ reads
\begin{equation*}
\xx(\tau)=\frac{\tau-2/5}{\tau}\, \sqrt{\frac{\tau-4/5}{\tau-2/5}}\  \frac{\tau-6/5}{\tau-4/5}
\end{equation*}
which is real negative in the intervals $[0,2/5]$ and $[4/5,6/5]$. This creates two portions of arctic curve, each with a cusp, which are tangent 
to the $x$-axis at $x=\tau$ for the two solutions of $\xx(\tau)=-1/\gamma$, one in each of the above intervals (see Fig.~\ref{fig:onetwoone}).
For $\gamma\to 0$, the left cusp increases while the right one disappears. Simultaneously, the left tangency point reaches the value $x=0$
while the right one reaches $x=4/5$. At $\gamma=0$, the arctic curve degenerates into an algebraic curve of degree $6$ and a line segment 
$y=x$ for $x\in[0,3\sqrt{2}/5]$ tangent to the algebraic curve at their contact point $(3\sqrt{2}/5,3\sqrt{2}/5)$. This is due to the fact that, at $\gamma=0$, the first maximal slope interval induces
a macroscopic domain where all the paths have only vertical steps. The boundary of this ``vertically frozen phase" appears as a line segment with slope $1$ tangent to the algebraic curve. 

\subsection{Fully frozen boundaries}
\begin{figure}
\begin{center}
\includegraphics[width=14cm]{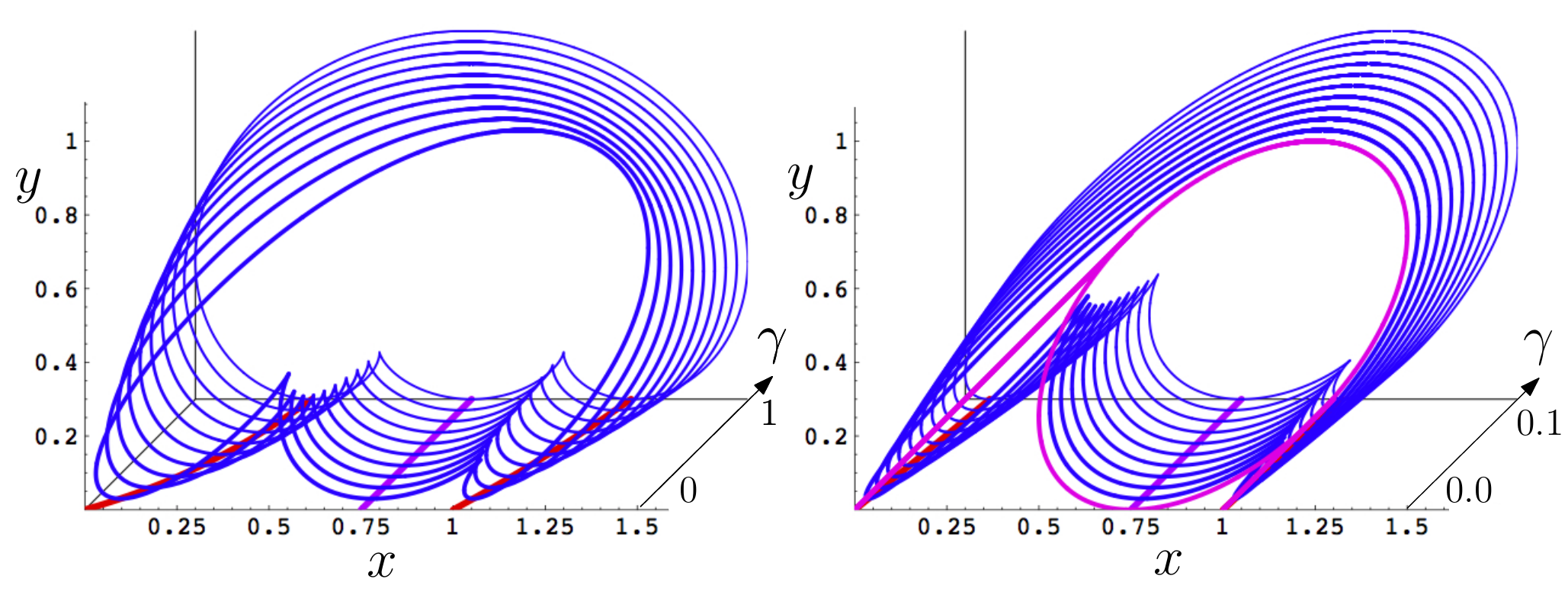}
\end{center}
\caption{\small{The variation of the arctic curve for $\qq=1$ as a function of $\gamma\in[0,1]$ in the case of a fully frozen boundary with 
two minimal slope intervals separated by a gap (see text for the value of $\al(\sigma)$). Three extra portions of arctic curve are created, 
which form the lower part of the arctic curve with three tangency points alternating with two cusps. The left and right tangency points have 
$x=\tau$ solution of \ $\xx(\tau)=-1/\gamma$ (thick red curve), while the middle tangency point has $x=\tau$ solution of \ $\xx(\tau)=1$ (thick 
purple line). The left figure shows the variation for $\gamma\in[0.1,1]$ and the right one for $\gamma\in [0,0.1]$ in order to 
emphasize the approach to the $\gamma=0$ degenerate limit indicated in magenta.}}
\label{fig:gapfilledgapgamma}
\end{figure}
We now address the case of a fully frozen boundary made of two minimal slope intervals separated by a gap. For illustration, we choose the distribution of starting
points such that $\al(\sigma)=\sigma$ for $\sigma\in [0,1/2]$ and $\al(\sigma)=\sigma+1/2$ for $\sigma\in [1/2,1]$. For $\qq=1$, this corresponds to
\begin{equation*} 
\xx(\tau)=\frac{(\tau-1/2)(\tau-3/2)}{\tau(\tau-1)}\ .
\end{equation*}
This function is defined for all values of $\tau$ as the three extra intervals $[0,1/2]\cup[1/2,1]\cup[1,3/2]=[0,3/2]$ cover the range between $0$ and $\al(1)=3/2$.
These intervals are responsible for three extra portions of arctic curve
which altogether form the lower part of the arctic curve. The latter has three tangency points separated by two cusps (see Fig.~\ref{fig:gapfilledgapgamma}).
The left and right tangency points are at $x=\tau$, where $\tau$ are the two solutions of  $\xx(\tau)=-1/\gamma$ in $[0,1/2]$ and $[1,3/2]$ respectively.
The middle tangency point has $x=\tau$  where $\xx(\tau)=1$ in $[1/2,1]$, namely $x=3/4$. For $\gamma=0$, our solution is a particular instance
of that discussed in Section 7.4 of \cite{DFG2} and describes the rhombus tiling of a hexagonal domain. It is well known that the limiting arctic curve is an ellipse,
as shown in  magenta in Fig.~\ref{fig:gapfilledgapgamma}. In the present NILP setting, the hexagonal domain is extended into a rectangle but the paths outside
of the hexagon are all frozen into vertical or horizontal segments. This is due to the fact that, at $\gamma=0$, the maximal slope intervals induce macroscopic 
domains where all the paths have only vertical steps (see Ref.~\cite{DFGUI}, Fig. 22). The boundaries of these ``vertically frozen phases" appear as two line segments with slope $1$ tangent to the ellipse.  
It is interesting to visualize in Fig.~\ref{fig:gapfilledgapgamma} how these two segments arise as $\gamma$ goes to $0$ from the closing of two outgrowths obtained by the merging of
the two cusps with the external part of the arctic curve.
 
A final example of fully frozen boundaries is given by the tiling of the Aztec Diamond, corresponding to a situation with no defect, hence $m=n$ and $a_i=i$, $i=0,\ldots,n$.
The starting point distribution is simply $\al(\sigma)=\sigma$ for $\sigma\in[0,1]$, hence corresponds to a minimal slope interval of maximal size $1$. The function $\xx(t)$
is simply
\begin{equation*}
\xx(t)=\frac{1}{\qq^2}\, \frac{t-\qq^2}{t-1}\ ,
\end{equation*}
defined for all $t$ and negative between $1$ and $\qq^2$. 

\begin{figure}
\begin{center}
\includegraphics[width=14cm]{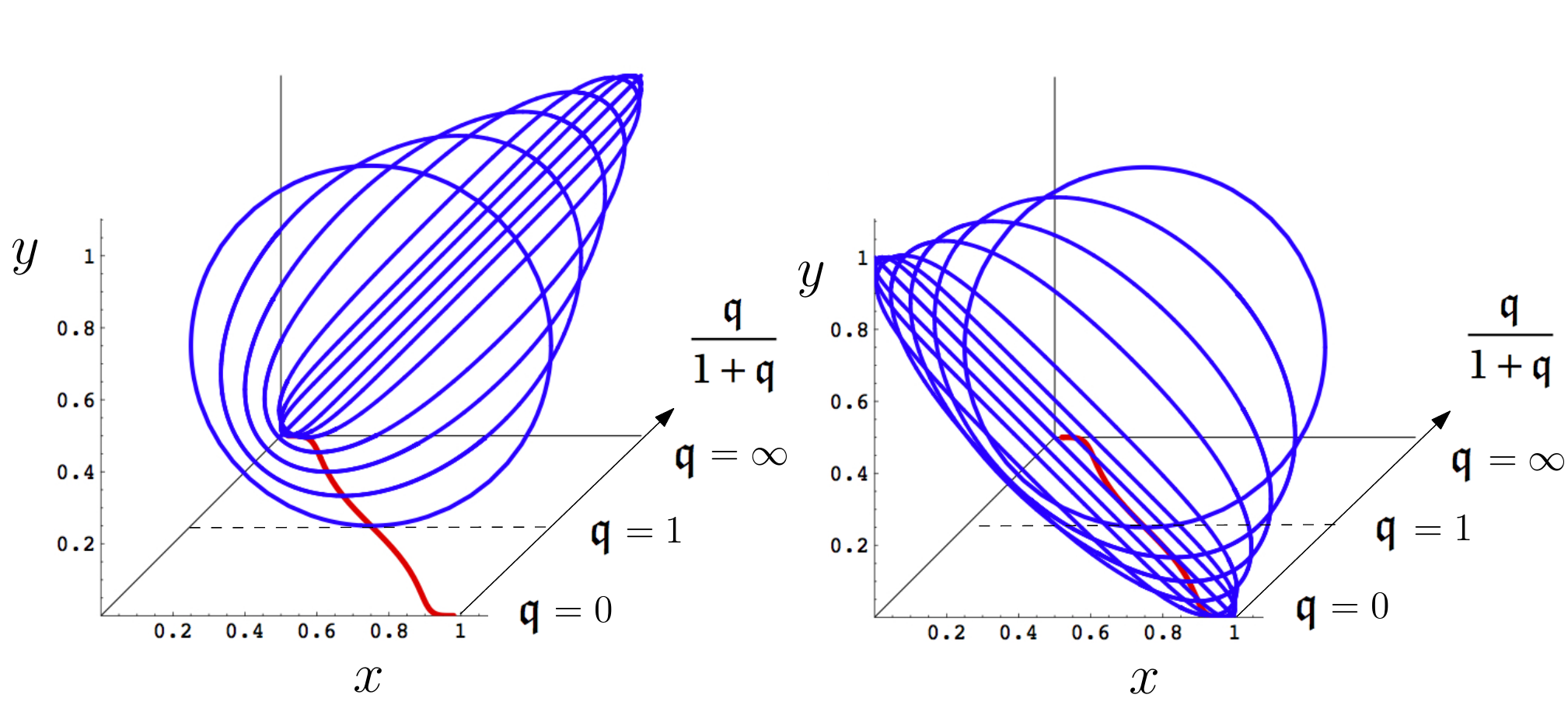}
\end{center}
\caption{\small{The variation of the arctic curve for $\al(\sigma)=\sigma$ (Aztec Diamond) at fixed $\gamma=1$ as a function of $\qq\in [1,+\infty[$ (left) and $\qq\in [0,1]$ (right). 
The curve for a given $\qq$ is displayed in
the vertical plane of depth $\qq/(1+\qq)$. At $\qq=1$, we recover the celebrated arctic circle of Ref.~\cite{JPS}. The curve is tangent to the $x$-axis at $x=\tau$ such that
$\xx(\tau)=-1/\gamma=-1$, as indicated by the thick red curve.}}
\label{fig:carreq}
\end{figure}
For $\qq=1$ and $\gamma=1$, the arctic curve is the celebrated arctic circle of Ref.~\cite{JPS}. Fig.~\ref{fig:carreq} shows the evolution of this curve for
varying $\qq$ at $\gamma=1$. The curve is tangent to the $x$-axis at position $x={\rm Log}(t)/{\rm Log}(\qq^2)$ with $t$ such that $\xx(t)=-1/\gamma=-1$,
namely 
\begin{equation*}
t=\frac{2\qq^2}{1+\qq^2}\ .
\end{equation*}
For $\qq\to 0$, the arctic curve degenerates into the line segment joining $(0,1)$ to $(1,0)$, corresponding to the unique dominant NILP configuration 
with minimal area where all the paths have only diagonal steps. The segment is then simply the limit of the domain covered by the paths. 
Similarly, for $\qq\to +\infty$, the arctic curve degenerates into the line segment joining $(0,0)$ to $(1,1)$, corresponding to the unique dominant NILP configuration 
with maximal area where the $i$-th path is made of $i$ vertical steps followed by $i$ horizontal ones. The segment is then simply the locus of the changes of slope of the paths.

\begin{figure}
\begin{center}
\includegraphics[width=14cm]{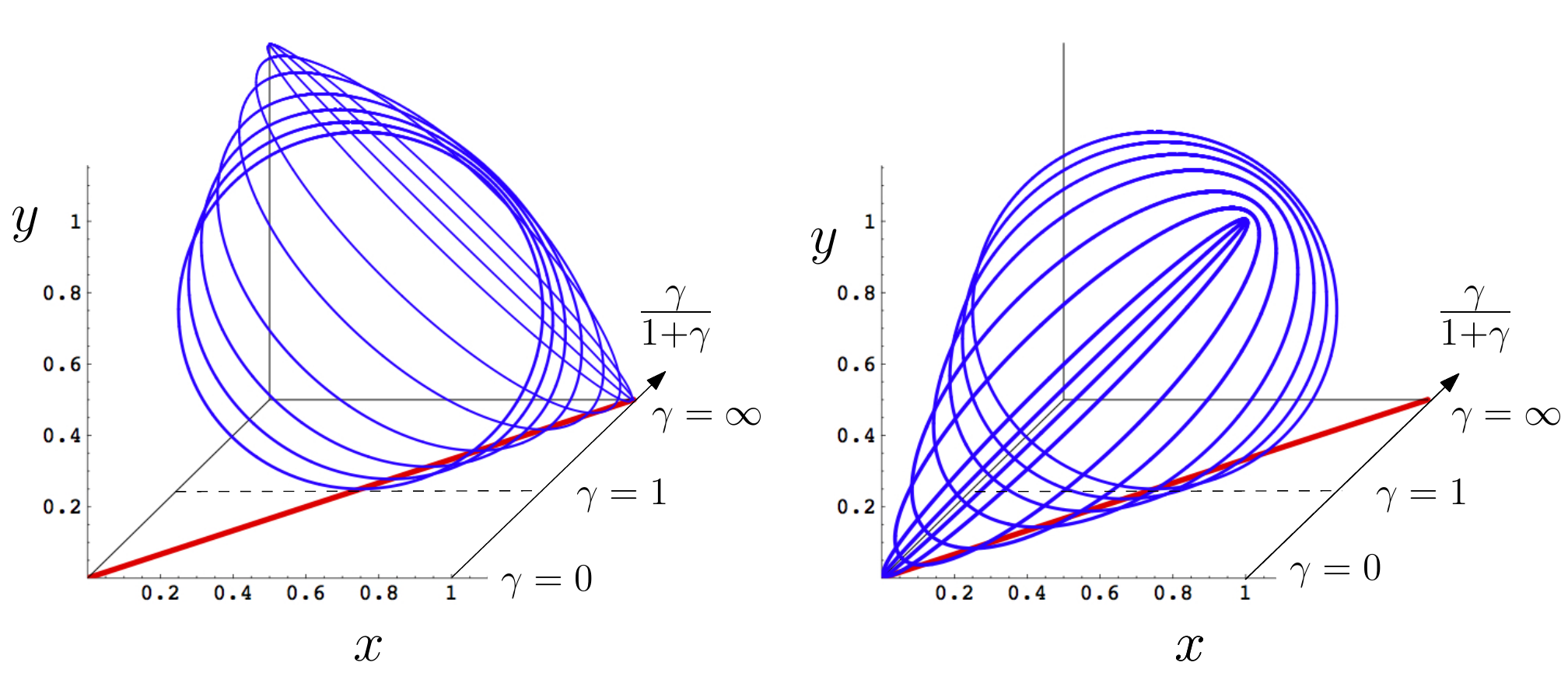}
\end{center}
\caption{\small{The variation of the arctic curve for $\al(\sigma)=\sigma$ (Aztec Diamond) at fixed $\qq=1$ as a function of $\gamma\in [1,+\infty[$ (left) and $\gamma\in [0,1]$ (right). 
The curve for a given $\gamma$ is displayed in
the vertical plane of depth $\gamma/(1+\gamma)$. At $\gamma=1$, we recover the arctic circle. The curve is tangent to the $x$-axis at $x=\tau$ such that
$\xx(\tau)=-1/\gamma$, i.e.\ $x=\gamma/(1+\gamma)$ as indicated by the thick red line.}}
\label{fig:carregamma}
\end{figure}
For $\qq=1$, the arctic curve is obtained from the function $\xx(\tau)=1-1/\tau$. Its variation with $\gamma$ is displayed in Fig~.\ref{fig:carregamma}: it is
tangent to the $x$-axis at position $x=\tau$ solution of \ $\xx(\tau)=-1/\gamma$, namely $x=\gamma/(1+\gamma)$. 
For $\gamma\to 0$, the arctic curve degenerates into the line segment joining $(0,0)$ to $(1,1)$, corresponding to the unique dominant NILP configuration with no diagonal step,
hence where the $i$-th path is made of $i$ vertical steps followed by $i$ horizontal ones.
For $\gamma\to +\infty$, the arctic curve degenerates 
into the line segment joining $(0,1)$ to $(1,0)$, corresponding again to the unique dominant NILP configuration where all the paths have only diagonal steps.

\section{Discussion and Conclusion}
\label{sec:conclusion}

\subsection{Summary and comparison to known results}
\label{sec:summary}
In this paper, we have extended previous results about the arctic phenomenon in NILP configurations with arbitrary starting points to a wider class of models
involving weighted Schr\"oder lattice paths. These are in bijection with domino tilings of Aztec rectangles with defects along one boundary. The weights incorporate two parameters
$q$ and $\gamma$ keeping track respectively of the area below the paths and of the number of steps of one particular type. The case $\gamma=0$ recovers the
weighted Rhombus tiling problem addressed in \cite{DFG2}. The changes in the asymptotic behavior of the configurations and in the shape
of the arctic curve induced by the introduction of a non-zero value of $\gamma$ are most drastic in the presence of freezing boundaries. In particular, the arctic curve 
develops new outgrowths associated with new tangency points (see Figs.~\ref{fig:onetwoone} and \ref{fig:gapfilledgapgamma}).

Like in Refs.~\cite{DFLAP,DFGUI,DFG2}, the results of the present paper are obtained by applying the Tangent Method
of \cite{COSPO}. This method is not fully rigorous, and our aim was twofold: on one hand, add evidence to the validity of the method by considering new examples for which the method reproduces known results; on the other hand 
use the method to derive new results.

The domino tiling of Aztec rectangles was considered in \cite{BuKni}, but with a different definition of defects along
the S-boundary: as opposed to our case, where each of the $m-n$ defects is created by adding an extra unit square along the S-boundary,
this paper considers $m-n$ defects created by {\em removing} unit squares from the S-boundary. The resulting NILP 
are different, as in the latter case the first step of each path can only be vertical or diagonal, whereas in our case it can also be horizontal. However, we expect this subtle difference to be irrelevant in the large size asymptotics. Indeed,
as noted in Remark \ref{qonerem}, the $\qq=1$ version of our main Theorem \ref{thmac} agrees with a result
of \cite{BuKni} in the case of {\em fully frozen boundaries}. In this case, the function $\xx(\tau)$ of Remark \ref{rmkfamily} takes the form of a rational fraction with equal number of single zeros and poles, and can be identified at $\gamma=1$ with the function
$\Pi_s(\theta)$ (see \cite{BuKni}, Theorem 5.1), namely $\Pi_s(\theta)=\xx(\tau)^{-1}$, while $\theta=\tau$. 
The parametric equation for the arctic curve \eqref{aczero} then matches Eq. (5.3) of \cite{BuKni} (modulo a misprint
$L\Pi_s\to L\Pi_s'=\Pi_s'/\Pi_s$).
For $\gamma\neq 1$, our result in the case of fully frozen boundaries matches that of \cite{BuKni}, Appendix A, Theorem 8.2, with the correspondence $q=1/\gamma$.
Let us stress that the methods employed in \cite{BuKni} are completely different and use asymptotic representation theory. This agreement is therefore highly non-trivial.

Finally, our result extends beyond the particular case of fully frozen boundaries and to arbitrary $\qq\neq 1$ as well.

\subsection{Symmetry of the arctic curve}
\label{sec:symsec}
Comparing Figs.~\ref{fig:NILP} and \ref{fig:otherNILP}, which describe \emph{the same distribution of defects} characterized in the scaling limit by the distribution $\al(\sigma)$,
we immediately see an obvious left-right symmetry in the path configuration design, under which $a_i\to {\tilde a}_i:=m-a_{n-i}=a_n-a_{n-i}$. Performing this symmetry on the paths of Fig.~\ref{fig:otherNILP} creates paths
similar to those of  Fig.~\ref{fig:NILP} but this symmetry induces the following changes: 
\begin{itemize}
\item{The new defect distribution after symmetry is characterized by the function $\tilde{\al}(\sigma):=\al(1)-\al(1-\sigma)$.}
\item{The weights are changed into $\tilde \gamma=1/\gamma$ and $\tilde{\qq}=1/\qq$ (or equivalently $\tilde q=1/q$ before rescaling).}
\end{itemize}
This is easily seen by comparing the tile-to-path dictionary for red and blue path steps:
\begin{equation}\label{leftrightarea}
\raisebox{-.6cm}{\hbox{\epsfxsize=12.cm \epsfbox{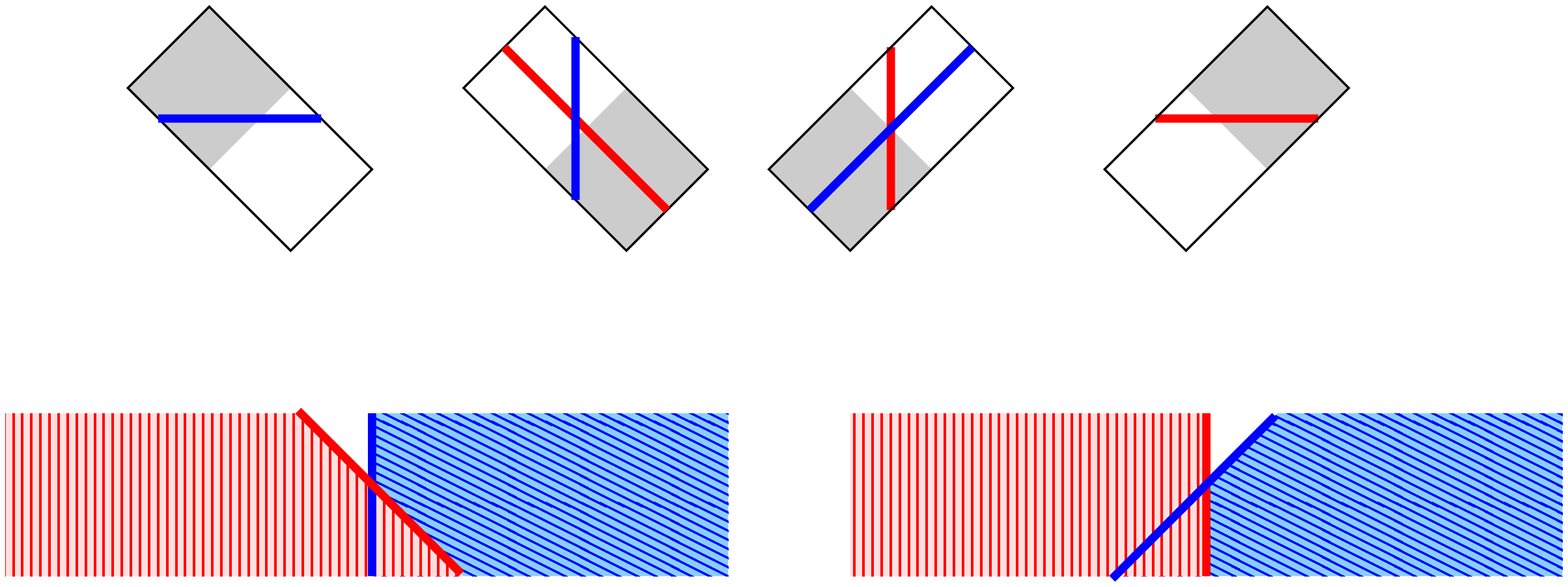}}}
\end{equation}
We  note that the total area to the left of all red paths, which may be decomposed into horizontal strips to the left of each vertical and diagonal step (see \eqref{leftrightarea} for an illustration),
plus the total area to the right of all blue paths similarly decomposed into horizontal strips to the right of the dual diagonal and vertical
steps, sum up to $(2m+1)\times n(n+1)/2$ ($=2m+1$ per vertical/diagonal step, with $i$ such contributions for the $i$-th path from the bottom,
$i=0,1,...,n$). Similarly the diagonal red steps correspond to vertical blue ones, while the total of vertical plus diagonal 
blue steps is equal to the total of vertical plus diagonal red ones, {\it i.e.} $n(n+1)/2$. As a consequence, we get
$$ Z(\{a_i\},\gamma,q)= q^{(2m+1)n(n+1)/2}\, \gamma^{n(n+1)/2}\, Z(\{{\tilde a}_i\},\tilde\gamma,\tilde q) .$$

In the continuum limit, expressing the exponential moment-generating function \eqref{defxt} as
$$\xx(\qq^{2\al(1)}\, t)=\qq^{-2t\textstyle{\int_0^1 \frac{d\sigma}{t-\qq^{2\al(\sigma)-2\al(1)}}}}
={\tilde \qq}^{2t\textstyle{\int_0^1 \frac{d \sigma}{t-{\tilde \qq}^{2\tilde{\al}(\sigma)}}}}=\frac{1}{\tilde\xx(t)}$$
allows to rewrite \eqref{eq:Ft} as:
$$ {\mathcal F}_{t}(X,Y)=-\qq^{2\al(1)}\,\gamma \, Y \,\xx(t)^2\, 
\tilde{\mathcal F}_{\tilde t}(\tilde X,\tilde Y)=0\ , \qquad \tilde t=t \qq^{-2\al(1)},\ \tilde X=X\qq^{-2\al(1)},\  \tilde Y=Y^{-1}\ .
$$
where $\tilde{\mathcal F}_{\tilde t}(\tilde X,\tilde Y)=0$ is the equation for the tangent family with weights 
$\tilde\gamma$ and $\tilde\qq$
and moment-generating function $\tilde\xx(\tilde t)$, and where we identify the arguments 
$\tilde X=\tilde\qq^{2\tilde x}$ and $\tilde Y=\tilde\qq^{2\tilde y}$, namely $\tilde x=\al(1)-x$ and
$\tilde y=y$ in terms of the new (reflected) coordinates $(\tilde x,\tilde y)$. 
The symmetry of the family of tangent curves is obviously shared by its envelope,
in summary:
\begin{prop}
The arctic curve of Theorem \ref{thmac}
is ``left-right symmetric" under the simultaneous change $(x,y,t,\gamma,\qq)\mapsto 
({\tilde x},{\tilde y},\tilde{t},\tilde \gamma,\tilde\qq)$ of coordinates and parameters, with 
$$ \tilde x=\al(1)-x,\ \ \tilde y=y,  \ \ \tilde{t}=t\, \qq^{-2\al(1)} ,\ \ \tilde \gamma=\gamma^{-1},\ \ \tilde\qq=\qq^{-1} \ .$$
\end{prop}

It is worth mentioning another kind of ``left-right" symmetry obeyed by the geodesics of Section \ref{sec:freetrajectory}. 
We note indeed that the algebraic equation $G_\gamma(X,Y;U,V)=0$ \eqref{geo} for the geodesics obeys the relation: 
$$ G_{\gamma^{-1}}(X^{-1},Y;U^{-1},V) =\frac{1}{U^2 X^2 \gamma}\, G_\gamma(X,Y;U,V)  .$$
This relation is unphysical, in the sense that only when $X,Y,U,V>1$ or $X,Y,U,V<1$ can the Schr\"oder paths be well-defined, and this symmetry violates these conditions.
However we believe it is of a very different (non-combinatorial) nature, having to do rather with properties
of analytic continuation of $q$-binomials and trinomials. 
This can be traced back to a symmetry property of the partition function for a single Schr\"oder path, as expressed
via Theorem \ref{polthm}, eq.~\eqref{eq:polz} which expresses the partition function for a weighted path from $(i,0)$ to $(0,j)$ as a polynomial $z_j(t)\equiv z_j(t;\gamma)$ evaluated at $t=q^{2i}$. Performing the change of summation variable $k\to j-k$ and
$s\to j+1-s$ in the numerator of the product, we easily get the symmetry relation:
$$ z_j(t;\gamma)= (-\gamma \, q \, t)^j \, z_j(q^{-2} t^{-1};\gamma^{-1})  ,$$
which allows to analytically continue $Z_{(i,0)\to (0,j)}=z_j(t=q^{2i};\gamma)$ to negative values of $i$, at the expense
of changing $\gamma\to\gamma^{-1}$. 
The map $(X,Y,U,V,\gamma)\mapsto (X^{-1},Y,U^{-1},V,\gamma^{-1})$ does not interchange the two branches
corresponding to the $\qq<1$ and $\qq>1$ geodesics.

\subsection{A geometric construction of the arctic curve for $\qq=1$}
\begin{figure}
\begin{center}
\includegraphics[width=10cm]{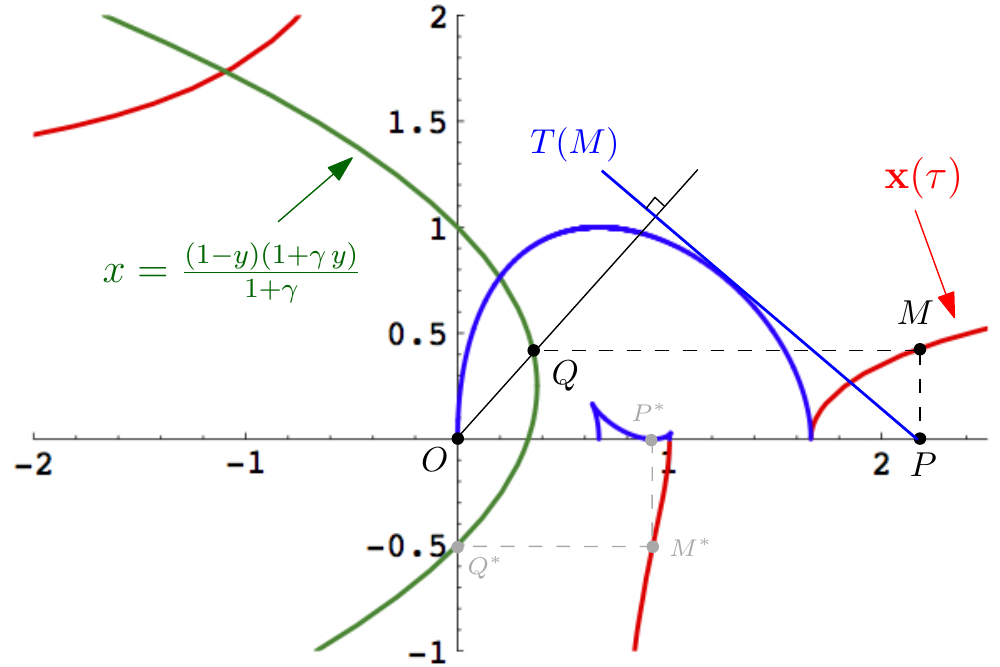}
\end{center}
\caption{\small{
The tangent line $T(M)$ (in blue) associated with the point $M:=(\tau,\xx(\tau))$ lying on the plot of the function $\xx(\tau)$ (in red) is the perpendicular 
to the line $(OQ)$ joining the origin $O$ to the horizontal projection $Q$ of $M$ on the parabola $x=\frac{(1-y)(1+\gamma\, y)}{1+\gamma}$ (in green)
passing trough the vertical projection $P$ of $M$ on the horizontal axis. Letting $M$ vary along the whole plot of the function $\xx(\tau)$ produces a family of tangent lines whose envelope 
is the entire arctic curve (curve in blue). Here we have chosen $\gamma=2$ and a function $\xx(\tau)$ corresponding to a distribution $\al(\sigma)$ with a minimal slope interval.
Note that, since the parabola hits the vertical axis at the point $Q^*$ (in gray) with $y=-1/\gamma$, the point $M^*$ such that $\xx(\tau)=-1/\gamma$ leads by the above construction 
to a line $T(M^*)$ precisely equal to the horizontal axis, hence the arctic curve is tangent to this axis at the associated point $P^*$, in agreement with Remark.~\ref{qonetgrmk}
}.}
\label{fig:wulff}
\end{figure}
In the particular case $\qq=1$, the family of tangent curves may be obtained via a simple geometric construction as in \cite{DFGUI}.  Using \eqref{eq:Ft} with $t=\qq^{2\tau}$ in the limit $\qq\to 1$, 
the tangent curves form a family of straight lines with equation (with the same slight abuse of notations as before)
\begin{equation*}
\mathcal{F}_\tau(x,y):= (1+\gamma)\, \xx(\tau)\, y+\big(1-\xx(\tau)\big)\big(1+\gamma\, \xx(\tau)\big)(x-\tau)=0\ ,
\end{equation*}
where $\xx(\tau)$ is as in Remark~\ref{rmkfamily}. In particular, the line $\mathcal{F}_\tau(x,y)=0$ is the line passing though the point $P:=(\tau,0)$ and orthogonal to the vector
$\overrightarrow{OQ}$ joining the origin $O$ to the point $Q:=\left(\frac{(1-\xx(\tau))(1+\gamma\, \xx(\tau))}{1+\gamma}, \xx(\tau)\right)$. To obtain the family of tangent curves,
in the $(O,x,y)$ plane, we may first represent the plot of the function $\xx(\tau)$ in this plane (see Fig.~\ref{fig:wulff}), pick a point $M:=(\tau,\xx(\tau))$ on this plot, denote by $P$ its vertical projection 
on the horizontal axis and by $Q$ its horizontal projection on the parabola $x=\frac{(1-y)(1+\gamma\, y)}{1+\gamma}$. The tangent line $T(M)$ with equation $\mathcal{F}_\tau(x,y)=0$ is the perpendicular 
to the line $(OQ)$ passing through $P$. Letting $M$ vary along the whole plot of $\xx(\tau)$ builds the entire family of tangent curves. 

Interestingly enough, the reverse construction allows, from the knowledge of the arctic curve at any fixed $\gamma$, to easily recover the function $\xx(\tau)$, hence all the moments of
the distribution of starting points. 

\subsection{Possible generalizations}
For $\gamma=0$, our model is identified as a $5$-vertex model on the square lattice \cite{DFG2}. More generally,
our weighted Schr\"oder path model for arbitrary $\gamma$ can be reformulated as a $10$-vertex model on the triangular lattice as follows: 
recall that the oriented lattice $\mathcal{N}$ in which our paths are embedded
is topologically equivalent to a regular triangular lattice. After rotation by $45^\circ$, paths give rise to the following 10 possible vertex environments:
\begin{equation*}
\raisebox{-.6cm}{\hbox{\epsfxsize=14.cm \epsfbox{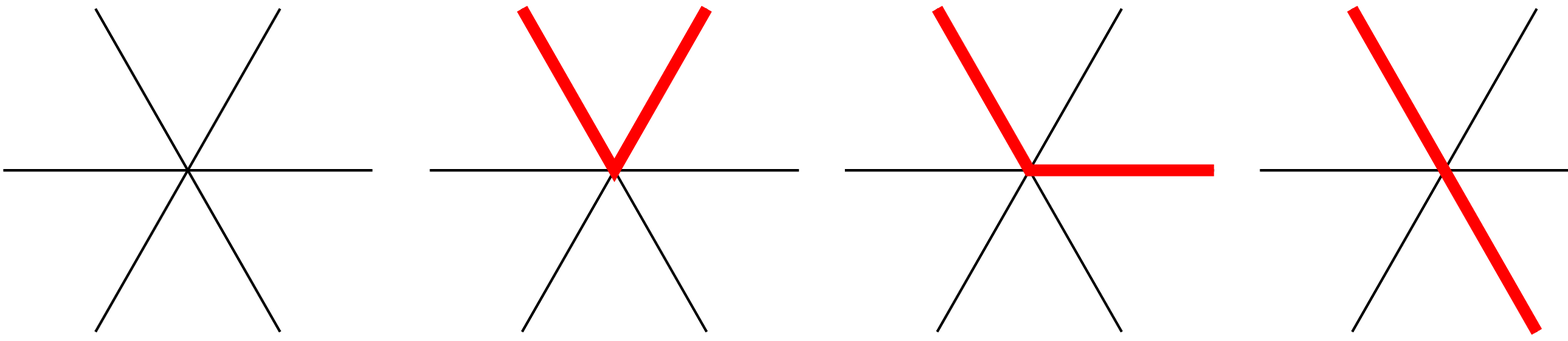}}}
\end{equation*}

This model is a particular case of a more general $20$-vertex model studied in Ref.~\cite{Kel}. In terms of paths, the extra $10$ vertices correspond to allowing for ``kissing points",
with the following new environments:
\begin{equation*}\raisebox{-.6cm}{\hbox{\epsfxsize=14.cm \epsfbox{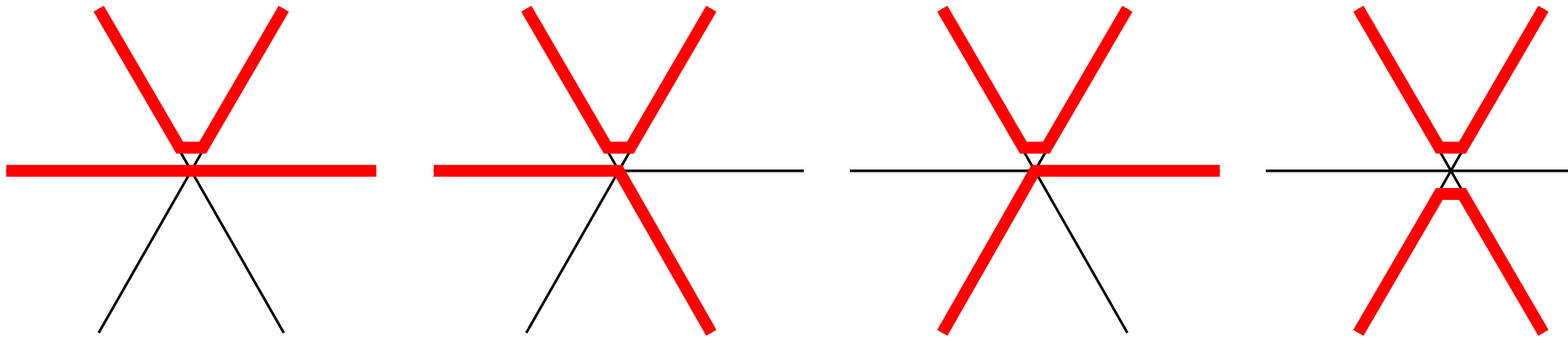}}}
\end{equation*}

As shown in Ref.~\cite{Kel}, the model is solvable by Bethe Ansatz techniques along a particular integrable variety of Boltzmann weights.

The kissing path formulation of the $20$-vertex model is the natural generalization of the osculating path formulation of the $6$-vertex model used in \cite{COSPO,CPS} to obtain in particular the arctic curve for 
the limit shape of large Alternating Sign Matrices via the Tangent Method. In this respect, it is tempting to postulate the existence of a limiting shape for 
the $20$-vertex model with appropriate boundary conditions. We also expect the Tangent Method to be applicable to this model to derive its arctic curve, whose shape should depend
on the new interaction weights between paths at the kissing points.   

In a different direction, the path formulation of the refined topological vertex of Ref.~\cite{RefTop}
involves elementary area weights alternating between two values $q$ and $t$ within strips depending 
on the boundary conditions, giving rise to interesting arctic curves (see \cite{EyKo} for a treatment using matrix models). It would be interesting to apply the Tangent Method 
to this situation.

\bibliographystyle{amsalpha} 

\bibliography{qAztec}

\end{document}